\pdfoutput=1

\documentclass[nolineno, thm-restate]{socg-lipics-v2021}
\hideLIPIcs

\usepackage{amsthm}
\usepackage{mathtools}
\usepackage{caption}
\usepackage[linesnumbered,ruled,vlined]{algorithm2e}
\usepackage{xspace}
\usepackage{multirow}
\usepackage{tikz}
\usetikzlibrary{calc,angles,decorations.pathreplacing,arrows,arrows.meta}

\usepackage[T1]{fontenc} % arxiv
\usepackage{lmodern} % arxiv

\usepackage{todonotes}

\usetikzlibrary{patterns.meta}
\tikzdeclarepattern{
  name=mylines,
  parameters={
      \pgfkeysvalueof{/pgf/pattern keys/size},
      \pgfkeysvalueof{/pgf/pattern keys/angle},
      \pgfkeysvalueof{/pgf/pattern keys/line width},
  },
  bounding box={
    (0,-0.5*\pgfkeysvalueof{/pgf/pattern keys/line width}) and
    (\pgfkeysvalueof{/pgf/pattern keys/size},
     0.5*\pgfkeysvalueof{/pgf/pattern keys/line width})},
  tile size={(\pgfkeysvalueof{/pgf/pattern keys/size},
              \pgfkeysvalueof{/pgf/pattern keys/size})},
  tile transformation={rotate=\pgfkeysvalueof{/pgf/pattern keys/angle}},
  defaults={
    size/.initial=5pt,
    angle/.initial=45,
    line width/.initial=.4pt,
  },
  code={
      \draw [line width=\pgfkeysvalueof{/pgf/pattern keys/line width}]
        (0,0) -- (\pgfkeysvalueof{/pgf/pattern keys/size},0);
  },
}

\usepgflibrary{patterns} % LaTeX and plain TeX and pure pgf
\usepgflibrary[patterns] % ConTeXt and pure pgf
\usetikzlibrary{patterns} % LaTeX and plain TeX when using TikZ
\usetikzlibrary[patterns] % ConTeXt when using TikZ

\definecolor{guard}{HTML}{2527F2}

\definecolor{feasible region}{HTML}{1A8F2D} %22BB3B

\definecolor{notvisible}{HTML}{E9002D}
\definecolor{SkyBlue}{rgb}{0.53, 0.81, 0.92}
\definecolor{visible triangle}{HTML}{EEEEEE}

\definecolor{Qi}{HTML}{FC4E2A}
\definecolor{Qi outer}{HTML}{F08F6E}
\definecolor{Qi inner}{HTML}{FABF7B}

\definecolor{shadowcolor}{HTML}{009ADE}

\definecolor{darkgray}{HTML}{616161}

\definecolor{extracolor}{HTML}{AF58BA}

\tikzset{
    qi/.style={color=Qi outer, pattern={mylines[size=2pt, line width=.6pt, angle=45]}, pattern color=Qi inner},
    F/.style={thick, color=feasible region, pattern={mylines[size=4pt, line width=1.2pt, angle=-45]}, pattern color=feasible region!30},
    Fbare/.style={color=white, pattern={mylines[size=4pt, line width=1.2pt, angle=-45]}, pattern color=feasible region!30},
    nv/.style={thick, color = notvisible, pattern = dots, pattern color = notvisible!70},
    shadow/.style={color = white, pattern = horizontal lines, pattern color = shadowcolor!40},
    shadow2/.style={color = white, pattern = vertical lines, pattern color = shadowcolor!20},
    extra/.style={color=white, pattern= crosshatch, pattern color=extracolor!35},
}

\DeclareFontShape{T1}{lmr}{m}{scit}{<->ssub * lmr/m/scsl}{}%
{ \sffamily%
  \DeclareFontShape{T1}{lmss}{m}{sc}{<->ssub * lmr/m/sc}{}%
  \DeclareFontShape{T1}{lmss}{bx}{sc}{<->ssub * cmr/bx/sc}{}%
}

\theoremstyle{claimstyle}
\newtheorem{case}{Case}
\newtheorem{subcase}{Subcase}[case]
\newtheorem{subsubcase}{Subsubcase}[subcase]
\makeatletter\@addtoreset{case}{theorem}\makeatother
\makeatletter\@addtoreset{subcase}{case}\makeatother
\makeatletter\@addtoreset{subsubcase}{subcase}\makeatother
\theoremstyle{plain}

\newenvironment{proofsketch}{%
    \renewcommand{\proofname}{Proof Sketch}\proof}{\endproof}

\usepackage[ruled,vlined]{algorithm2e}
\LinesNumbered

\newcommand{\FeasibleRegion}[1]{F\!\lrParents{#1}}

\newcommand{\I}[2]{\lrBrackets{#1, #2}}
\newcommand{\Fc}{\mathcal{V}\xspace}
\DeclareMathOperator{\Gsc}{\textsc{G}}
\newcommand{\GscNEW}[1]{\textsc{G}\!\lrParents{#1}}

\newcommand{\visPol}[2]{V\!P\hspace{-0.2 em}\lrParents{#1, #2}}
\newcommand{\lastVisPoint}[3]{lastVisiblePoint\hspace{-0.2 em}\lrParents{#1, #2, #3}}
\newcommand{\tri}[1]{\triangle #1}

\newcommand{\Oh}[1]{\mathcal{O}\!\lrParents{#1}}
\newcommand{\OhMega}[1]{\Omega\!\lrParents{#1}}
\newcommand{\OhTheta}[1]{\Theta\!\lrParents{#1}}
\newcommand{\lrBrackets}[1]{\left[ #1 \right]}
\newcommand{\lrCurlyBrackets}[1]{\left\{ #1 \right\}}
\newcommand{\lrParents}[1]{\left( #1 \right)}
\newcommand{\lrSize}[1]{\left| #1 \right|}

\newcommand{\myMax}[2]{\max\!\lrCurlyBrackets{#1 \,|\, #2}}

\newcommand{\linesegment}[2]{\overline{#1\,#2}}
\newcommand{\lineextension}[2]{L(#1,#2)}
\newcommand{\ray}[2]{\overrightarrow{#1\,#2}}

\usepackage{amsmath}
\usepackage{amssymb}
\newcommand{\R}{\mathbb{R}}

\newcommand{\cagp}{contiguous art gallery problem\xspace}

\newcommand{\CAG}{Contiguous Art Gallery\xspace}
\newcommand{\opt}{{k^*}\xspace}

\newcommand{\Function}[1]{\textup{\textsc{#1}}\xspace}
\newcommand{\Greedy}{\Function{GreedyInterval}}
\newcommand{\RepeatedGreedy}{Algorithm~\ref{alg:greedy}\xspace}

\newcommand{\Cpp}{C\nolinebreak\hspace{-.05em}\raisebox{.4ex}{\tiny\bf +}\nolinebreak\hspace{-.10em}\raisebox{.4ex}{\tiny\bf +}\xspace}

\funding{Supported by Independent Research Fund Denmark (DFF), grants~9131-00113B and 10.46540/3103-00334B}


\title{The \CAG Problem is Solvable in Polynomial Time}
\authorrunning{Merrild, Rysgaard, Schou and Svenning}

\author{Magnus Christian Ring Merrild}{Department of Computer Science, Aarhus University, Denmark}{merrild@cs.au.dk}{https://orcid.org/0009-0004-7272-7839}{}

\author{Casper Moldrup Rysgaard}{Department of Computer Science, Aarhus University, Denmark}{rysgaard@cs.au.dk}{https://orcid.org/0000-0002-3989-123X}{}

\author{Jens Kristian Refsgaard Schou}{Department of Computer Science, Aarhus University, Denmark }{jkrs@cs.au.dk}{https://orcid.org/0000-0002-9915-7061}{}

\author{Rolf Svenning}{Department of Computer Science, Aarhus University, Denmark}{rolfsvenning@cs.au.dk}{https://orcid.org/0000-0002-9903-4651}{}

\acknowledgements{We thank Joseph O'Rourke for organizing the open problem session~\cite{CCCG2024_open} at the Canadian Conference on Computational Geometry 2024 (CCCG24) at Brock University, Canada, where Thomas C. Shermer proposed the \cagp. The authors also thank the participants of CCCG24 for their lively discussions of the problem at the conference, especially Frederick Stock. This paper was submitted to the Symposium on Computational Geometry 2025 conference and accepted as a merge with \cite{robson2024analyticarccoverproblem} and \cite{biniaz2024contiguousboundaryguarding} into \cite{cagp_merge}. We thank our collaborators for the engaging discussions that inspired optimizations to our work.}

\begin{document}

\maketitle

\begin{abstract}
    In this paper, we study the \emph{Contiguous Art Gallery Problem}, introduced by Thomas C. Shermer at the 2024 Canadian Conference on Computational Geometry, a variant of the classical art gallery problem from 1973 by Victor Klee. 
    In the contiguous variant, the input is a simple polygon $P$, and the goal is to partition the boundary into a minimum number of polygonal chains such that each chain is visible to a guard. 
    We present a polynomial-time RAM algorithm, which solves the contiguous art gallery problem. 
    Our algorithm is simple and practical, and we make a \Cpp implementation available.
    
    In contrast, many variations of the art gallery problem are at least \textsc{NP}-hard, making the contiguous variant stand out. 
    These include the classical art gallery problem and the \emph{edge-covering problem}, both of which being proven to be $\exists\R$-complete recently by Abrahamsen, Adamaszek, and Miltzow [J. ACM 2022] and Stade [SoCG 2025], respectively.
    Our algorithm is a greedy algorithm that repeatedly traverses the polygon's boundary.
    To find an optimal solution, we show that it is sufficient to traverse the polygon polynomially many times, resulting in a runtime of $\mathcal{O}\!\left( n^6 \log n \right)$ arithmetic operations. We further bound the bit complexity of the computed values, showing that problem is in \textsc{P}.
    Additionally, we provide algorithms for the restricted settings, where either the endpoints of the polygonal chains or the guards must coincide with the vertices of the polygon.
\end{abstract}

% --------------------------------------------------------------------
\section{Introduction}
% --------------------------------------------------------------------

The \emph{art gallery problem}, introduced by Victor Klee in 1976, is a classical computational geometry problem where the goal is to find a minimum set of guards (points) in the interior of an input polygon $P$ that sees every other point in $P$ of the polygon. 
There are numerous variations of this problem, many of which are at least \textsc{NP}-hard or require complicated algorithms with a high polynomial running time.
This is particularly the case for unrestricted variants, where guard locations and the part of the polygon they cover are not constrained to the vertices of the input polygon.
In this paper, we study the \emph{contiguous art gallery problem} where the boundary of $P$ should be partitioned into a minimum number of contiguous intervals, i.e., polygonal chains, such that each chain is visible to a guard in the interior of $P$.
Neither the guards nor the endpoints of the chains are restricted to the vertices of $P$.
The problem was introduced by Thomas C. Shermer at the Canadian Conference on Computational Geometry 2024.
We resolve it by providing a polynomial-time real RAM algorithm.

\begin{restatable}{theorem}{thmCAGIsInP}
\label{thm:CAG_is_in_P}
    The \cagp for a simple polygon with $n$ vertices is solvable in $\Oh{\opt n^5\log n}$ arithmetic operations, where $\opt$ is the size of an optimal solution.
\end{restatable}

In section~\ref{sec:bit_complexity} we show that the bit complexity of the arithmetic operations are bounded polynomially in the number of bits used to encode the input polygon, therefore showing the aforementioned algorithm is in fact a polynomial-time RAM algorithm. 

\begin{restatable}{theorem}{thmCAGIsTruelyInP}
\label{thm:CAG_is_truely_in_P}
    The \cagp for a simple polygon encoded in $N$ bits is a member of the complexity class \textsc{P}.
\end{restatable}

In contrast to many other art gallery variants, it is a surprising and positive result that this variant allows for an efficient and simple algorithm.
We demonstrate this by implementing the algorithm in \Cpp using CGAL~\cite{fabri2000design, cgal:cgal_intersections_reg_boolean_fwzh-rbso2-24b, cgal_vis_pol_hhb-visibility-2-24b, cgal:eb-24b, cgal:arrangements_wfzh-a2-24b}.
Our code is available online~\cite{github_impl}.

% --------------------------------------------------------------------
\subsection{Related work}
% --------------------------------------------------------------------

\begin{table}[t]
    \centering

    \caption{Summary of the work related to the \cagp, for polygons with $n$ vertices where $\opt$ is the size of the minimal decomposition. 
    We abbreviate Art Gallery as AG, $\dagger$~refers to the guard location being restricted to vertices, and $\ddagger$~refers to the vertices of the pieces (the guarded area) being restricted to the vertices of the input polygon. 
    Notice how the difficulty of the four problems increases along the diagonal from the bottom left vertex to the top right.}
    \label{tab:table_of_related_work}
    \renewcommand*{\arraystretch}{1.3}
    \begin{tabular}{l|c|c|c}
        Problem & Vertex restricted & Unrestricted & With holes \\ \hline \hline
        Standard AG & \textsc{NP}-hard$^\dagger$~\cite{vertex_restricted_AG_nphard} & $\exists\R$-complete~\cite{AG_is_ER_complete} & $\exists\R$-complete~\cite{AG_is_ER_complete} \\\hline
        Edge-covering AG & \textsc{NP}-hard$^\dagger$~\cite{pidgeon_figure} & $\exists\R$-complete~\cite{stade2025pointboundaryartgalleryproblem} & $\exists\R$-complete~\cite{stade2025pointboundaryartgalleryproblem} % \textsc{NP}-hard~\cite{pidgeon_figure} 
        \\\hline
        Star-shaped partition & $\Oh{n^7\log n}^\ddagger$~\cite{doi:10.1137/0214056} & $\Oh{n^{105}}$~\cite{starshaped_decomposition} & \textsc{NP}-hard~\cite{joseph_survey} \\\hline
        \multirow{2}{*}{Contiguous AG} & $\Oh{n^2 \log^2 n}^\ddagger$ \emph{new} & \multirow{2}{*}{$\Oh{\opt n^5\log n}$ \emph{new}} & \multirow{2}{*}{\emph{Open}} \\
        & $\Oh{n^2\log n}^\dagger$ \emph{new} &  &
    \end{tabular}
\end{table}
The literature contains a variety of variations of the classical art gallery formulation~\cite{joseph_survey, Tom_survey, urrutia_survey}, and we summarize the most important ones in Table~\ref{tab:table_of_related_work}.
Most art gallery variants can be framed as \emph{decomposition problems}~\cite{decomposition_1985_KEIL1985197, decomposition_2000_MARKKEIL2000491}, where the goal is to decompose an input polygon $P$ into less complicated components whose union is $P$. 
If the pieces may overlap, it is \emph{covering problem}, and if not it is a \emph{partition problem}. 
In the art gallery setting each piece must be guarded.
The variants can further be categorized as \emph{restricted} or \emph{unrestricted}.
In a restricted version of the problem, the guards and/or the vertices of each piece must coincide with the vertices of $P$.
Conversely, an unrestricted version imposes no such constraints on the placement of guards or the vertices of the pieces.
Variants also differ based on the complexity of $P$.
We consider the two important cases of simple polygons with and without holes, with the latter typically being significantly more complicated.
In Table~\ref{tab:table_of_related_work} summarizes the difficulty of three fundamental art gallery variants and our contiguous variant.

\textit{The classical art gallery problem} can be viewed as a covering problem in which the input polygon is decomposed into a minimal number of \emph{star-shaped} polygons. 
A polygon is star-shaped exactly when it is possible to place a guard that sees all other points in the polygon. 
The decision problem \emph{``Can this polygon be guarded by $k$ guards?''} was recently shown by Abrahamsen, Adamaszek, and Miltzow~\cite{AG_is_ER_complete} to be $\exists\R$-complete~\cite{ExistR_introduction} for any simple polygon with or without holes. If the guards are restricted to the vertices of the polygon, then Lee and Lin~\cite{AG_is_NP_hard} showed that the problem is \textsc{NP}-hard. 

\textit{The edge-covering problem} is a variant of the art gallery problem, where the aim is to only cover the edges of a polygon, motivated by protecting valuable art on the walls of the gallery. 
The edge-covering variant was shown by Laurentini to be \textsc{NP}-hard~\cite{pidgeon_figure} for guards restricted to vertices and $\exists\R-$hard for unrestricted guards and polygons with holes.
As seen in Figure~\ref{fig:cag_not_star_shaped}, covering the edges does not imply covering the interior of $P$.

\textit{The minimum star-shaped partition problem} was shown by Keil~\cite{doi:10.1137/0214056} to be solvable in $\Oh{n^7\log n}$ arithmetic operations, when the star-shaped regions start and end at vertices. Without vertex restriction, Abrahamsen, Blikstad, Nusse, and Zhang~\cite{starshaped_decomposition} introduced a breakthrough algorithm solving this harder variant using $\Oh{n^{105}}$ arithmetic operations. With holes, computing the minimum star-shaped partition is \textsc{NP}-hard due to O'Rourke~\cite{joseph_survey}.

\textit{The \cagp} was introduced at the 2024 Canadian Conference on Computational Geometry, by Thomas C. Shermer and this variant is the study of this paper. 
Framed as a decomposition problem, the goal is to partition the boundary $\partial P$ into a minimal number of polygonal chains such that each can be seen by a guard interior to $P$.
We mainly focus on the unrestricted version for a simple polygon without holes, but we also describe algorithms for restricted versions. \RepeatedGreedy does not generalize to polygons with holes, see Appendix~\ref{apx:polygons_with_holes}, and determining the hardness of this variant is an interesting open problem.
The placement of guards in optimal solutions may differ across all variants, as demonstrated in Figures~\ref{fig:cag_not_star_shaped}~and~\ref{fig:AG_neq_CAG}.

\begin{figure}[ht]
    \centering
    
    \def\coords{
        \coordinate (A) at (0,10);
        \coordinate (B) at (12,10);
        \coordinate (C) at (12,0);
        \coordinate (D) at (0,0);
        \coordinate (E) at (2.24,9.2);
        \coordinate (F) at (2,7.5);
        \coordinate (G) at (9.76,9.2);
        \coordinate (H) at (2.24, 0.8);
        \coordinate (I) at (9.76, 0.8);
        \coordinate (J) at (10,7.5);
        \coordinate (K) at (10,2.5);
        \coordinate (L) at (2,2.5);
        \coordinate (M) at (14,8.34); 
        \coordinate (N) at (-2, 8.34); 
        \coordinate (Q) at (14,1.67); 
        \coordinate (R) at (-2, 1.67); 
        \coordinate (O) at (14,5);
        \coordinate (P) at (-2,5);
        \coordinate (X) at (6,10);
        \coordinate (Y) at (6,0);
    }
    
    \begin{tikzpicture}[scale = 0.3]
        \begin{scope}
            \coords
            \draw[color=red] (J) -- (M) -- (Q) -- (K) -- (I);
            \draw[color=blue] (L) -- (R) -- (N) -- (F) -- (E);
            \draw[color=teal] (J) -- (G) -- (B) -- (A) -- (E) -- (F);
            \draw[color=green] (I) -- (C) -- (D) -- (H) -- (L);
            \fill[color=blue] (L) circle[radius=0.2];
            \fill[color=red] (J) circle[radius=0.2];
            \fill[color=teal] (E) circle[radius=0.2];
            \fill[color=green] (I) circle[radius=0.2];
        \end{scope}

        \begin{scope}[xshift=22cm]
            \coords
            \draw[color=red] (X) -- (B) -- (G) -- (J) -- (M) -- (Q) -- (K) -- (I) -- (C) -- (Y);
            \draw[color=blue] (Y) -- (D) -- (H) -- (L) -- (R) -- (N) -- (F) -- (E) -- (A) -- (X);
            \fill[color=red] (P) circle[radius=0.2];
            \fill[color=blue] (O) circle[radius=0.2];
            \draw[dashed, color = red] (P) -- (X);
            \draw[dashed, color = red] (P) -- (G);
            \draw[dashed, color = red] (P) -- (J);
            \draw[dashed, color = red] (P) -- (K);
            \draw[dashed, color = red] (P) -- (I);
            \draw[dashed, color = red] (P) -- (Y);

        \end{scope}
    \end{tikzpicture}
    
    \caption{Left, an optimal solution to the vertex restricted guards contiguous art gallery problem, discovered by greedily maximizing in one direction. 
    Right, an optimal unrestricted contiguous art gallery solution requires 2 guards whose guarded pieces start and end at non-vertex points on the boundary, and these boundary points are not directly defined by segments of the polygon.}
    \label{fig:2_guard_incidence_points}
\end{figure}
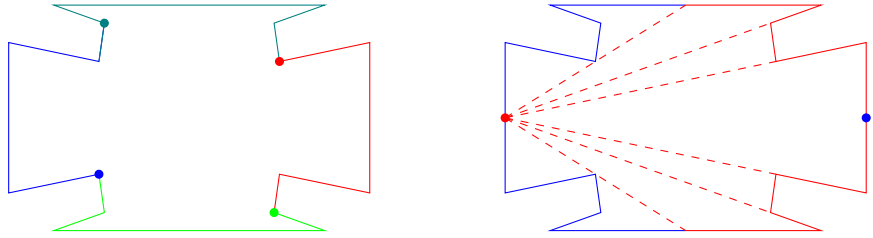

\begin{figure}[ht]
    \centering

    \begin{tikzpicture}
        \draw (0, 0) -- (-5, 0) -- (-5, 3) -- (-4, 3) -- (-4, 1) -- (-1, 1) -- (-1, 3) -- (0, 3) -- cycle;
        \fill[color=blue] (-1.5, 0) circle[radius=0.07] node[below right] {\textcolor{blue}{$g_r$}};
        \fill[color=red] (-3.5, 0) circle[radius=0.07]  node[below left] {\textcolor{red}{$g_\ell$}};;
        \draw[dashed, color=blue] (0, 3) -- (-1.5, 0) -- (-5, 1.4);
        \draw[dashed, color=red] (0, 1.4) -- (-3.5, 0) -- (-5, 3);
        \draw [decorate,decoration={brace,amplitude=5pt,raise=3pt}]
                (0, 3) -- node[xshift=15pt] {$r$} (0, 1.4);
        \draw [decorate,decoration={brace,amplitude=5pt,raise=3pt}]
                (-5, 1.4) -- node[xshift=-15pt] {$\ell$} (-5, 3);
        \draw [decorate,decoration={brace,amplitude=5pt,raise=5pt}]
                (-1.5, 0) -- node[yshift=-15pt] {$s$} (-3.5, 0);
    \end{tikzpicture}
    
    \caption{A polygon demonstrating where the blue guard $g_r$ sees all of $r$ and none of $\ell$, and the red guard $g_\ell$ is symmetric. 
    Any guard placed along $s$ will yield a trade-off between how much of $\ell$ and $r$ it sees.
    This shows that it is non-trivial with unrestricted guards to find a finite set of polygonal chains that include an optimal solution as there may be infinite maximal intervals.}
    \label{fig:CAG_tradeoff_infinite_different}
\end{figure}
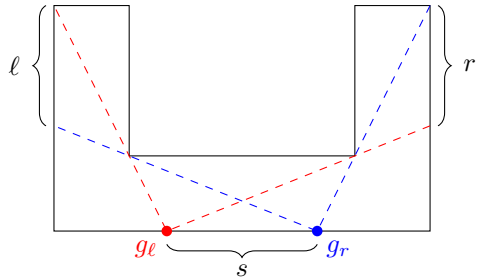

\begin{figure}[ht]
    \centering
    
    \def\coords{
        \coordinate (h00) at (0.8, 0.3);
        \coordinate (h01) at (1.8, -0.2);
        \coordinate (h02) at (2.5, 0.0);
        \coordinate (h03) at (1.7, 1.6);
        \coordinate (h04) at (2.0, 0.5);
        \coordinate (h10) at (-0.6598076211353314, 0.5428203230275511);
        \coordinate (h11) at (-0.7267949192431118, 1.6588457268119896);
        \coordinate (h12) at (-1.2499999999999996, 2.165063509461097);
        \coordinate (h13) at (-2.235640646055102, 0.672243186433546);
        \coordinate (h14) at (-1.433012701892219, 1.4820508075688776);
        \coordinate (h20) at (-0.14019237886466884, -0.8428203230275508);
        \coordinate (h21) at (-1.0732050807568885, -1.458845726811989);
        \coordinate (h22) at (-1.250000000000001, -2.165063509461096);
        \coordinate (h23) at (0.5356406460551006, -2.272243186433546);
        \coordinate (h24) at (-0.5669872981077817, -1.982050807568877);
        \coordinate (b00) at (0.9301208850715019, 0.5172240161046361);
        \coordinate (b01) at (1.345427755869507, 1.210537978279261);
        \coordinate (b02) at (0.6333609120801347, 1.2220015638506319);
        \coordinate (b03) at (0.37268901556668577, 0.5261981285661453);
        \coordinate (b10) at (-0.9129895799297773, 0.546896307010069);
        \coordinate (b11) at (-1.721070519370448, 0.5599056264000507);
        \coordinate (b12) at (-1.3749648537990262, -0.06249414229983674);
        \coordinate (b13) at (-0.6420454545454548, 0.05965909090909144);
        \coordinate (b20) at (-0.017131305141725006, -1.064120323114705);
        \coordinate (b21) at (0.3756427635009406, -1.7704436046793115);
        \coordinate (b22) at (0.7416039417188911, -1.1595074215507952);
        \coordinate (b23) at (0.26935643897876876, -0.5858572194752367);
        \coordinate (t0) at (-0.307900130957115, 0.1153499781738141);
        \coordinate (t1) at (0.054054054054053946, -0.3243243243243241);
        \coordinate (t2) at (0.25384607690306116, 0.2089743461505098);
    }
    
    \begin{tikzpicture}
        \begin{scope}
            \coords
            \draw[red] (h00) -- (h01) -- (h02) -- (h03) -- (h04) -- (b00);
            \draw[teal] (b00) -- (b01) -- (b02) -- (b03);
            \draw[red] (b03) -- (h10);
            \draw[blue] (h10) -- (h11) -- (h12) -- (h13) -- (h14) -- (b10);
            \draw[red] (b10) -- (b11) -- (b12) -- (b13);
            \draw[blue] (b13) -- (h20);
            \draw[teal] (h20) -- (h21) -- (h22) -- (h23) -- (h24) -- (b20);
            \draw[blue] (b20) -- (b21) -- (b22) -- (b23);
            \draw[teal] (b23) -- (h00);
            \fill[color=red] (h04) circle[radius=0.07];
            \fill[color=blue] (h14) circle[radius=0.07];
            \fill[color=teal] (h24) circle[radius=0.07];
            \draw[dashed, red] (h04) -- (t0);
            \draw[dashed, blue] (h14) -- (t1);
            \draw[dashed, teal] (h24) -- (t2);
            \fill[gray, opacity=0.5] (t0) -- (t1) -- (t2) -- cycle;
        \end{scope}
        
        \begin{scope}[xshift=7cm]
            \coords
            \draw[red]  (h00) -- (h01) -- (h02) -- (h03) -- (h04) -- (b00);
            \draw[blue] (h10) -- (h11) -- (h12) -- (h13) -- (h14) -- (b10);
            \draw[teal] (h20) -- (h21) -- (h22) -- (h23) -- (h24) -- (b20);
            \draw[green]  (b00) -- (b01) -- (b02) -- (b03) -- (h10);
            \draw[orange] (b10) -- (b11) -- (b12) -- (b13) -- (h20);
            \draw[cyan]   (b20) -- (b21) -- (b22) -- (b23) -- (h00);
            \fill[color=red]    ($(h04) + (0.1,-0.3)$) circle[radius=0.07];
            \fill[color=blue]   ($(h14) + (0,0.1)$) circle[radius=0.07];
            \fill[color=teal]   ($(h24) - (0.2,0)$) circle[radius=0.07];
            \fill[color=green]  (b03) circle[radius=0.07];
            \fill[color=orange] ($(b13) + (0,0.2)$) circle[radius=0.07];
            \fill[color=cyan]   ($(b23) - (0.15,0)$) circle[radius=0.07];
        \end{scope}
    \end{tikzpicture}

    \caption{Left, 3 optimal edge-covering guards where the inner gray triangle is not visible. 
    Right, 6 guards that optimally contiguously cover the boundary.
    Guards are visualized as points whose color matches the part of the boundary they guard.}
    \label{fig:cag_not_star_shaped}
\end{figure}
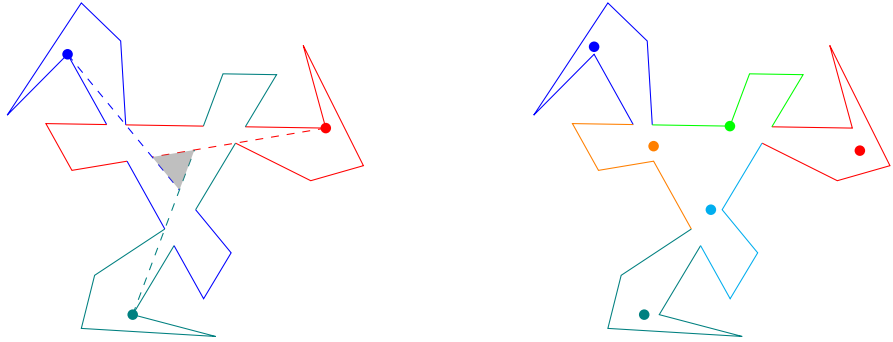

\begin{figure}[ht]
    \centering

    \begin{tikzpicture}
        \coordinate (G1) at (-1.5, 0);
        \coordinate (G2) at (0, 0);
        \coordinate (G3) at (1.5, 0);
        
        \foreach \h/\i in {l/-1, r/1} {
        \foreach \v/\j in {b/-1, t/1} {
            \coordinate (A\h\v) at (1*\i, 1*\j);
            \coordinate (B\h\v) at (0.8*\i, 1.4*\j);
            \coordinate (C\h\v) at (1.4*\i, 1.4*\j);
            \coordinate (D\h\v) at (2*\i, 1*\j);
            \coordinate (E\h\v) at (2*\i, 0.8*\j);
            \coordinate (F\h\v) at (2.4*\i, 0.95*\j);
            \coordinate (G\h\v) at (2.4*\i, 0.5*\j);
            \coordinate (H\h\v) at (2*\i, 0.4*\j);
            \draw (A\h\v) -- (B\h\v) -- (C\h\v) -- (D\h\v) -- (E\h\v) -- (F\h\v) -- (G\h\v) -- (H\h\v);
            
            \draw[dashed] (G2) -- (E\h\v);
            \draw[dashed] (G2) -- (H\h\v);

            \coordinate (Z\h\v) at (1.25*\i, 0.5*\j);
        }}
        \draw (Art) -- (Alt);
        \draw (Arb) -- (Alb);
        \draw (Hrt) -- (Hrb);
        \draw (Hlt) -- (Hlb);
        
        \draw[dashed] (G1) -- (Blt);
        \draw[dashed] (G1) -- (Blb);
        \draw[dashed] (G3) -- (Brt);
        \draw[dashed] (G3) -- (Brb);
        
        \foreach \g in {1, 2, 3}
            \node[draw, fill=blue, shape=rectangle, minimum size=0.14cm, inner sep=0pt] at (G\g) {};
        \foreach \g in {lt, lb, rt, rb}
            \fill[color=red] (Z\g) circle[radius=0.07];
    \end{tikzpicture}
    
    \caption{A polygon, where the red guards (circles) represent an optimal solution to the contiguous art gallery problem and the blue guards (squares) are an optimal solution to the classical art gallery problem.}
    \label{fig:AG_neq_CAG}
\end{figure}
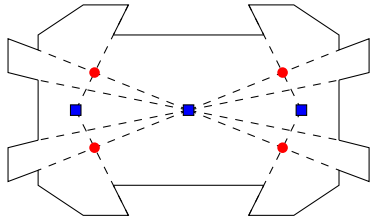

% --------------------------------------------------------------------
\subsection{Limitations of Existing Approaches}
\label{sec:limitations_of_existing_approaches}
% --------------------------------------------------------------------

We first sketch why the restricted version of the \cagp is relatively simpler to solve in polynomial time than the unrestricted variant, which is our focus.
When partitioning the boundary of a simple polygon, it is natural to view it as a circle and the polygonal chains as arches or intervals; see Figure~\ref{fig:intervals_on_boundaries}. 
By doing so, finding a minimal partition is similar to finding a \textit{minimal circle-cover} among a set of intervals $C$.
This problem was studied by Lee and Lee~\cite{circle_cover_arc_minimization_1984_LEE1984109} for finite cardinality $C$, giving an algorithm that runs in $\Oh{\lrSize{C}\log{\lrSize{C}}}$ time.
Thus, given a finite set $C$ of polygon chains of the boundary of $P$ such that an optimal solution is contained in $C$, then the \cagp can be solved in $\Oh{\lrSize{C}\log{\lrSize{C}}}$ time by simply viewing each polygonal chain as an arc and running the minimal circle-covering algorithm.
In the restricted \cagp, that is, restricting guard positions to vertices of $P$ or restricting the endpoints of the polygonal chains to coincide with vertices of $P$, leads to a set of polygonal chains $C$ that includes an optimal solution of size $\Oh{n^2}$ and $\Oh{n}$, respectively.
The set $C$ must also be computed, and in Appendix~\ref{apx:vertex_restricted_variants} we describe how to solve these two problems in $\Oh{n^2 \log n}$ and $\Oh{n^2 \log^2 n}$ time, respectively.

The unrestricted \cagp is our main focus, where guards may be placed anywhere interior to $P$, and polygonal chains of the boundary of $P$ may start and end anywhere.
In this setting, it is non-trivial to generate a polynomial-sized set of intervals that contains an optimal solution, which would allow the use of the algorithm for the circle-cover minimization problem. 
An approach that does not work is to generate all the different candidate intervals that are maximal, as there may be infinitely many of these, as shown in Figure~\ref{fig:CAG_tradeoff_infinite_different}.
Another approach that we were unable to rule out is based on showing that an optimal solution coincides with a point of low $\emph{degree}$~\cite{starshaped_decomposition} either with a guard or the endpoint of a polygonal chain.
Points of degree $0$ are the vertices of $P$, and points of degree $i > 0$ are points formed by intersections of lines formed by pairs of points of degree $(i-1)$.
Letting $D_i$ be the number of points of degree $i$, then $D$ grows as $\OhTheta{n^{4^i}}$ so even if this approach is feasible, it leads to high polynomial running time even for low degree points.
For the minimum star-shaped partition problem the authors showed that first-degree points suffice as candidates~\cite{starshaped_decomposition}.

% --------------------------------------------------------------------
\subsection{Our Greedy Optimal Solution}
% --------------------------------------------------------------------

Our solution takes a different and remarkably simple approach.
At the core of our method is the \Greedy algorithm, which, given a point $x$ on the polygon's boundary, finds the furthest point $y$ along the boundary in the clockwise direction such that the polygonal chain from $x$ to $y$ can be guarded by an interior point of the polygon.
We denote these as \emph{greedy intervals}.
Repeatedly taking these greedy steps until the boundary has been covered gives a solution of size $k \leq \opt + 1$. 
Each traversal of the boundary we denote by a \emph{revolution}, and continuing this process performing $\Oh{k n^3}$ revolutions ultimately yields an optimal solution. 
Performing a single revolution by repeatedly using the \Greedy algorithm takes $\Oh{n^2 \log n}$ arithmetic operations, leading to $\Oh{\opt n^5\log n}$ total arithmetic operations.
The model of computation used here is the real RAM model~\cite{PR_complexity_class} where arithmetic operations on real numbers take unit time.
Here, the size of the input $n$ is the number of real-valued inputs, that is, the vertices of a circular list of vertices of $P$.
The main result of our paper is Theorem~\ref{thm:CAG_is_in_P} that the \cagp is solvable in polynomial time in the real RAM model. 

We also consider the problem in the RAM model, where the polygon is represented as an array of points, each point having two coordinates consisting of fractions with integer numerator and denominator. We consider the case where $N$ bits are used to describe the entire input and show that $poly(N)$ operations is sufficient to find an optimal guarding. This is the contents of Section~\ref{sec:bit_complexity}.

When the input polygon has holes our approach does not work and it is an intriguing open problem to determine whether this variant of the \cagp is solvable in polynomial time.

% --------------------------------------------------------------------
\subsection{Concurrent work}
% --------------------------------------------------------------------
This work is concurrent with \cite{robson2024analyticarccoverproblem} and
\cite{biniaz2024contiguousboundaryguarding} that both solve the \cagp using different techniques, the merge of this paper \cite{cagp_merge} was published at SoCG25. A notable difference is that this paper traverses the polygon clockwise, while the three others are counter-clockwise.

In \cite{robson2024analyticarccoverproblem} it is explored how \Greedy can be expressed as a piecewise rational linear function and show that repeatedly composing \Greedy with itself leads to a piecewise rational linear function representing \emph{all} optimal solutions. 

In \cite{biniaz2024contiguousboundaryguarding} it is explored how to generate $\Oh{n^4}$ candidate solutions by starting \RepeatedGreedy from a carefully selected set of points, such that at least one of these will be optimal.

% --------------------------------------------------------------------
\subsection{Organization}
% --------------------------------------------------------------------

In Section~\ref{sec:greedy_algo} we present our algorithm for solving the \cagp, and in Section~\ref{sec:alg_greedy_interval} we cover details of how our algorithm can be implemented using basic computational geometry operations. 
In Sections~\ref{sec:the_greedy_algorithm_a_combinatorial_viewpoint}~and~\ref{sec:the_greedy_algorithm_geometry}, we derive an upper bound on the number of iterations of our greedy algorithm based on the geometric and combinatorial properties of our algorithm. 
In Section~\ref{sec:bit_complexity} we then use this bound to prove that the \cagp is in the complexity class \textsc{P}.
Finally, in Section~\ref{sec:open_problems} we state related open problems.

In Appendix~\ref{apx:vertex_restricted_variants} we describe how to solve the \cagp under vertex restrictions. 
In Appendix~\ref{apx:proof_of_supporting_lemmas} we prove geometric supporting lemmas. 
In Appendix~\ref{apx:counter_examples} we give two examples that illustrate interesting behavior of \RepeatedGreedy.

% --------------------------------------------------------------------
\section{The Repeated Greedy Algorithm}
\label{sec:greedy_algo}
% --------------------------------------------------------------------

Formally, let $P$ be a simple polygon with $n$ vertices $v_0,\dots,v_{n-1}$ and edges $e_i$ connecting $v_i$ and $v_{i+1}$.
The boundary of $P$ we denote as $\partial P$. 
The goal of the \cagp is to find the minimum $\opt$ such that one can partition the boundary $\partial P$ into $\opt$ contiguous \emph{visible} polygonal chains $I_1,I_2,I_3,\dots,I_\opt$ whose union is $\partial P$. 
A chain $I_i$ is visible if there exists a guard $g_i$ such that, for all points $x$ on $I_i$, the line segment $\linesegment{x}{g_i}$ is contained in $P$.
Conceptually, each $I_i$ corresponds to an interval of the boundary between two points $a$ and $b$, denoted as $\lrBrackets{a, b}$.
We allow guards to be collocated and consider them non-blocking, i.e., they can see through each other.
Furthermore, for convenience, we always index the edges and vertices of $P$ modulo $n$, meaning that $e_n = e_0$ and $v_n = v_0$. 

\begin{algorithm}[ht]
    \caption{\label{alg:greedy}}
    \KwIn{A simple polygon $P$ with vertices $v_0, v_1, v_2, \dots, v_n$}
    \KwOut{An explicit solution to the \cagp on $P$}
    $x_0 \longleftarrow v_0$\\
    \nllabel{line:start} \For{$i=1$ \KwTo $T$} {
        \nllabel{line:loop_greedy_interval1} $x_i, g_i \longleftarrow GreedyInterval\!\lrParents{x_{i - 1}, P}$ \\
    }
    \nllabel{line:loop_greedy_interval2} $j \longleftarrow \myMax{j \leq T - 2}{x_j \in \left( x_{T - 1}, x_T \right] }$ \\
    \nllabel{line:return_last_segments1} 
    \Return $\lrCurlyBrackets{\lrParents{x_{i - 1}, x_i, g_i} |\; j < i \leq T}$ \tcp*[h]{The last segments/guards covering $\partial P$} 
    \nllabel{line:return_last_segments2}
\end{algorithm}

To solve the \cagp, we propose Algorithm~\ref{alg:greedy} which is a con\-cep\-tu\-al\-ly simple greedy algorithm that starts at any point $x_0$ in $\partial P$ (Line~\ref{line:start}). 
Then repeatedly finds \emph{greedy} intervals, that is, the longest visible interval of $\partial P$ from a given starting point of $\partial P$, starting from where the previous segment ended (Lines~\ref{line:loop_greedy_interval1}~-~\ref{line:loop_greedy_interval2}).
The greedy interval from any point $x \in P$ can be found in $\Oh{poly(n)}$ time by combining basic computational geometry operations to calculate visibility polygons~\cite{visibility_of_polygon_from_an_edge_convex_viewing_lemma_1981_1675729, visibility_polygon_1981_ELGINDY1981186}, intersections of polygons~\cite{Weiler_polygon_clipping_1977_10.1145/965141.563896} and common tangents between disjoint polygons~\cite{common_tangents}.
The \Greedy algorithm is described in details in Section~\ref{sec:alg_greedy_interval} as Algorithm~\ref{alg:greedy_interval}. 
After $T \geq c k^2n^3$ iterations, where $c$ is a sufficiently large constant and $k$ is the number of iterations in the first revolution, the algorithm returns the start and endpoint of the last greedy intervals that form a partition of the boundary and the corresponding guards that see the segments (Lines~\ref{line:return_last_segments1}~-~\ref{line:return_last_segments2}).

% The core part of this paper is to show that polynomially many iterations are sufficient to guarantee that \RepeatedGreedy finds an optimal solution.
% We conjecture that $T = \Oh{k}$ iterations are sufficient and mention it is an open problem to improve the dependency on $T$ in Section~\ref{sec:open_problems}.
% Our algorithm can be viewed as running the circle-cover minimization algorithm~\cite{circle_cover_arc_minimization_1984_LEE1984109} on the implicit and potentially infinite set of greedy intervals defined by the polygon.

% --------------------------------------------------------------------
\section{The \Greedy algorithm}
\label{sec:alg_greedy_interval}
% --------------------------------------------------------------------

In this section, we present and analyze the \Greedy algorithm for computing greedy intervals, i.e. the longest visible interval $I$ from a given point $x \in P$, along with a corresponding guard $g$ that sees $I$.
The algorithm uses only basic computational geometry operations. 
We assume for simplicity that $P$ is not star-shaped, as otherwise $I = \partial P$.
The idea is to extend $I$ from one vertex of $P$ to the next (clockwise around $\partial P$), starting from the edge containing $x$, while maintaining the feasible region of $I$. 
The feasible region of $I$ is the set of points of $P$ that can see all of $I$, i.e. the region of $P$ where a guard for $I$ can be placed.
This happens on lines~\ref{VisibilityPolygon}~-~\ref{VisibilityPolygon_and_intersection} where $\visPol{x}{P}$ is a subroutine to compute the \emph{visibility polygon}~\cite{visibility_polygon_1981_ELGINDY1981186} of point $x \in P$ and $\cap$ computes the intersection between two simple polygons~\cite{Weiler_polygon_clipping_1977_10.1145/965141.563896}.
The correctness of extending $I$ in discrete steps around $\partial P$ follows by contrapositive of Lemma~\ref{lem:convex_viewing_lemma}, which is that any point $u \in P$ that can see two points $a, b$ on some edge $e_j$, can see all points on $e_j$ between $a$ and $b$. 
Finally, we extend $I$ along edge $e_{i-1}$ towards the last vertex $v_i$ that causes the feasible region to be empty.
This happens on line~\ref{extending_F_on_last_edge} where $\lastVisPoint{P}{F_{i - 1}}{e_{i - 1}}$ is a subroutine for computing the point furthest along $e_{i-1}$ visible from $F_{i - 1}$. 
To do so simply and efficiently we prove Lemma~\ref{lem:extending_on_last_edge_vertices_of_F_enough} which shows that it suffices to consider the vertices of the last feasible region.

\begin{algorithm}[ht]
    \caption{\Greedy\label{alg:greedy_interval}}
    \KwIn{Point $x$ on edge $e_i = (v_i, v_{i + 1})$ of a simple and a non-star-shaped polygon $P$}
    \KwOut{The endpoint of the longest visible interval $I$ from $x$ and a guard $g$ that can see $I$}
    \nllabel{VisibilityPolygon} $F_i \longleftarrow \visPol{P}{x}$ \\
    \While{$F  \neq \emptyset$} {
        $i \longleftarrow i + 1$ \\
        \nllabel{VisibilityPolygon_and_intersection} $F_i \longleftarrow F_{i - 1} \cap \visPol{P}{v_i}$ \\
    }
    \nllabel{extending_F_on_last_edge} \Return $\lastVisPoint{P}{F_{i - 1}}{e_{i - 1}}$
\end{algorithm}

The main concern is to show that the running time of \Greedy is $\Oh{poly(n)}$. 
Theorem~\ref{thm:running_time_alg_greedy_segment} gives the precise output-sensitive running time, depending on the number $e$ of edges of $P$ that are intersected by $I$. 

\begin{restatable}{theorem}{thmGreedyIntervalRuntime}
\label{thm:running_time_alg_greedy_segment}
    The running time of the \Greedy algorithm for a polygon $P$ with $n$ vertices where the greedy interval that is output intersects $e$ edges of $P$ is $\Oh{e n \log n}$.
\end{restatable}

In the following three sections, we clarify how to perform the necessary computational geometry primitives and their efficiency to establish Theorem~\ref{thm:running_time_alg_greedy_segment}.

% --------------------------------------------------------------------
\subsection{Visibility polygons}
\label{sec:vis_pol}
% --------------------------------------------------------------------

For a simple polygon $P$ with $n$ vertices, the \emph{visibility polygon}~\cite{visibility_polygon_1981_ELGINDY1981186} $\visPol{P}{x}$ of point $x \in P$ is a polygon containing all points of $P$ visible from $x$. 
A point that guards $x$ must be placed in $\visPol{P}{x}$, which we also call the feasible region of $x$. 
The visibility polygon $\visPol{P}{x}$ has $\Oh{n}$ vertices and can be computed in $\Oh{n}$ time using the algorithm by Gindy and Avis~\cite{visibility_polygon_1981_ELGINDY1981186}.

% --------------------------------------------------------------------
\subsection{Intersecting polygons}
\label{sec:intersections}
% --------------------------------------------------------------------

Computing the intersection of two polygons $P$ and $Q$ with $n$ combined vertices is known as \emph{polygon clipping}, and is a classical problem in computer graphics and computational geometry. 
The running time of many of these algorithms, and in particular the output, depends on the number of intersections $h$ between $P$ and $Q$.
An algorithm by Martínez, Rueda, and Feito~\cite{martinez_2009_clipping_MARTINEZ20091177} shows that the problem of polygon clipping can be solved in $\Oh{(n + h) \log n}$ time.
Their approach is similar to the classical sweep line algorithm by Bentley and Ottmann~\cite{sweep_intersections_segments_1979_1675432} for computing the intersections between a set of segments.
There are also earlier algorithms for polygon clipping, but they are less efficient or do not handle degenerate cases~\cite{hormann_polygon_clipping_1998_10.1145/274363.274364, Vatti_polygon_clipping_1992_10.1145/129902.129906, Weiler_polygon_clipping_1977_10.1145/965141.563896}.

In general, the intersection of two polygons of sizes $n$ and $m$ may have size $h=\OhMega{nm}$, which could cause the repeated intersections of \Greedy to grow to size $\OhMega{n^{n}}$. 
However, we prove that any feasible region will have at most $\Oh{n}$ vertices:

\begin{lemma}
\label{lem:complexity_of_feasible_region_vertices}
    The feasible region of a set of points $S$ in a simple polygon $P$ with $n$ vertices has $\Oh{n}$ vertices.
\end{lemma}

\begin{proof}
    Let $v_1,v_2,\dots,v_t$ be the vertices of $S$. Let $F := VP(P, S)$ be the intersection $\bigcap\limits_{i = 1}^tVP(P, v_i)$. The boundaries of $VP(P, v_i)$ consists of sub line segments of $\partial P$ and sub line segments of rays from $v_i$ to reflex vertices of $P$. Hence, the boundary of $F$ consists of the same sub line segments.

    Any edge of $\partial P$ can only contribute to one edge of $\partial F$ (since $F \cap C$ is convex for every convex subset of $P$, Lemma~\ref{lem:feasible_region_is_connected}), so at most $n$ edges of $\partial F$ come from these. Thus, we only need to restrict the number of edges coming from rays. 

    Consider a reflex vertex $v_i$ and the rays stemming from visibility polygons defined by $v_i$. We now consider three rays, $r_1$, $r_2$ and $r_3$, which all pass through $v_i$. We will show, that at most two of these rays will contain sub line segments that are part of $\partial F$.

    Order the rays, by the angle they make with edge $v_{i-1}v_i$ and assume that $r_1$ and $r_3$ contain sub line segments part of $\partial F$. Thus $F$ will be contained in the half planes consisting of everything away from the path $v_{i-1}v_iv_{i+1}$ on the other side of $r_1$ and $r_3$ extended to lines. However, the only intersection of $r_2$ with these two half planes is $v_i$, hence it can not contribute an edge to $\partial F$. Doing this repeatedly, we get that at most two rays through $v_i$ can contribute edges to the feasible region, see Figure~\ref{fig:reflex_2_rays}.

    By Lemma~\ref{lem:feasible_region_is_connected} it holds that $F \cap C$ is convex for each convex subset $C$ of $P$, and we get that each ray can contribute at most one edge to the feasible region. Thus each reflex vertex can contribute at most two edges to $F$, meaning that the complexity of $F$ is at most $3n = \Oh{n}$.

\begin{figure}[h]
    \centering
    \begin{tikzpicture}
        \draw (4.68, -1.16) node[below] {$v_{i+1}$} -- (6.18, 0.62) node[above] {$v_i$} -- (6.95, -1.46) node[below] {$v_{i-1}$};

        \draw[feasible region] (3.88, -0.45) -- (8.78, 1.83);
        \draw[red] (3.76, 0.31) -- (9.77, 1.08);
        \draw[feasible region] (4.75, 1.33) -- (8.98, -0.78);

        \draw[feasible region, ->] (5.24, 1.09) -- (5.46, 1.55);
        \draw[feasible region, ->] (7.14, 1.07) -- (6.91, 1.55);

        \node[above, feasible region] at (4.64, -0.1) {$r_3$};
        \node[above, red] at (4.51, 0.4) {$r_2$};
        \node[above, feasible region] at (5.08, 1.17) {$r_1$};
        
    \end{tikzpicture}
    \caption{$r_1, r_2$ and $r_3$ are rays passing through $v_i$. If both $r_1$ and $r_3$ contribute to the edges of $F$, then $F$ lies in the quarter plane shown with arrows. Now $r_2$ will not be able to contribute to edges of $F$.}
    \label{fig:reflex_2_rays}
\end{figure}
    
\end{proof}

Combining Lemma~\ref{lem:complexity_of_feasible_region_vertices} with the polygon clipping algorithm, we get a running time of $\Oh{n \log n}$ to compute the intersection between the previous feasible region and a visibility polygon.

% --------------------------------------------------------------------
\subsection{The last visible point}
\label{sec:greedy_segment_ext_last_edge}
% --------------------------------------------------------------------

When computing $\lastVisPoint{P}{F_{i - 1}}{e_{i - 1}}$, the input is a simple polygon $P$ with $n$ vertices, an edge $e_{i - 1} = \lrParents{v_{i - 1}, v_i}$, and a feasible region $F_{i - 1}$ with $f$ vertices where $v_{i - 1}$ is visible from all points of $F_{i - 1}$ and $v_i$ is not visible from any point of $F_{i - 1}$. 
The output is the point $y$ on $e_{i - 1}$ closest to $v_i$ visible from $F_{i - 1}$ and the guard $g \in F_{i - 1}$ that can see $y$. 
We first show that it suffices to consider only the vertices of $F_{i - 1}$.

\begin{lemma}
\label{lem:extending_on_last_edge_vertices_of_F_enough}
    Let ${P}$, ${F_{i - 1}}$, and ${e_{i - 1}}$ be defined as above. Then the vertices of ${F_{i - 1}}$ can see at least as far along ${e_{i - 1}}$ as all of ${F_{i - 1}}$.
\end{lemma}

\begin{proof}
    We prove the lemma by showing that points in the interior of ${F_{i - 1}}$ can see no further than points on the boundary $\partial {F_{i - 1}}$ which again can see no further than the vertices of ${F_{i - 1}}$.
    
    First, let $y_I$ be a point on ${e_{i - 1}}$ visible from a point $p_I$ in the interior.
    Clearly, the point of $\linesegment{p_I}{y_I} \cap {F_{i - 1}}$ with minimal distance to $y_I$ is a point on the boundary of ${F_{i - 1}}$ that can see $y_I$.
    
    Second, let $y_B$ be a point along $e_{i - 1}$ visible from a point $p_B$ on edge $(a, b)$ of the boundary of $F_{i-1}$ that can see $y_B$. 
    Note $F_{i - 1}$ cannot intersect the line $H := \lineextension{v_{i-1}}{v_i}$ defined by $e_{i - 1}$, since that would imply a point in $F_{i - 1}$ that can see $v_i$ contradicting the input assumption that $e_{i-1}$ cannot be seen from $F_{i-1}$. 
    Recall that $F_{i - 1}$ is connected by Lemma~\ref{lem:feasible_region_is_connected} which further implies that it must lie exclusively on one side of $H$. 
    If ${F_{i - 1}}$ is below $H$ then all points of $F_{i - 1}$ can see no further than $v_{i - 1}$.
    Otherwise, ${F_{i - 1}}$ is above $H$, and in particular edge $(a, b)$ of $F_{i-1}$ is above $H$. 
    Let $a$ (symmetrically $b$) be the vertex on the same side of the half-plane defined by $\lineextension{p_B}{y_B}$ as $v_{i - 1}$.
    By the contrapositive of Lemma~\ref{lem:convex_viewing_lemma} the triangle $\tri{p_B y_B v_{i - 1}}$ is inside $P$.
    Thus, if $a$ is in $\tri{y_B p_B v_{i - 1}}$ it can also see $y_B$.
    Otherwise, $a$ and $v_i$ lie on opposite sides of the half-plane defined by $(p_B, v_{i - 1})$. 
    See Figure~\ref{fig:extending_on_last_edge_vertices_of_F_enough} for an example of this case.
    Since $a$ can see $v_{i - 1}$, the triangle $\tri{a p_B v_{i - 1}}$ is also inside $P$ by the contrapositive of Lemma~\ref{lem:convex_viewing_lemma}.
    Thus, the segment from $a$ to $y_B$ is fully contained in triangles $\tri{a p_B v_{i - 1}}$ and $\tri{y_B p_B v_{i - 1}}$, which implies that $a$ can see $y_B$.
\end{proof}

\begin{figure}[ht]
    \centering
    
    \begin{tikzpicture}[scale = -1]
        % Define points
        \coordinate (vi-1) at (2.5, 3);
        \coordinate (vi) at (6, 1);
        \coordinate (a) at (0, 2);
        \coordinate (b) at (0, 0.5);
        \coordinate (yB) at ($(vi-1)!0.7!(vi)$);
        \coordinate (pB) at ($(a)!0.6!(b)$);
        \coordinate (Tri1) at ($  (a)!0.333!(vi-1)!0.333!(pB) $);
        \coordinate (Tri2) at ($ (yB)!0.333!(vi-1)!0.333!(pB) $);

        % Fill triangles
        % \fill[red!20] (a) -- (vi-1) -- (pB) -- cycle;
        % \fill[green!20] (yB) -- (vi-1) -- (pB) -- cycle;

        % Draw edges
        \draw[feasible region] (a) -- (b);
        \draw (vi-1) --   (vi);
        \draw[densely dashed] (pB) -- (vi-1);
        \draw[densely dashed, guard] (a) -- (vi-1);

        % Extending lines
        \draw[densely dashed, color = guard] (pB) -- (yB);
        \draw[densely dashed, color = guard] (pB) -- (vi-1);
        \draw[densely dashed] ($(vi-1)!-0.3!(vi)$) -- (vi-1);
        \draw[densely dashed] (vi) -- ($(vi-1)!1.3!(vi)$);

         % Draw points and labels
        \filldraw[color = feasible region] (a)       circle (1.5 pt) node[below right] {$a$};
        \filldraw[color = feasible region] (b)       circle (1.5 pt) node[above left] {$b$};
        \filldraw (vi-1)                             circle (1.5 pt) node[left, xshift=-2pt] {$v_{i - 1}$};
        \filldraw (vi)                               circle (1.5 pt) node[above] {$v_i$};
        \filldraw (yB)                               circle (1.5 pt) node[above, xshift=2pt] {$y_B$};
        \filldraw[color = feasible region] (pB)      circle (1.5 pt) node[below right] {$p_B$};
        \node[above] at ($(vi)!1.3!(vi-1)$) {$H$};
    \end{tikzpicture}
    
    \caption{The case in the proof of Lemma~\ref{lem:extending_on_last_edge_vertices_of_F_enough} where vertex $a$ from the feasible region $F_{i-1}$ satisfies the following conditions: (1) it lies above the half-plane $H$ (2) it is on the same side of the half-plane defined by $(p_B, y_B)$ as $v_{i - 1}$ (3) it is on the opposite side of the half-plane defined by $(p_B, v_{i - 1})$ from $v_i$.}
    \label{fig:extending_on_last_edge_vertices_of_F_enough}
\end{figure}

\begin{remark}
\label{rem:_multiple_optimal_guards_lie_on_line_with_endpoint}
    The proof of Lemma~\ref{lem:extending_on_last_edge_vertices_of_F_enough} also shows that if multiple guard placements are optimal, they must be collinear with the endpoint ($y_B$ above).
\end{remark}

Thus we want to find the vertex of $F_{i-1}$, which sees most of $e_{i-1}$. We do this by computing the common tangents with a part of $\partial P$ and $F_{i-1}$. First, we need some terminology:

\begin{definition}[Free area and congested area]
Consider $e_{i-1} = (v_{i-1}, v_i)$ and $F_{i-1}$. Let $t_r$ and $t_m$ be the right and left tangent to $F_{i-1}$ going through $v_{i-1}$. The region between these two tangents bounded by $F_{i-1}$ will be without any part of $\partial P$ (since $F_{i-1}$ can see $v_{i-1}$) and we will therefor call it the free area. Let the left tangent to $F_{i-1}$ through $v_i$ be $t_\ell$. The region bounded by $t_m, t_\ell$ and $e_{i-1}$ will need to contain parts of $\partial P$ and will be called the congested area. Let the intersection point between $t_\ell$ and $t_m$ be $q$, see Figure~\ref{fig:last_feasible_region}.

\begin{figure}[h]
    \centering
    \begin{tikzpicture}
        \def\coords{
            \coordinate (A) at (4,1);
            \coordinate (B) at (1,1);
            \coordinate (C) at (3.5,3);
            \coordinate (D) at (2,3.5);
            \coordinate (E) at (1,4);
            \coordinate (F) at (0.5,4.5);
            \coordinate (G) at (0.42773,5);
            \coordinate (H) at (1.5,5);
            \coordinate (I) at (2,4.5);
            \coordinate (J) at (3,5);
            \coordinate (K) at (3.5,5);
            \coordinate (L) at (4,4);
            \coordinate (M) at (0.5,1.5);
            \coordinate (N) at (2,1.5);
            \coordinate (O) at (1.5,2.5);
            \coordinate (P) at (0.5,2);
            \coordinate (Q) at (1,2);
            \coordinate (R) at (0,1.5);
            \coordinate (S) at (0,2.5);
            \coordinate (T) at (1,3);
            \coordinate (U) at (0,4.5);
            \coordinate (V) at (0,5);
            \coordinate (W) at (0.922, 1.5);
            \coordinate (X) at (0.8274, 2.1673);
            \coordinate (Y) at (0.7422, 2.8645);
            \coordinate (Z) at (1,3.5);
            \coordinate (A1) at (0.0512, 4.4437);
            \coordinate (B1) at (0.5445, 0.9416);
            \coordinate (C1) at (1.5, 1);
            }

        \def\target{
            \coords
            \draw[thick] (A) node[anchor = north] {$v_{i-1}$} -- (B) node[anchor = north] {$v_i$};
            \fill (A) circle (2pt);
            \fill (B) circle (2pt);
        }

        \def\F{
            \coords
            \filldraw[F] (C) -- (D) -- (E) -- (F) -- (G) -- (H) -- (I) -- (J) -- (K) -- (L)  -- cycle;
            \node[feasible region] at ($(I) + (0,0.6)$) {$F_{i-1}$};
            }

        \newcommand{\partP}[1]{
            \coords
            \draw[#1] (B) -- (M) -- (N) -- (O) -- (P) -- (Q) -- (R) -- (S) -- (T) -- (E) -- (F) -- (U) -- (V);
        }

        \def\tr{
            \coords
            \draw[blue, thick] ($(A)!-0.05!(L)$) -- ($(A)!1.4!(L)$);
            \node[blue] at ($(A)!0.5!(L)$) [anchor = west] {$t_r$};
        }

        \def\tm{
            \coords
            \draw[blue] ($(A)!-0.05!(F)$) -- ($(A)!1.2!(F)$);
            \node[blue] at ($(A)!0.5!(F)$) [anchor = south west, label distance = 0.3] {$t_m$};
        }

        \def\tl{
            \coords
            \draw[blue] ($(B)!-0.02!(F)$) -- ($(B)!1.2!(F)$);
            \node[blue] at ($(B)!0.6!(F)$) [anchor = south east, label distance = 0.3] {$t_l$};
        }

        \def\cutout{
            \coords
            \draw[thick, dashed, blue] (B) -- (W) -- (N) -- (O) -- (X) -- (Y) -- (T) -- (E);
        }

        \def\Bbound{
            \coords
            \draw[thick, dashed, blue] (B) -- (W) -- (N) -- (O) -- (X) -- (Y) -- (T) -- (Z);
        }

        \def\B{
            \coords
            \filldraw[color = blue, fill = SkyBlue] (B) -- (W) -- (N) -- (O) -- (X) -- (Y) -- (T) -- (Z) -- (A1) -- (B1) -- cycle;
        }

        \def\qr{
            \coords
            \filldraw[feasible region] (F) circle (2pt) node[anchor = south] {$q$};
            \filldraw[blue] (A1) circle (2pt) node[anchor = east] {$r$};
        }

        \def\q{
            \coords
            \filldraw[blue] (F) circle (2pt) node[anchor = north east] {$q$};
        }

        \def\tangents{
            \coords
            \draw[red] ($(N)!-0.5!(C)$) -- ($(N)!1.3!(C)$);
            \draw[red] ($(F)!-0.7!(Z)$) -- ($(F)!1.3!(Z)$);
            \draw[red] ($(A1)!-0.3!(D)$) -- ($(A1)!1.3!(D)$);
            \draw[red] ($(G)!-0.3!(A1)$) -- ($(G)!1.3!(A1)$);
        }

        \def \gcy{
            \coords
            \filldraw[feasible region] (C) circle (2pt) node[anchor = north west] {$g$};
            \filldraw[blue] (N) circle (2pt) node[anchor = north west] {$c$};
            \filldraw[black] (C1) circle (2pt) node[anchor = north] {$y$};
        }

        \begin{scope}
            \F
            \partP{thick}
            \tr
            \tm
            \tl
            \target
            \q
        \end{scope}
            
    \end{tikzpicture}
    \caption{The last feasible region marked in green and tangents $t_r, t_m$ and $t_\ell$ marked in blue.} \label{fig:last_feasible_region}

\end{figure}

\end{definition}

The only parts of $\partial P$, that can block $F_{i-1}$ from seeing $v_i$ is the parts in the congested area. Thus, we construct a polygon from the $\partial P$ in the congested area:

\begin{definition}[Blocking polygon]
\label{def:blocking_polygon}
    We define the blocking polygon $B$ by considering the following path: Starting from $v_i$, follow $t_\ell$ until you hit $\partial P$ again. Then follow $\partial P$ into the congested region until you hit $t_\ell$ again. Continue this until you hit $q$ or $F_{i-1}$. If one hits $q$, one moves a little to the left to a point $r$ that lies below $t_m$ and then finishes by moving parallel to $t_\ell$ until you can move directly left to $v_i$. 
    
    If you hit $F_{i-1}$ before $q$, then this must occur along $t_m$. By Lemma~\ref{lem:feasible_region_is_connected}, $F_{i-1}$ will only touch $t_m$ along one edge of $\partial F_{i-1}$, so this edge is the one our path will touch. When the path hits this edge, we back up a small bit (small here depends on how close the rest of the path previously has gone to $t_m$ and the distance from the point of contact to $q$) and then head directly to $r$ (defined just as above) and then finish as above. By choosing the distance of the backtrack and distance between $q$ and $r$ small enough, it can be made so that the path does not intersect itself, as the path could not have hit $t_m$ before the intersection point with $F_{i-1}$. This is illustrated on Figure~\ref{fig:blocking_polygon}.
    This path will be $\partial B$ and $B$ is the region bounded by this path.
    
    \begin{figure}[h]
    \centering
    \begin{tikzpicture}[scale = 0.9]
        \def\coords{
            \coordinate (A) at (4,1);
            \coordinate (B) at (1,1);
            \coordinate (C) at (3.5,3);
            \coordinate (D) at (2,3.5);
            \coordinate (E) at (1,4);
            \coordinate (F) at (0.5,4.5);
            \coordinate (G) at (0.42773,5);
            \coordinate (H) at (1.5,5);
            \coordinate (I) at (2,4.5);
            \coordinate (J) at (3,5);
            \coordinate (K) at (3.5,5);
            \coordinate (L) at (4,4);
            \coordinate (M) at (0.5,1.5);
            \coordinate (N) at (2,1.5);
            \coordinate (O) at (1.5,2.5);
            \coordinate (P) at (0.5,2);
            \coordinate (Q) at (1,2);
            \coordinate (R) at (0,1.5);
            \coordinate (S) at (0,2.5);
            \coordinate (T) at (1,3);
            \coordinate (U) at (0,4.5);
            \coordinate (V) at (0,5);
            \coordinate (W) at (0.922, 1.5);
            \coordinate (X) at (0.8274, 2.1673);
            \coordinate (Y) at (0.7422, 2.8645);
            \coordinate (Z) at (1,3.5);
            \coordinate (A1) at (0.0512, 4.4437);
            \coordinate (B1) at (0.5445, 0.9416);
            \coordinate (C1) at (1.5, 1);
            }

        \def\target{
            \coords
            \draw[thick] (A) node[anchor = north] {$v_{i-1}$} -- (B) node[anchor = north] {$v_i$};
            \fill (A) circle (2pt);
            \fill (B) circle (2pt);
        }

        \def\F{
            \coords
            \filldraw[F] (C) -- (D) -- (E) -- (F) -- (G) -- (H) -- (I) -- (J) -- (K) -- (L)  -- cycle;
            \node[feasible region] at ($(I) + (0,0.6)$) {$F_{i-1}$};
            }

        \newcommand{\partP}[1]{
            \coords
            \draw[#1] (B) -- (M) -- (N) -- (O) -- (P) -- (Q) -- (R) -- (S) -- (T) -- (E) -- (F) -- (U) -- (V);
        }

        \def\tr{
            \coords
            \draw[blue, thick] ($(A)!-0.05!(L)$) -- ($(A)!1.4!(L)$);
            \node[blue] at ($(A)!0.5!(L)$) [anchor = west] {$t_r$};
        }

        \def\tm{
            \coords
            \draw[blue] ($(A)!-0.05!(F)$) -- ($(A)!1.2!(F)$);
            \node[blue] at ($(A)!0.5!(F)$) [anchor = south west, label distance = 0.3] {$t_m$};
        }

        \def\tl{
            \coords
            \draw[blue] ($(B)!-0.02!(F)$) -- ($(B)!1.2!(F)$);
            \node[blue] at ($(B)!0.6!(F)$) [anchor = south east, label distance = 0.3] {$t_l$};
        }

        \def\tlgray{
            \draw[gray] ($(B)!-0.02!(F)$) -- ($(B)!1.2!(F)$);
        }

        \def\cutout{
            \coords
            \draw[thick, dashed, blue] (B) -- (W) -- (N) -- (O) -- (X) -- (Y) -- (T) -- (E);
        }

        \def\Bbound{
            \coords
            \draw[thick, dashed, blue] (B) -- (W) -- (N) -- (O) -- (X) -- (Y) -- (T) -- (Z);
        }

        \def\B{
            \coords
            \filldraw[color = blue, fill = SkyBlue] (B) -- (W) -- (N) -- (O) -- (X) -- (Y) -- (T) -- (Z) -- (A1) -- (B1) -- cycle;
        }

        \def\qr{
            \coords
            \filldraw[feasible region] (F) circle (2pt) node[anchor = south east] {$q$};
            \filldraw[blue] (A1) circle (2pt) node[anchor = east] {$r$};
        }

        \def\q{
            \coords
            \filldraw[blue] (F) circle (2pt) node[anchor = north east] {$q$};
        }

        \def\tangents{
            \coords
            \draw[red] ($(N)!-0.5!(C)$) -- ($(N)!1.3!(C)$);
            \draw[red] ($(F)!-0.7!(Z)$) -- ($(F)!1.3!(Z)$);
            \draw[red] ($(A1)!-0.3!(D)$) -- ($(A1)!1.3!(D)$);
            \draw[red] ($(G)!-0.3!(A1)$) -- ($(G)!1.3!(A1)$);
        }

        \def \gcy{
            \coords
            \filldraw[feasible region] (C) circle (2pt) node[anchor = north west] {$g$};
            \filldraw[blue] (N) circle (2pt) node[anchor = north west] {$c$};
            \filldraw[black] (C1) circle (2pt) node[anchor = north] {$y$};
        }

        \begin{scope}
            \F
            \partP{very thin, gray}
            \tlgray
            \cutout
            \target
        \end{scope}

        \begin{scope}[xshift = 5cm]
            \F
            \Bbound
            \target
        \end{scope}

        \begin{scope}[xshift = 10cm]
            \B
            \F
            \qr
            \target
        \end{scope}
    \end{tikzpicture}
    \caption{Following $t_\ell$ and $\partial P$ into the congested area, one traces the boundary of $B$, marked in blue. Once we hit $F_{i-1}$ we backtrack and head directly for $r$.} \label{fig:blocking_polygon}
\end{figure}
\end{definition} 

Now the claim is the following: $F_{i-1}$ and $B$ are disjoint polygons and one of the common tangents between $F_{i_1}$ and $B$ will be the line of sight of the guard that sees the most of $e_{i-1}$, see Figure~\ref{fig:common_tangents}. Thus, to find $\lastVisPoint{P}{F_{i-1}}{e_{i-1}}$ we can run the $\textsc{CommonTangents}$ algorithm of Abrahamsen and Walczak \cite{common_tangents} to find all common tangents (of which there are at most four) and manually check them to find the guard position that sees the most of $e_{i-1}$. The claims will be proven in the following.

\begin{figure}[h]
    \centering
    \begin{tikzpicture}
        \def\coords{
            \coordinate (A) at (4,1);
            \coordinate (B) at (1,1);
            \coordinate (C) at (3.5,3);
            \coordinate (D) at (2,3.5);
            \coordinate (E) at (1,4);
            \coordinate (F) at (0.5,4.5);
            \coordinate (G) at (0.42773,5);
            \coordinate (H) at (1.5,5);
            \coordinate (I) at (2,4.5);
            \coordinate (J) at (3,5);
            \coordinate (K) at (3.5,5);
            \coordinate (L) at (4,4);
            \coordinate (M) at (0.5,1.5);
            \coordinate (N) at (2,1.5);
            \coordinate (O) at (1.5,2.5);
            \coordinate (P) at (0.5,2);
            \coordinate (Q) at (1,2);
            \coordinate (R) at (0,1.5);
            \coordinate (S) at (0,2.5);
            \coordinate (T) at (1,3);
            \coordinate (U) at (0,4.5);
            \coordinate (V) at (0,5);
            \coordinate (W) at (0.922, 1.5);
            \coordinate (X) at (0.8274, 2.1673);
            \coordinate (Y) at (0.7422, 2.8645);
            \coordinate (Z) at (1,3.5);
            \coordinate (A1) at (0.0512, 4.4437);
            \coordinate (B1) at (0.5445, 0.9416);
            \coordinate (C1) at (1.5, 1);
            }

        \def\target{
            \coords
            \draw[thick] (A) node[anchor = north] {$v_{i-1}$} -- (B) node[anchor = north] {$v_i$};
            \fill (A) circle (2pt);
            \fill (B) circle (2pt);
        }

        \def\F{
            \coords
            \filldraw[F] (C) -- (D) -- (E) -- (F) -- (G) -- (H) -- (I) -- (J) -- (K) -- (L)  -- cycle;
            \node[feasible region] at ($(I) + (0,0.6)$) {$F_{i-1}$};
            }

        \newcommand{\partP}[1]{
            \coords
            \draw[#1] (B) -- (M) -- (N) -- (O) -- (P) -- (Q) -- (R) -- (S) -- (T) -- (E) -- (F) -- (U) -- (V);
        }

        \def\tr{
            \coords
            \draw[blue, thick] ($(A)!-0.05!(L)$) -- ($(A)!1.4!(L)$);
            \node[blue] at ($(A)!0.5!(L)$) [anchor = west] {$t_r$};
        }

        \def\tm{
            \coords
            \draw[blue] ($(A)!-0.05!(F)$) -- ($(A)!1.2!(F)$);
            \node[blue] at ($(A)!0.5!(F)$) [anchor = south west, label distance = 0.3] {$t_m$};
        }

        \def\tl{
            \coords
            \draw[blue] ($(B)!-0.02!(F)$) -- ($(B)!1.2!(F)$);
            \node[blue] at ($(B)!0.6!(F)$) [anchor = south east, label distance = 0.3] {$t_l$};
        }

        \def\cutout{
            \coords
            \draw[thick, dashed, blue] (B) -- (W) -- (N) -- (O) -- (X) -- (Y) -- (T) -- (E);
        }

        \def\Bbound{
            \coords
            \draw[thick, dashed, blue] (B) -- (W) -- (N) -- (O) -- (X) -- (Y) -- (T) -- (Z);
        }

        \def\B{
            \coords
            \filldraw[color = blue, fill = SkyBlue] (B) -- (W) -- (N) -- (O) -- (X) -- (Y) -- (T) -- (Z) -- (A1) -- (B1) -- cycle;
        }

        \def\qr{
            \coords
            \filldraw[feasible region] (F) circle (2pt) node[anchor = south] {$q$};
            \filldraw[blue] (A1) circle (2pt) node[anchor = east] {$r$};
        }

        \def\tangents{
            \coords
            \draw[red] ($(N)!-0.5!(C)$) -- ($(N)!1.3!(C)$);
            \draw[red] ($(F)!-0.7!(Z)$) -- ($(F)!1.3!(Z)$);
            \draw[red] ($(A1)!-0.3!(D)$) -- ($(A1)!1.3!(D)$);
            \draw[red] ($(G)!-0.3!(A1)$) -- ($(G)!1.3!(A1)$);
        }

        \def \gcy{
            \coords
            \filldraw[feasible region] (C) circle (2pt) node[anchor = north west] {$g$};
            \filldraw[blue] (N) circle (2pt) node[anchor = north west] {$c$};
            \filldraw[black] (C1) circle (2pt) node[anchor = north] {$y$};
        }

        \begin{scope}[xshift = 5cm, yshift = -5cm]
            \B
            \F
            \tangents
            \target
        \end{scope}
            
    \end{tikzpicture}
    \caption{Common tangents drawn. One of which sees the line of sight from the optimal guard to the farthest point on $e_{i-1}$, which is visible from $F_{i-1}$}\label{fig:common_tangents}

\end{figure}

\begin{lemma}[Disjoint]
    $B$ and $F_{i-1}$ are disjoint.
\end{lemma}
\begin{proof}
    $F_{i-1}$ is contained in the free area and $B$ is contained below $t_m$. Thus, the only intersection points will be along $t_m$. By backtracking on the path of $\partial B$ before an intersection with $F_{i-1}$ and then going directly to $r$, we no longer touch the free area for the rest of the path, so $\partial B$ does not touch $F_{i-1}$ and must $B$ not either.
\end{proof}

Let $g$ be the vertex of $F_{i-1}$ that sees the most of $e_{i-1}$ and let $c$ be the first vertex of $\partial P$ blocking $g$ from seeing any more of $e_{i-1}$ (when moving along $\partial P$ from $v_i$ clockwise). We want to show that $\lineextension{g}{c}$ is a common tangent of $F_{i-1}$ and $B$. First, we show that $c$ is a vertex of $B$.

\begin{lemma}[$c$ is present]
    $c$ is a vertex of $B$.
\end{lemma}
\begin{proof}

    For $c$ to be the blocking vertex, it must be in the congested area. Many vertices of $P$, which are in the congested area, are vertices of $B$. The only ones that get skipped are ones that are hidden by other parts of $\partial P$ like vertex $h$ in Figure~\ref{fig:h_is_hidden}, and vertices that would be skipped, because we hit $F_{i-1}$ and head directly for $r$.

    Vertices such as $h$ cannot be the blocking vertex, since to draw a line from $F_{i-1}$ to $h$, one must cross other edges of $\partial P$, so we exclude these.

    Assume now that the path followed in Definition~\ref{def:blocking_polygon} where to intersect $F_{i-1}$, and the first time this happens is at point $s$. If $s$ itself is a vertex of $\partial P$, it cannot be $c$, as by placing a guard at $s$ (which is in $F_{i-1}$), one sees past $s$ and sees more of $e_{i-1}$ than $g$ would, contradicting that $s$ blocked the guard that sees the farthest. 

    Now assume that $c$ is a vertex of $\partial P$ in the congested area, which would have been hit after $s$ by following $\partial P$ and $t_\ell$ as in Definition~\ref{def:blocking_polygon} if we did not skip past them after hitting $s$. Since the part of $\partial P$, in the congested region that contains $s$ must enter and leave the congested area via $t_\ell$, the path after $s$ must continue left towards $t_\ell$ or stay away from $t_\ell$. So $c$ would be stuck to the left of the part of $\partial B$ until $s$. So we again have a problem, since drawing a line from $F_{i-1}$ to $s$, will hit $\partial B$ in the congested area, i.e., hit $\partial P$ and $c$ cannot be a blocking vertex for $g$.

    Since none of the above cases for $c$ are possible, it must be that $c$ is a vertex of $B$.

    \begin{figure}[h]
    \centering
    \begin{tikzpicture}[scale = 0.9]
        \def\coords{
            \coordinate (A) at (4,1);
            \coordinate (B) at (1,1);
            \coordinate (C) at (3.5,3);
            \coordinate (D) at (2,3.5);
            \coordinate (E) at (1,4);
            \coordinate (F) at (0.5,4.5);
            \coordinate (G) at (0.42773,5);
            \coordinate (H) at (1.5,5);
            \coordinate (I) at (2,4.5);
            \coordinate (J) at (3,5);
            \coordinate (K) at (3.5,5);
            \coordinate (L) at (4,4);
            \coordinate (M) at (0.5,1.5);
            \coordinate (N) at (2,1.5);
            \coordinate (O) at (1.5,2.5);
            \coordinate (P) at (0.5,2);
            \coordinate (Q) at (1,2);
            \coordinate (R) at (0,1.5);
            \coordinate (S) at (0,2.5);
            \coordinate (T) at (1,3);
            \coordinate (U) at (0,4.5);
            \coordinate (V) at (0,5);
            \coordinate (W) at (0.922, 1.5);
            \coordinate (X) at (0.8274, 2.1673);
            \coordinate (Y) at (0.7422, 2.8645);
            \coordinate (Z) at (1,3.5);
            \coordinate (A1) at (0.0512, 4.4437);
            \coordinate (B1) at (0.5445, 0.9416);
            \coordinate (C1) at (1.5, 1);
            }

        \def\target{
            \coords
            \draw[thick] (A) node[anchor = north] {$v_{i-1}$} -- (B) node[anchor = north] {$v_i$};
            \fill (A) circle (2pt);
            \fill (B) circle (2pt);
        }

        \def\F{
            \coords
            \filldraw[F] (C) -- (D) -- (E) -- (F) -- (G) -- (H) -- (I) -- (J) -- (K) -- (L)  -- cycle;
            \node[feasible region] at ($(I) + (0,0.6)$) {$F_{i-1}$};
            }

        \newcommand{\partP}[1]{
            \coords
            \draw[#1] (B) -- (M) -- (N) -- (O) -- (P) -- (Q) -- (R) -- (S) -- (T) -- (E) -- (F) -- (U) -- (V);
        }

        \def\tr{
            \coords
            \draw[blue, thick] ($(A)!-0.05!(L)$) -- ($(A)!1.4!(L)$);
            \node[blue] at ($(A)!0.5!(L)$) [anchor = west] {$t_r$};
        }

        \def\tm{
            \coords
            \draw[blue] ($(A)!-0.05!(F)$) -- ($(A)!1.2!(F)$);
            \node[blue] at ($(A)!0.5!(F)$) [anchor = south west, label distance = 0.3] {$t_m$};
        }

        \def\tl{
            \coords
            \draw[blue] ($(B)!-0.02!(F)$) -- ($(B)!1.2!(F)$);
            \node[blue] at ($(B)!0.6!(F)$) [anchor = south east, label distance = 0.3] {$t_l$};
        }

        \def\tlgray{
            \draw[gray] ($(B)!-0.02!(F)$) -- ($(B)!1.2!(F)$);
        }

        \def\cutout{
            \coords
            \draw[thick, dashed, blue] (B) -- (W) -- (N) -- (O) -- (X) -- (Y) -- (T) -- (E);
        }

        \def\Bbound{
            \coords
            \draw[thick, dashed, blue] (B) -- (W) -- (N) -- (O) -- (X) -- (Y) -- (T) -- (Z);
        }

        \def\B{
            \coords
            \filldraw[color = blue, fill = SkyBlue] (B) -- (W) -- (N) -- (O) -- (X) -- (Y) -- (T) -- (Z) -- (A1) -- (B1) -- cycle;
        }

        \def\qr{
            \coords
            \filldraw[feasible region] (F) circle (2pt) node[anchor = south east] {$q$};
            \filldraw[blue] (A1) circle (2pt) node[anchor = east] {$r$};
        }

        \def\q{
            \coords
            \filldraw[blue] (F) circle (2pt) node[anchor = north east] {$q$};
        }

        \def\tangents{
            \coords
            \draw[red] ($(N)!-0.5!(C)$) -- ($(N)!1.3!(C)$);
            \draw[red] ($(F)!-0.7!(Z)$) -- ($(F)!1.3!(Z)$);
            \draw[red] ($(A1)!-0.3!(D)$) -- ($(A1)!1.3!(D)$);
            \draw[red] ($(G)!-0.3!(A1)$) -- ($(G)!1.3!(A1)$);
        }

        \def \gcy{
            \coords
            \filldraw[feasible region] (C) circle (2pt) node[anchor = north west] {$g$};
            \filldraw[blue] (N) circle (2pt) node[anchor = north west] {$c$};
            \filldraw[black] (C1) circle (2pt) node[anchor = north] {$y$};
        }

        \def\hhidden{
            \filldraw[gray] (Q) circle (2pt) node[anchor = north west] {$h$};
        }

        \begin{scope}
            \F
            \partP{very thin, gray}
            \tlgray
            \cutout
            \hhidden
            \target
        \end{scope}

    \end{tikzpicture}
    \caption{$h$ is hidden behind other parts of $\partial P$, so $h$ cannot be a blocking vertex.} \label{fig:h_is_hidden}
\end{figure}
\end{proof}

Finally, we show that $\lineextension{g}{c}$ is a common tangent.
\begin{lemma}[Sight line is a common tangent]
    $\lineextension{g}{c}$ is a common tangent of $F_{i-1}$ and $B$
\end{lemma}
\begin{proof}
    Consider the area below $\lineextension{g}{c} =:\ell$ let $\ell$ intersect $e_{i-1}$ at $y$, see Figure~\ref{fig:gcy_common_tangent}. If a point $g'$ of $F_{i-1}$ lies below $\ell$, $g'$ will see farther along $e_{i-1}$ than $y$: The line segment $\linesegment{g'}{y}$ will now not touch $B$, as it is strictly below $\ell$. Thus there will be some point $y'$ further along $e_{i-1}$ which also is not blocked from view of $g'$. This contradicts the optimality of $g$'s sight, so this cannot happen.

    If a point $b$ of $B$ lies below $\ell$, it must be a point along $\partial P$, and this will break the line of sight from $g$ to $y$, since $b$ must be connected to the rest of $B$, which lies above $\ell$. Thus $\ell$ is a tangent of $F_{i-1}$ and $B$.

    \begin{figure}[h]
    \centering
    \begin{tikzpicture}
        \def\coords{
            \coordinate (A) at (4,1);
            \coordinate (B) at (1,1);
            \coordinate (C) at (3.5,3);
            \coordinate (D) at (2,3.5);
            \coordinate (E) at (1,4);
            \coordinate (F) at (0.5,4.5);
            \coordinate (G) at (0.42773,5);
            \coordinate (H) at (1.5,5);
            \coordinate (I) at (2,4.5);
            \coordinate (J) at (3,5);
            \coordinate (K) at (3.5,5);
            \coordinate (L) at (4,4);
            \coordinate (M) at (0.5,1.5);
            \coordinate (N) at (2,1.5);
            \coordinate (O) at (1.5,2.5);
            \coordinate (P) at (0.5,2);
            \coordinate (Q) at (1,2);
            \coordinate (R) at (0,1.5);
            \coordinate (S) at (0,2.5);
            \coordinate (T) at (1,3);
            \coordinate (U) at (0,4.5);
            \coordinate (V) at (0,5);
            \coordinate (W) at (0.922, 1.5);
            \coordinate (X) at (0.8274, 2.1673);
            \coordinate (Y) at (0.7422, 2.8645);
            \coordinate (Z) at (1,3.5);
            \coordinate (A1) at (0.0512, 4.4437);
            \coordinate (B1) at (0.5445, 0.9416);
            \coordinate (C1) at (1.5, 1);
            }

        \def\target{
            \coords
            \draw[thick] (A) node[anchor = north] {$v_{i-1}$} -- (B) node[anchor = north] {$v_i$};
            \fill (A) circle (2pt);
            \fill (B) circle (2pt);
        }

        \def\F{
            \coords
            \filldraw[F] (C) -- (D) -- (E) -- (F) -- (G) -- (H) -- (I) -- (J) -- (K) -- (L)  -- cycle;
            \node[feasible region] at ($(I) + (0,0.6)$) {$F_{i-1}$};
            }

        \newcommand{\partP}[1]{
            \coords
            \draw[#1] (B) -- (M) -- (N) -- (O) -- (P) -- (Q) -- (R) -- (S) -- (T) -- (E) -- (F) -- (U) -- (V);
        }

        \def\tr{
            \coords
            \draw[blue, thick] ($(A)!-0.05!(L)$) -- ($(A)!1.4!(L)$);
            \node[blue] at ($(A)!0.5!(L)$) [anchor = west] {$t_r$};
        }

        \def\tm{
            \coords
            \draw[blue] ($(A)!-0.05!(F)$) -- ($(A)!1.2!(F)$);
            \node[blue] at ($(A)!0.5!(F)$) [anchor = south west, label distance = 0.3] {$t_m$};
        }

        \def\tl{
            \coords
            \draw[blue] ($(B)!-0.02!(F)$) -- ($(B)!1.2!(F)$);
            \node[blue] at ($(B)!0.6!(F)$) [anchor = south east, label distance = 0.3] {$t_l$};
        }

        \def\cutout{
            \coords
            \draw[thick, dashed, blue] (B) -- (W) -- (N) -- (O) -- (X) -- (Y) -- (T) -- (E);
        }

        \def\Bbound{
            \coords
            \draw[thick, dashed, blue] (B) -- (W) -- (N) -- (O) -- (X) -- (Y) -- (T) -- (Z);
        }

        \def\B{
            \coords
            \filldraw[color = blue, fill = SkyBlue] (B) -- (W) -- (N) -- (O) -- (X) -- (Y) -- (T) -- (Z) -- (A1) -- (B1) -- cycle;
        }

        \def\qr{
            \coords
            \filldraw[feasible region] (F) circle (2pt) node[anchor = south] {$q$};
            \filldraw[blue] (A1) circle (2pt) node[anchor = east] {$r$};
        }

        \def\tangents{
            \coords
            \draw[red] ($(N)!-0.5!(C)$) -- ($(N)!1.3!(C)$);
            \draw[red] ($(F)!-0.7!(Z)$) -- ($(F)!1.3!(Z)$);
            \draw[red] ($(A1)!-0.3!(D)$) -- ($(A1)!1.3!(D)$);
            \draw[red] ($(G)!-0.3!(A1)$) -- ($(G)!1.3!(A1)$);
        }

        \def \gcy{
            \coords
            \filldraw[feasible region] (C) circle (2pt) node[anchor = north west] {$g$};
            \filldraw[blue] (N) circle (2pt) node[anchor = north west] {$c$};
            \filldraw[black] (C1) circle (2pt) node[anchor = north] {$y$};
        }

        \begin{scope}
            \B
            \F
            \tangents
            \target
            \gcy
        \end{scope}
            
    \end{tikzpicture}
    \caption{The line of sight, is included in one of the common tangents.}\label{fig:gcy_common_tangent}
\end{figure}
\end{proof}

\begin{proposition}[Runtime]
    $\lastVisPoint{P}{F_{i-1}}{e_{i-1}}$ runs in time $\Oh{n\log n}$
\end{proposition}
\begin{proof}
    $\lastVisPoint{P}{F_{i-1}}{e_{i-1}}$ has three steps:
    \begin{enumerate}
        \item Compute $B$.
        \item Compute common tangents.
        \item Check which tangent is $\lineextension{g}{c}$ and compute $y$.
    \end{enumerate}
    Step 2 is simply running the $\textsc{CommonTangents}$ algorithm of Abrahamsen and Walczak \cite{common_tangents}, which takes $\Oh{n}$ time, and Step 3 can also be done in linear time, since there are at most four tangents to check. Thus, we need to show that $B$ can be computed in $\Oh{n\log n}$ time. 

    First, we compute $t_m$ and $t_\ell$, which can be done in linear time. Next, we compute which segments of $\partial P$ intersect $t_\ell$ and sort them along $t_\ell$. This takes $\Oh{n \log n}$ time. Finally, determine the vertex or edge of $F_{i-1}$ that lies on $t_m$. This can be done in linear time.

    Now we can trace the path from Definition~\ref{def:blocking_polygon}, keeping track of the minimal distance from the vertices considered to $t_m$. If $\partial B$ does not hit $F_{i-1}$ before reaching $q$, finish directly. Otherwise, compute a suitable backtrack distance and distance to $r$, so as not to overlap $\partial B$. This can be done in constant time, when the distance from $\partial B$ to $t_m$ is known. Then connect to $r$ and finish as before. This takes $\Oh{n}$ time.
    
    In total, it takes $\Oh{n \log n}$ time to compute $B$ which is the bottleneck of the algorithm.
\end{proof}

% --------------------------------------------------------------------
\subsection{Proving the runtime of \Greedy}
\label{sec:greedy_runtime_proof}
% --------------------------------------------------------------------

We now have the tools necessary to prove Theorem~\ref{thm:running_time_alg_greedy_segment}, as restated below.

\thmGreedyIntervalRuntime*

\begin{proof}
    Computing the initial feasible region on line~\ref{VisibilityPolygon} takes $\Oh{n}$ time per Section~\ref{sec:vis_pol}.
    The while loop runs for $e + 1 > 2$ iterations and in each iteration, a visibility polygon of complexity $\Oh{n}$ is computed and the intersection of two polygons of complexity $\Oh{n}$ is computed, resulting in a new complexity $\Oh{n}$ polygon (by Lemma~\ref{lem:complexity_of_feasible_region_vertices} the feasible region has complexity $\Oh{n}$), taking $\Oh{n\log n}$ operations, according to Section~\ref{sec:intersections}. Thus, the while loop takes $\Oh{en\log n}$ operations.
    
    Computing how far the greedy interval extends on the last edge takes $\Oh{n\log n}$ time per Section~\ref{sec:greedy_segment_ext_last_edge}.
    In total, it takes $\Oh{en \log n}$ operations to run \Greedy.
\end{proof}

The above also implies a straightforward bound on the time for Algorithm~\ref{alg:greedy} to perform a revolution of $P$.

\begin{corollary}
\label{cor:greedy_traverses_P_in_n3_log_n}
    Repeatedly using the \Greedy algorithm to perform a revolution of $P$ takes $\Oh{n^2 \log n}$ time.
\end{corollary}

\begin{proof}
    Let $k$ be the number of iterations of the \Greedy algorithm to perform a revolution of $P$.
    The bound on the running time follows by applying Theorem~\ref{thm:running_time_alg_greedy_segment} to each of the $k$ iterations and using that the number of edges intersected by the greedy intervals sum to at most $2n$.
\end{proof}

With the time complexity of \Greedy being polynomially bound, we only need to show that a polynomial number of revolutions is needed to find an optimal solution. This will be the content of Sections~\ref{sec:the_greedy_algorithm_a_combinatorial_viewpoint} and \ref{sec:the_greedy_algorithm_geometry}.

% --------------------------------------------------------------------
\section{\RepeatedGreedy, a combinatorial viewpoint}
\label{sec:the_greedy_algorithm_a_combinatorial_viewpoint}
% --------------------------------------------------------------------

Summing up what we have done so far: Given any starting point $x$, the \Greedy algorithm finds the point on the boundary $y$ that is furthest from $x$ in the clockwise direction, such that there is a guard that can see the entire polygon chain between $x$ and $y$.
We use the notation $[x,y]$ for this polygonal chain, and call it a \emph{greedy interval}.
Generally, for any $a, b \in \partial P$ we let $[a,b]$ denote the polygon chain of $\partial P$ from $a$ to $b$ including $a$ and $b$. 
Likewise, we define $(a,b)$ to be the polygon chain of $\partial P$ from $a$ to $b$ excluding $a$ and $b$.
The notations $(a,b]$ and $[a, b)$ are defined analogously, with the inclusion of $b$ or $a$, respectively.
In this section, we focus solely on the combinatorial properties of greedy intervals, i.e. their relative positions to one another. 
As such, we visualize $\partial P$ as a circle, as shown in Figure~\ref{fig:intervals_on_boundaries}.
In this way, we cast the problem as a circle-cover minimization problem with infinitely many circular arches. 
The setting with finitely many arches was studied by Lee and Lee~\cite{circle_cover_arc_minimization_1984_LEE1984109}. 
Naturally, some concepts from their work are related, and we will highlight these connections where relevant. 

\begin{definition}[Visible intervals]
    We denote a polygonal chain of the boundary between vertices $a$ and $b$ as the \emph{interval} $[a, b] \subseteq \partial P$. 
    If one can place a guard in $P$, that can see all of $[a, b]$ we call it a \emph{visible interval}. 
    The set of all visible intervals is denoted $\Fc$.
\end{definition}

\begin{remark}
    To simplify the notation and figures, we write $\Gsc$ instead of \Greedy and let the input polygon and the output guard position be implicit, i.e. $\Gsc$ takes a point $c$ on the boundary and returns the farthest clockwise boundary point $d$ such that $[c,d]$ is a visible interval.
\end{remark}

\begin{figure}[ht]
    \centering

    \begin{tikzpicture}[scale=0.7]
        \coordinate (A) at (3.78, 4.78);
        \coordinate (B) at (2.94, 6.54);
        \coordinate (C) at (4.32, 7.36);
        \coordinate (D) at (3.7, 8.86);
        \coordinate (E) at (4.82, 9.16);
        \coordinate (F) at (5.06, 7.6);
        \coordinate (G) at (6.62, 7.78);
        \coordinate (H) at (8.62, 6.7);
        \coordinate (I) at (8.52, 4.74);
        \coordinate (J) at (7.36, 5.9);
        \coordinate (K) at (7.56, 4.56);
        \coordinate (L) at (7, 4);
        \coordinate (M) at (6.14, 3.36);
        \coordinate (N) at (4.16, 4);
        \coordinate (O) at (5.62, 4.54);
        \coordinate (P) at (4.88, 5.56);
        \coordinate (Q) at (3.48, 6.86);
        \coordinate (R) at (5.72, 7.68);
        \coordinate (S) at (8.58, 5.92);
        \coordinate (T) at (5.2, 5.11);

        \draw[color = black] (A) -- (B) -- (C) -- (D) -- (E) -- (F) -- (G) -- (H) -- (I) -- (J) -- (K) -- (L) -- (M) -- (N) -- (O) -- (P) -- cycle;

        \draw[color = blue, thick] (Q) -- (C) -- (D) -- (E) -- (F) -- (R);
        \draw[color = blue] (4.2, 9) node[anchor = south] {$(a,b)$};
        \filldraw[color = blue, fill = white] (Q) circle (2 pt);
        \draw[color = blue] (Q) node[anchor = south] {$a$};
        
        \filldraw[color = blue, fill = white] (R) circle (2 pt);
        \draw[color = blue] (R) node[anchor = south] {$b$};
        
        \draw[color = red, thick] (J) -- (K) -- (L) -- (M) -- (N) -- (O) -- (P);
        \draw[color = red] (6.2, 2.5) node[anchor = south] {$[c, d]$};

        \filldraw[color = red] (J) circle (2pt);
        \draw[color = red, above] (J) node {$c$};

        \filldraw[color = red] (P) circle (2pt);
        \draw[color = red, above] (P) node {$d$};
    \end{tikzpicture}%
    \hspace{7em}%
    \begin{tikzpicture}[scale = 0.5]
        \draw (0, 0) circle (3.2);
        \draw[blue, thick, domain = 72.811:118.648] plot ({3.2*cos(\x)}, {3.2*sin(\x)});
        \draw[blue, above] (100 : 3.2) node {$(a,b)$};

        \filldraw[blue, fill = white] (-1.534, 2.808) circle (4pt);
        \draw[blue] (-1.750, 3.203) node {$a$};
        \filldraw[blue, fill = white] (0.946, 3.057) circle (4pt);
        \draw[blue] (1.079, 3.487) node {$b$};

        \draw[red, below] (1, -3.2) node {$[c,d]$};
        \draw[red, thick, domain = -139.183 : -1.673] plot ({3.2*cos(\x)}, {3.2*sin(\x)});
        \filldraw[red] (3.199, -0.093) circle (4pt);
        \draw[red] (3.648, -0.107) node {$c$};
        \filldraw[red] (-2.422, -2.092) circle (4pt);
        \draw[red] (-2.9, -2.4) node {$d$};
    \end{tikzpicture}
    
    \caption{A simple polygon visualized as a circle such that contiguous polygonal chains of the boundary overlap in the polygon if and only if the corresponding intervals overlap in the circle representation. 
    The interval $(a,b)$ is not a greedy interval since $b$ can be moved further around the polygon (clockwise) and remain visible to a guard.
    The interval $[c, d]$ is a greedy interval since $d$ can be moved no further, i.e. $\Gsc(c) = d$.}
    \label{fig:intervals_on_boundaries}
\end{figure}

\begin{remark}
\label{rem:one_guard_is_not_enough}
    We will from now on always assume that $P$ is not star-shaped, as otherwise when we run the \Greedy algorithm it will find that the entire boundary can be covered by a single guard, which is clearly optimal.
    Thus, $\Gsc(x) \neq x$ for all $x \in \partial P$.
    Furthermore, the greedy interval $[x,\Gsc(x)]$ from any point $x \in \partial P$ always contains at least one edge, since a guard can be placed at the first vertex encountered from $x$.
\end{remark}

\begin{lemma}[Properties of visible intervals]
\label{lem:properties_of_Fc}
    \text{}
    \begin{enumerate}
        \item If $[a,b] \subseteq [c,d]$ and $[c,d] \in \Fc$ then $[a, b]\in \Fc$.
        \item $[x, \Gsc(x)] \in \Fc$.
        \item $[x,y] \subseteq [x, \Gsc(x)]$ if and only if $[x, y] \in \Fc$.
    \end{enumerate}
\end{lemma}

\begin{proof}
    1. and 2. follow by definition.
    3. The only if implication follows directly from 1. and~2. For the if implication let $[x, y] \in \Fc$ and assume for contradiction that $[x,y] \not\subseteq [x, \Gsc(x)]$. Then $\Gsc(x)$ must be strictly before $y$, contradicting the maximality of \Greedy.
\end{proof}

Lemma~\ref{lem:properties_of_Fc} together with Remark~\ref{rem:one_guard_is_not_enough} acts as combinatorial axioms for our problem. Thus, for the rest of the section, we will use these to obtain properties of \Greedy.

\begin{lemma}
\label{lem:properties_of_greedy}
    Let $[a, b] \in \Fc$ and $x, y\in \partial P$, then
    \begin{enumerate}
        \item $[x, \Gsc(x)] \subseteq [a, b]$ implies $\Gsc(x) = b$.
        \item $x \in [a, b]$ implies $b \in [x, \Gsc(x)]$.
        \item If $[x,y] \subseteq [a,b]$, then $\Gsc(x) \in [b, \Gsc(y)]$.
    \end{enumerate}
\end{lemma}

\begin{figure}[ht]
    \centering
    
    \begin{tikzpicture}[scale=0.5]
        \begin{scope}
            \node at (-5, 0) {};
            
            \draw (0, 0) circle (3.2);
            \filldraw[color=black] (-3.192, 0.226) circle (4pt); % angle = 175.944
            \draw[black] (-3.641, 0.258) node {$a$};
            \filldraw[color=black] (-0.946, 3.057) circle (4pt); % angle = 107.189
            \draw[black] (-1.079, 3.487) node {$x$};
            \filldraw[color=black] (3.192, 0.226) circle (4pt); % angle = 4.056
            \draw[black] (3.641, 0.258) node {$b$};
            \draw[-Stealth,line width=0.3mm, red] (32.704:3.2) arc[radius=3.2, start angle=32.704, end angle= 4.056];
            \filldraw[color=black] (2.693, 1.729) circle (4pt); % angle = 32.704
            \draw[black] (3.271, 2.172) node {$\Gsc(x)$};
        \end{scope}
            
        \begin{scope}[xshift=12cm]
            \node at (5, 0) {};
                
            \draw (0, 0) circle (3.2);
            \draw[red, domain = 175.944 : 4.056] plot ({3.2*cos(\x)}, {3.2*sin(\x)});
            \draw[blue, domain = 107.189 : -110.535] plot ({3*cos(\x)}, {3*sin(\x)});
            \draw[green, domain = 212.7 - 360 : 32.704 ] plot ({2.8*cos(\x)}, {2.8*sin(\x)});
            
            \filldraw[color=black] (-3.192, 0.226) circle (3pt); % angle = 175.944
            \draw[black] (-3.641, 0.258) node {$a$};
            \filldraw[color=black] (-0.946, 3.057) circle (3pt); % angle = 107.189
            \draw[black] (-1.079, 3.487) node {$x$};
            \filldraw[color=black] (2.693, 1.729) circle (3pt); % angle = 32.704
            \draw[black] (3.071, 1.972) node {$y$};
            \filldraw[color=black] (3.192, 0.226) circle (3pt); % angle = 4.056
            \draw[black] (3.641, 0.258) node {$b$};
            \filldraw[color=black] (-1.123, -2.997) circle (3pt); % angle = -110.535
            \draw[black] (-1.380, -3.618) node {$\Gsc(x)$};
            \filldraw[color=black] (-2.693, -1.729) circle (3pt); % angle = 212.7
            \draw[black] (-3.671, -1.972) node {$\Gsc(y)$};
         \end{scope}
    \end{tikzpicture}
        
    \caption{Left, in Lemma~\ref{lem:properties_of_greedy}.1 $\Gsc(x)$ has to reach $b$ or further. Right, in Lemma~\ref{lem:properties_of_greedy}.3 the intervals $[a,b], [x, \Gsc(x)]$ and $[y,\Gsc(y)]$ and their relative positions are marked.}
    \label{fig:properties_of_greedy}
\end{figure}

\begin{proof}
    1. If $[x, \Gsc(x)] \subseteq [a,b]$ then $[x, \Gsc(x)] \subseteq [x, b] \in \Fc$, since $[x, b] \subseteq [a,b]$ and $[a,b] \in \Fc$. Using Lemma~\ref{lem:properties_of_Fc}.3. we get $[x, b] \subseteq [x, \Gsc(x)]$ thus $\Gsc(x) = b$.
    
    2. If $\Gsc(x) \not \in [x, b]$, then $x,b$ and $\Gsc(x)$ appear in that order on $\partial P$, thus $b \in [x,\Gsc(x)]$. If $\Gsc(x) \in [x,b]$ then $\Gsc(x) = b$ by Lemma~\ref{lem:properties_of_greedy}.1.
    
    3. Since $x \in [x, y] \subseteq [a, b] \in \Fc$ we have by Lemma~\ref{lem:properties_of_greedy}.2. that $y \in [x, \Gsc(x)]$ and using Lemma~\ref{lem:properties_of_greedy}.2 again $\Gsc(x) \in [y, \Gsc(y)]$. If $\Gsc(x) \not \in [y, b)$, we would contradict Lemma~\ref{lem:properties_of_greedy}.1, so $\Gsc(x) \in [b,\Gsc(y)]$.
\end{proof}

An \emph{optimal solution} refers to a set of points $y_0,\dots,y_{\opt-1}$ on the boundary such that they form a partition of $\partial P$ into $\opt$ visible intervals, i.e. $[y_i,y_{i+1}]\in\Fc$ and $\bigcup_{i=0}^{\opt} [y_i,y_{i+1}] = \partial P$ and $\opt$ is minimal with this property. 
Next, we show that any \emph{greedy step}, $(x,\Gsc(x)]$, always contains at least one point from every optimal solution (similar to Lemma 2.3 in~\cite{circle_cover_arc_minimization_1984_LEE1984109}). 

\begin{corollary}[Greedy steps contain optimal endpoints]
\label{cor:greedy_steps_contain_optimal_endpoints}
    Let $y_0, y_1,\dots,y_{\opt-1}$ be an optimal solution and $x \in \partial P$. Then $(x, \Gsc(x)]$ contains at least one $y_i$.
\end{corollary}

\begin{proof}
    \begin{figure}[ht]
        \centering
        
        \begin{tikzpicture}[scale=0.5]
            \begin{scope}
                \node at (-5, 0) {};
            
                \draw (0, 0) circle (3.2);
                \draw[blue, thick, domain = 107.189:4.056] plot ({3.2*cos(\x)}, {3.2*sin(\x)});
                \filldraw[color=black] (-3.083, 0.856) circle (4pt); % angle = 164.485
                \draw[black] (-3.617, 0.976) node {$y_i$};
                \filldraw[color=blue, fill = white] (-0.946, 3.057) circle (4pt); % angle = 107.189
                \draw[blue] (-1.079, 3.487) node {$x$};
                \filldraw[color=black] (2.296, 2.229) circle (4pt); % angle = 44.163
                \draw[black] (2.818, 2.743) node {$y_{i+1}$};
                \filldraw[color=blue] (3.192, 0.226) circle (4pt); % angle = 4.056
                \draw[blue] (4.241, 0.258) node {$G(x)$};
            \end{scope}
            
            \begin{scope}[xshift=12cm]
                \node at (5, 0) {};
            
                \draw (0, 0) circle (3.2);
                \draw[blue, thick, domain = 107.189:4.056] plot ({3.2*cos(\x)}, {3.2*sin(\x)});
                \filldraw[color=blue, fill = white] (-0.946, 3.057) circle (4pt); % angle = 107.189
                \draw[blue] (-1.079, 3.487) node {$x$};
                \filldraw[color=black] (-0.946, 3.057) circle (2pt); % angle = 107.189
                \draw[black] (-3.7, 0) node {$y_i$};
                \filldraw[black] (-3.2,0) circle (4pt);
                \draw[black] (-0.679, 2.487) node {$y_{i+1}$};
                \filldraw[color=black] (2.983, 1.160) circle (4pt); % angle = 21.245
                \draw[black] (3.902, 1.323) node {$y_{i+2}$};
                \filldraw[color=blue] (3.192, 0.226) circle (4pt); % angle = 4.056
                \draw[blue] (4.241, 0.258) node {$G(x)$};
            \end{scope}
        \end{tikzpicture}
        
        \caption{All greedy steps contain a point from every optimal solution.}
        \label{fig:greedy_steps_contain_optimal_endpoints}
    \end{figure}

    Let $x \in [y_i, y_{i+1}]$, then $y_{i+1} \in [x, G(x)]$ by Lemma~\ref{lem:properties_of_greedy}.2. If $y_{i+1} \neq x$ we are done (Figure~\ref{fig:greedy_steps_contain_optimal_endpoints} left). If $y_{i+1} = x$ we must have $[x, y_{i+2}] \in \Fc$ which by Lemma~\ref{lem:properties_of_Fc}.3 implies that $[x, y_{i+2}] \subseteq [x, G(x)]$, and since $y_{i+1} \neq y_{i+2}$, we have $y_{i+2} \in (x, G (x)]$ (Figure~\ref{fig:greedy_steps_contain_optimal_endpoints} right).
\end{proof}

Iteratively computing greedy steps is exactly \RepeatedGreedy, leading to the following definition.

\begin{definition}[Greedy Sequence]
\label{def:greedy_sequence}
    Given a point $x$ on $\partial P$, we denote the sequence $(x_i)_{i = 0}^\infty$ with $x_0 = x$ and $x_{i+1} = \Gsc(x_{i})$ as the \emph{greedy sequence} starting at $x$. This is a sequence of endpoints of intervals.
    If none of the endpoints belong to an optimal solution then we call it a \emph{non-optimal} greedy sequence.
    Otherwise, it is an \emph{optimal} greedy sequence.
\end{definition}

Note that the greedy sequence $(x_i)_{i = 0}^\infty$ is infinite, whereas \RepeatedGreedy only computes a prefix of this sequence.
The goal of the rest of the paper is to show that for $T = \OhMega{poly(n)}$, then the sequence $(x_i)_{i = 0}^T$ is optimal.
Denote by \emph{revolution} a single traversal around $\partial P$ in $k + 1$ steps, where $k$ is the minimal number of greedy steps such that $\Gsc(x_k)=x_{k+1}\in[x_0,x_1)$.

Next, we show that the solution found after the first revolution is at most one interval longer than an optimal solution of size $\opt$ (similar to Theorem 2.4 in~\cite{circle_cover_arc_minimization_1984_LEE1984109}).

\begin{theorem}[One revolution is at most one-off optimal]
\label{thm:repeatedgreedy_is_a_+1_approximation}
    Running \RepeatedGreedy from any point $x \in \partial P$ and stopping once $x_{k+1} \in [x_0, x_1]$ guarantees a solution of size $k \leq \opt + 1$.
\end{theorem}

\begin{figure}[ht]
    \centering
    
    \begin{tikzpicture}[scale=0.75]
        \draw (-120:3.2) arc (-120:180:3.2);
        \draw[dotted, thick] (-120:3.2) arc (-120:-180:3.2);
        \filldraw[color=blue] (0.000, 3.200) circle (3pt); % angle = 90.000
        \draw[blue] (0.000, 3.650) node {$x_0$};
        \filldraw[color=blue] (2.693, 1.729) circle (3pt); % angle = 32.704
        \draw[blue] (3.071, 1.972) node {$x_1$};
        \filldraw[color=blue] (2.910, -1.332) circle (3pt); % angle = -24.592
        \draw[blue] (3.319, -1.519) node {$x_2$};
        \filldraw[color=blue] (0.452, -3.168) circle (3pt); % angle = -81.887
        \draw[blue] (0.515, -3.613) node {$x_3$};
        \filldraw[color=blue] (-3.069, 0.908) circle (3pt); % angle = -196.479
        \draw[blue] (-3.800, 1.035) node {$x_{k-1}$};
        \filldraw[color=blue] (-0.894, 3.073) circle (3pt); % angle = -253.775
        \draw[blue] (-1.020, 3.505) node {$x_{k}$};
        \filldraw[color=black] (1.246, 2.947) circle (3pt); % angle = 67.082
        \draw[black] (1.421, 3.362) node {$y_0$};
        \filldraw[color=black] (3.199, -0.093) circle (3pt); % angle = -1.673
        \draw[black] (3.648, -0.107) node {$y_1$};
        \filldraw[color=black] (2.386, -2.132) circle (3pt); % angle = -41.780
        \draw[black] (2.722, -2.432) node {$y_2$};
        \filldraw[color=black] (-1.487, 2.834) circle (3pt); % angle = -242.316
        \draw[black] (-1.796, 3.172) node {$y_{k-1}$};
    \end{tikzpicture}
    
    \caption{Each greedy interval $(x_i, x_{i + 1}]$ for $i = 0, 1, 2, ..., k - 1$ is disjoint and contains at least one point from any optimal solution by Corollary~\ref{cor:greedy_steps_contain_optimal_endpoints}. 
    This is not guaranteed for the last interval, i.e. the interval $(x_k, x_0]$ at the start of the next revolution which may not be a greedy interval.}
    \label{fig:one_optimal_endpoint_in_each_greedy_step}
\end{figure}

\begin{proof}
    Let $(x_i)_{i=0}^\infty$ be the greedy sequence starting at $x$. Assume that $k > 1$ is minimal such that $x_{k+1} \in [x_0, x_1]$. Thus this candidate solution uses $k+1$ guards.

    Let $y_0,\dots,y_{\opt-1}$ be an optimal solution. From Corollary~\ref{cor:greedy_steps_contain_optimal_endpoints} we know that for each $i = 0,1,\dots,\opt-1$: $y_i\in (x_i, x_{i+1}]$, up to renumbering of the indices; see Figure~\ref{fig:one_optimal_endpoint_in_each_greedy_step}.
    All these intervals are disjoint; hence an optimal solution contains at least $k-1$ endpoints, i.e. at least $k-1$ guards. Thus the greedy sequence yields a solution of size at most $\opt+1$
\end{proof}

Theorem~\ref{thm:repeatedgreedy_is_a_+1_approximation} is very powerful since we now only need to distinguish between the case where \RepeatedGreedy is optimal and the case where it uses one extra guard. We now consider how the greedy sequence behaves in different settings, first if $x_0$ is an endpoint of an optimal solution:

\begin{lemma}
\label{lem:One_point_is_enough}
    Let $x \in \partial P$ be a point in some optimal solution. Then \RepeatedGreedy returns an optimal solution in one revolution starting from $x$.
\end{lemma}

\begin{proof}
    Let $y_0, y_1, \dots, y_{\opt - 1}$ be an optimal solution and $x = y_0$. Consider the first $\opt$ terms of the greedy sequence starting at $x$: $x_0, x_1, \dots, x_{\opt - 1}$, see Figure~\ref{fig:x0_to_xk-1_is_optimal}.

    \begin{figure}[ht]
        \centering
        
        \begin{tikzpicture}[scale=0.65]
            \draw (-120:3.2) arc (-120:180:3.2);
            \draw[dotted, thick] (-120:3.2) arc (-120:-180:3.2);
            \filldraw[color=blue] (0.000, 3.200) circle (3pt); % angle = 90.000
            \draw[blue] (0.000, 3.750) node {$x_0$};
            \filldraw[color=blue] (2.693, 1.729) circle (3pt); % angle = 32.704
            \draw[blue] (3.071, 1.972) node {$x_1$};
            \filldraw[color=blue] (2.910, -1.332) circle (3pt); % angle = -24.592
            \draw[blue] (3.319, -1.519) node {$x_2$};
            \filldraw[color=blue] (0.452, -3.168) circle (3pt); % angle = -81.887
            \draw[blue] (0.515, -3.613) node {$x_3$};
            \filldraw[color=blue] (-3.069, 0.908) circle (3pt); % angle = -196.479
            \draw[blue] (-3.900, 1.035) node {$x_{\opt - 2}$};
            \filldraw[color=blue] (-0.894, 3.073) circle (3pt); % angle = -253.775
            \draw[blue] (-1.220, 3.555) node {$x_{\opt - 1}$};
            \filldraw[color=black] (0.000, 3.200) circle (1.5pt); % angle = 90.000
            \draw[black] (0.000, 2.450) node {$y_0$};
            \filldraw[color=black] (2.296, 2.229) circle (3pt); % angle = 44.163
            \draw[black] (1.973, 1.916) node {$y_1$};
            \filldraw[color=black] (3.173, -0.412) circle (3pt); % angle = -7.403
            \draw[black] (2.727, -0.354) node {$y_2$};
            \filldraw[color=black] (1.650, -2.742) circle (3pt); % angle = -58.969
            \draw[black] (1.418, -2.356) node {$y_3$};
            \filldraw[color=black] (-1.487, 2.834) circle (3pt); % angle = -242.316
            \draw[black] (-1.278, 2.135) node {$y_{\opt - 1}$};
            
            \draw[black, thick, domain =  90.000-360:-242.316] plot ({3*cos(\x)}, {3*sin(\x)});
            \draw[blue, thick, domain =  90.000-360:-253.775] plot ({3.4*cos(\x)}, {3.4*sin(\x)});
        \end{tikzpicture}
        
        \caption{An optimal solution $y_0, y_1, y_2, \dots, y_{\opt - 1}$ and the greedy sequence $x_0, x_1, x_2, \dots, x_{\opt - 1}$ starting from $x_0 = y_0$, that is also optimal since $\I{x_{\opt - 1}}{x_0} \subseteq \I{y_{\opt - 1}}{y_0}$ and $\I{y_{\opt - 1}}{y_0} \in \Fc$.}
        \label{fig:x0_to_xk-1_is_optimal}
    \end{figure}

    Since we know the length is optimal and the intervals $[x_0, x_1],\dots,[x_{\opt - 2}, x_{\opt - 1}]$ are visible, it is sufficient to show that $[x_{\opt - 1}, x_0]$ is visible.

    We have $x_0 = y_0$ by assumption and using Corollary~\ref{cor:greedy_steps_contain_optimal_endpoints} we know that each of the intervals $(x_{i - 1}, x_i]$ contains a $y_j$ for each $i = 1, 2, \dots, \opt - 2$. 
    All these intervals are disjoint, and hence $y_i \in (x_{i - 1}, x_i]$ for $i = 1, 2, \dots, \opt - 2$. 
    Since $y_0, \dots, y_{\opt - 1}$ is an optimal solution we have $[y_{\opt - 1}, y_0] \in \Fc$. 
    But since $y_{\opt - 1} \in (x_{\opt - 2}, x_{\opt - 1}]$ we must have $x_{\opt - 1} \in [y_{\opt - 1}, y_0]$ and $[x_{\opt - 1}, x_0] \subseteq [y_{\opt - 1}, y_0] \in \Fc$, implying that $[x_{\opt - 1}, x_0] \in \Fc$ and $x_0, \dots, x_{\opt - 1}$ is an optimal solution.
\end{proof}

It follows immediately from Lemma~\ref{lem:One_point_is_enough} that once a greedy sequence reaches a point from an optimal solution, it will remain optimal from that point on.

\begin{corollary}[Once optimal, always optimal]
\label{cor:Once_optimal_always_optimal}
    Let $x \in \partial P$ and consider the greedy sequence starting at $x$. Assuming that there is some $i$ such that $x_i, x_{i + 1}, \dots, x_{i + \opt - 1}$ is an optimal solution then $x_{i + j}, \dots, x_{i + \opt - 1 + j}$ will be an optimal solution for all $j \geq 0$.
\end{corollary}

\begin{proof}
    The proof follows by induction in $j$. For $j = 0$ the claim is trivial.

    For the induction step, we assume $x_{i + j}, \dots, x_{i + \opt - 1 + j}$ is optimal. Thus $x_{i + j + 1}$ is an endpoint for an optimal solution, thus by Lemma~\ref{lem:One_point_is_enough} the greedy solution starting at $x_{i + j + 1}$ is optimal, which shows that the solution $x_{i + j + 1}, \dots, x_{i + \opt - 1 + j + 1}$ is optimal.
\end{proof}

We are now ready to determine when a solution is optimal.
The first condition we establish is if the greedy sequence repeats in the first revolution:

\begin{lemma}[Exact cover is optimal]
\label{lem:repetition_implies_optimality}
    Let $x\in \partial P$, then if $x_i, x_{i+1}, \dots, x_{i+k}, x_i$ appear in the greedy sequence during a single revolution, then $x_i,\dots,$ $x_{i+k}$ is an optimal solution.
\end{lemma}

\begin{proof}
    Let $y_0,y_1,\dots,y_{\opt-1}$ be an optimal solution. 
    Each of the intervals $(x_{i+j}, x_{i+j+1}]$ (for $j = 0,\dots, k-1$) and $(x_{i+k},x_i]$ must contain an endpoint from an optimal solution by Corollary~\ref{cor:greedy_steps_contain_optimal_endpoints}. 
    Thus, $\opt = k+1$ and $x_i, x_{i+1},\dots,x_{i+k}$ is optimal.
\end{proof}

\begin{remark}
    The first $(x_0, x_1]$ and last step $(x_{k}, x_{k+1}]$ in a revolution of $\partial P$ can overlap, hence an optimal point $y_0$ can satisfy the condition of Corollary~\ref{cor:greedy_steps_contain_optimal_endpoints} for both resulting in an $\opt + 1$ approximation. 
    This is impossible when the greedy sequence repeats, see Figure~\ref{fig:greedy_steps_non_periodic_periodic}.
    
    \begin{figure}[ht]
        \centering

        \newcommand{\arc}[3]{
            \draw[-latex, blue, thick] (#2 : #1) arc (#2 : #3 : #1);
        }
        
        \begin{tikzpicture}[scale=0.75]
            \begin{scope}
                \node at (-4.5, 0) {};
                
                \draw (-120:3.2) arc (-120:180:3.2);
                \draw[dotted, thick] (-120:3.2) arc (-120:-180:3.2);
                \filldraw[color=blue] (0.000, 3.200) circle (3pt); % angle = 90.000
                \draw[blue] (0.000, 3.80) node {$x_0$};
                \filldraw[color=blue] (2.693, 1.729) circle (3pt); % angle = 32.704
                \draw[blue] (3.121, 2.022) node {$x_1$};
                \filldraw[color=blue] (2.910, -1.332) circle (3pt); % angle = -24.592
                \draw[blue] (3.369, -1.569) node {$x_2$};
                \filldraw[color=blue] (0.452, -3.168) circle (3pt); % angle = -81.887
                \draw[blue] (0.515, -3.613) node {$x_3$};
                \filldraw[color=blue] (-3.069, 0.908) circle (3pt); % angle = -196.479
                \draw[blue] (-3.800, 0.835) node {$x_{k-1}$};
                \filldraw[color=blue] (-0.894, 3.073) circle (3pt); % angle = -253.775
                \draw[blue] (-1.120, 3.605) node {$x_{k}$};
                \filldraw[color=blue] (2.102, 2.412) circle (3pt); % angle = -311.070
                \draw[blue] (2.698, 2.952) node {$x_{k+1}$};
                \filldraw[color=black] (0.946, 3.057) circle (3pt); % angle = 72.811
                \draw[black] (0.813, 2.627) node {$y_0$};
                \filldraw[color=black] (3.153, 0.544) circle (3pt); % angle = 9.786
                \draw[black] (2.710, 0.467) node {$y_1$};
                \filldraw[color=black] (2.386, -2.132) circle (3pt); % angle = -41.780
                \draw[black] (2.051, -1.832) node {$y_2$};
                \filldraw[color=black] (-0.187, -3.195) circle (3pt); % angle = -93.346
                \draw[black] (-0.161, -2.745) node {$y_3$};
                \filldraw[color=black] (-1.487, 2.834) circle (3pt); % angle = -242.316
                \draw[black] (-1.208, 2.235) node {$y_{\opt-1}$};

                \arc{3.4}{90.000}{32.704}
                \arc{3.4}{32.704}{-24.592}
                \arc{3.4}{-24.592}{-81.887}
                \arc{3.55}{-196.479}{-253.775}
                \arc{3.55}{-253.775}{-311.070}
                
                % \draw[blue, thick, domain =  32.704:90.000] plot ({3*cos(\x)}, {3*sin(\x)});
                % \draw[blue, thick, domain =  32.704:-24.592] plot ({3.4*cos(\x)}, {3.4*sin(\x)});
                % \draw[blue, thick, domain =  -81.887:-24.592] plot ({3*cos(\x)}, {3*sin(\x)});
                % \draw[blue, thick, domain = -196.479:-253.775] plot ({3*cos(\x)}, {3*sin(\x)});
                % \draw[blue, thick, domain =  -311.070: -253.775] plot ({3.4*cos(\x)}, {3.4*sin(\x)});
            \end{scope}
    
            \begin{scope}[xshift=9.3cm]
                \node at (4.5, 0) {};
                
                \draw (-120:3.2) arc (-120:180:3.2);
                \draw[dotted, thick] (-120:3.2) arc (-120:-180:3.2);
                \filldraw[color=blue] (0.000, 3.200) circle (3pt); % angle = 90.000
                \draw[blue] (0.000, 3.650) node {$x_0$};
                \filldraw[color=blue] (2.507, 1.989) circle (3pt); % angle = 38.434
                \draw[blue] (2.909, 2.319) node {$x_1$};
                \filldraw[color=blue] (3.116, -0.727) circle (3pt); % angle = -13.132
                \draw[blue] (3.655, -0.829) node {$x_2$};
                \filldraw[color=blue] (1.368, -2.893) circle (3pt); % angle = -64.699
                \draw[blue] (1.560, -3.300) node {$x_3$};
                \filldraw[color=blue] (-2.473, 2.031) circle (3pt); % angle = -219.397
                \draw[blue] (-3.121, 2.317) node {$x_{k}$};
                \filldraw[color=black] (0.946, 3.057) circle (3pt); % angle = 72.811
                \draw[black] (0.813, 2.627) node {$y_0$};
                \filldraw[color=black] (3.153, 0.544) circle (3pt); % angle = 9.786
                \draw[black] (2.710, 0.467) node {$y_1$};
                \filldraw[color=black] (2.386, -2.132) circle (3pt); % angle = -41.780
                \draw[black] (2.051, -1.832) node {$y_2$};
                \filldraw[color=black] (-0.187, -3.195) circle (3pt); % angle = -93.346
                \draw[black] (-0.161, -2.745) node {$y_3$};
                \filldraw[color=black] (-1.487, 2.834) circle (3pt); % angle = -242.316
                \draw[black] (-1.278, 2.435) node {$y_{\opt-1}$};
                
                % \draw[blue, thick, domain = 38.434: 90.000] plot ({3*cos(\x)}, {3*sin(\x)});
                % \draw[blue, thick, domain = 38.434: -13.132] plot ({3.4*cos(\x)}, {3.4*sin(\x)});
                % \draw[blue, thick, domain = -64.699:  -13.132] plot ({3*cos(\x)}, {3*sin(\x)});
                % \draw[blue, thick, domain = 360-219.397: 90.000] plot ({3.4*cos(\x)}, {3.4*sin(\x)});

                \arc{3.4}{90.000}{38.434}
                \arc{3.4}{38.434}{-13.132}
                \arc{3.4}{-13.132}{-64.699}
                \arc{3.4}{-219.397}{-270.000}
            \end{scope}
        \end{tikzpicture}
        
        \caption{Left, a greedy sequence without repetitions where $y_0$ appears in both $(x_0,x_1]$ and $(x_{k-1},x_k]$ which explains how a greedy revolution might not be optimal.
        Right, a greedy sequence that repeats after each revolution. 
        Greedy intervals are visualized as directed circle arches.}
        \label{fig:greedy_steps_non_periodic_periodic}
    \end{figure}
\end{remark}

We now know that if the greedy sequence repeats after a single revolution of the polygon, we have an optimal solution.
However, to determine that a solution is optimal, we will need weaker conditions than repeating after just one revolution. To do this, we investigate how a non-optimal greedy sequence traverses $\partial P$, in relation to an optimal solution, see Figure~\ref{fig:placing_optimal_solution}. 

\begin{proposition}[Behavior of non-optimal greedy sequences]
\label{prop:behavior_of_non-optimal_greedy_sequences}
    Let $x \in \partial P$ and consider the non-optimal greedy sequence $(x_i)_{i=0}^{k \cdot N}$ starting at $x$. 
    Let $y_0, y_1,\dots, y_{\opt-1}$ be an optimal solution with $y_0 \in (x_0, x_1]$.
    Then $x_{ik + m} \in [y_{m-1}, x_{(i-1)k + m})$ for all $m \in \{1,\dots,k\}$ and $i \in \{1,2,\dots,N-1\}$.
\end{proposition}

\begin{proof}
    Like in the proof of Theorem~\ref{thm:repeatedgreedy_is_a_+1_approximation}, we distribute the points of an optimal solution $y_i$ in the first $k$ greedy steps. We can do this since the greedy sequence is not optimal implying that the first $k$ greedy steps do not complete the first revolution. 
    Theorem~\ref{thm:repeatedgreedy_is_a_+1_approximation} implies that $x_{k + 1} \in \I{x_0}{x_1}$.
    We have $y_i \in (x_i, x_{i+1}]$ for $i = 0,\dots, k-1$. 
    Now $x_{k+1} = \Gsc(x_{k})$ is the first element of the greedy sequence to lie in $[x_0, x_1]$.
    More specifically, we must have $x_{k+1} \in [y_0, x_1]$, because if not then $x_{k+1} \in [x_0, y_0)$ and $[y_{k-1}, y_0] \nsubseteq (x_{k-1},x_{k}]$ and since $[y_{k-1},y_0]$ is visible, we contradict Lemma~\ref{lem:properties_of_Fc}.3.
    
    We now show $x_{ik + m} \in [y_{m-1}, x_{(i-1)k + m}]$ for $i \leq N$ and $m \in \{1,\dots, k\}$ by induction in $ik+m$. The induction start is exactly $x_{k+1} \in [y_0, x_1]$. For the induction step, assume that $x_{ik + m} \in [y_{m-1}, x_{(i-1)k + m}]$.
    If $m = k$ we will have a carry when we take a step (i.e. increment $i$ by $1$ and set $m = 1$), otherwise, we will simply add one to $m$. We assume we are in the second case to ease notation, but the same argument holds in the carry case.
    
    It is clear, that $x_{(i-1)k + m} \in [y_{m-1}, y_m]$, combining this with $x_{ik + m} \in [y_{m-1}, x_{(i-1)k + m}]$ and Lemma~\ref{lem:properties_of_greedy}.3 we get $\Gsc(x_{ik + m}) \in [y_m, \Gsc(x_{(i-1)k + m})]$, i.e. $x_{ik + m+1} \in [y_m, x_{(i-1)k + m+1}]$, see Figure~\ref{fig:placing_optimal_solution}.
    
    Finally, since we assumed that none of the candidate solutions are optimal within $kN$ greedy steps, the sequence will not repeat in one revolution as this would contradict Lemma~\ref{lem:repetition_implies_optimality}, thus $x_{k i + m} \in [y_{m-1}, x_{k(i-1) + m})$ for all relevant $i < N$ and $m\in\{1,2,\dots, k\}$. 
\end{proof}

\begin{figure}[ht]
    \centering
    
    \begin{tikzpicture}[scale=0.75]
        \draw (-120:3.2) arc (-120:180:3.2);
        \draw[dotted, thick] (-120:3.2) arc (-120:-180:3.2);
        \filldraw[color=blue] (0.000, 3.200) circle (3pt); % angle = 90.000
        \draw[blue] (0.000, 3.650) node {$x_0$};
        \filldraw[color=blue] (2.693, 1.729) circle (3pt); % angle = 32.704
        \draw[blue] (3.200, 1.700) node {$x_1$};
        \filldraw[color=blue] (2.910, -1.332) circle (3pt); % angle = -24.592
        \draw[blue] (3.319, -1.519) node {$x_2$};
        \filldraw[color=blue] (0.452, -3.168) circle (3pt); % angle = -81.887
        \draw[blue] (0.515, -3.613) node {$x_3$};
        \filldraw[color=blue] (-3.069, 0.908) circle (3pt); % angle = -196.479
        \draw[blue] (-3.800, 1.035) node {$x_{k-1}$};
        \filldraw[color=blue] (-0.894, 3.073) circle (3pt); % angle = -253.775
        \draw[blue] (-1.020, 3.505) node {$x_{k}$};
        \filldraw[color=blue] (2.296, 2.229) circle (3pt); % angle = 44.163
        \draw[blue] (3.000, 2.300) node {$x_{k+1}$};
        \filldraw[color=blue] (1.937, 2.547) circle (3pt); % angle = 52.758
        \draw[blue] (2.600, 2.750) node {$x_{2k+1}$};
        \filldraw[color=blue] (1.534, 2.808) circle (3pt); % angle = 61.352
        \draw[blue] (1.900, 3.200) node {$x_{3k+1}$};
        \filldraw[color=blue] (3.028, -1.035) circle (3pt); % angle = -18.862
        \draw[blue] (3.694, -1.100) node {$x_{k+2}$};
        \filldraw[color=blue] (3.149, -0.570) circle (3pt); % angle = -10.268
        \draw[blue] (3.942, -0.651) node {$x_{2k+2}$};
        \filldraw[color=blue] (1.072, -3.015) circle (3pt); % angle = -70.428
        \draw[blue] (1.600, -3.439) node {$x_{k+3}$};
        \filldraw[color=blue] (1.650, -2.742) circle (3pt); % angle = -58.969
        \draw[blue] (2.400, -3.000) node {$x_{2k+3}$};
        \filldraw[color=blue] (-1.762, 2.671) circle (3pt); % angle = -236.586
        \draw[blue] (-2.010, 3.047) node {$x_{2k}$};
        \filldraw[color=blue] (-2.258, 2.268) circle (3pt); % angle = -225.127
        \draw[blue] (-2.575, 2.587) node {$x_{3k}$};
        \filldraw[color=black] (0.636, 3.136) circle (3pt); % angle = 78.541
        \draw[black] (0.546, 2.495) node {$y_0$};
        \filldraw[color=black] (3.083, 0.856) circle (3pt); % angle = 15.515
        \draw[black] (2.450, 0.736) node {$y_1$};
        \filldraw[color=black] (2.386, -2.132) circle (3pt); % angle = -41.780
        \draw[black] (2.051, -1.832) node {$y_2$};
        \filldraw[color=black] (-0.187, -3.195) circle (3pt); % angle = -93.346
        \draw[black] (-0.161, -2.745) node {$y_3$};
        \filldraw[color=black] (-2.827, 1.499) circle (3pt); % angle = -207.938
        \draw[black] (-2.230, 1.288) node {$y_{k-1}$};

        \draw[green, thick, domain = 78.541: 52.758] plot ({2.9*cos(\x)}, {2.9*sin(\x)});
        \draw[orange, thick, domain = 15.515: -10.268] plot ({2.9*cos(\x)}, {2.9*sin(\x)});

        \draw[black, thick, domain = -190.08:-102.06, |->] plot ({1.93*cos(\x)+3}, {1.93*sin(\x)+2}); 
        \draw[black] (1,0.9) node {$\Gsc$};
    \end{tikzpicture}
    
    \caption{For a non-optimal greedy sequence $\Gsc$  maps points from $[y_0,x_{2k+1}]$ (green) to $[y_1,x_{2k+2}]$ (orange), and in the next revolution from  $[y_0,x_{3k+1}]$ to $[y_1,x_{3k + 2}]$.}
    \label{fig:placing_optimal_solution}
\end{figure}

\begin{remark}
\label{rem:need_only_look_at_[x_m-1,x_(i-1)k+m]}
    It is clear that $x_{ik + m}\in[y_{m-1}, x_{(i-1)k+m})$ implies $x_{ik+m} \in [x_{m-1}, x_{(i-1)k + m})$. This is relevant in practice, as we do not know the location of the $y_i$'s. Thus, the result of Proposition~\ref{prop:behavior_of_non-optimal_greedy_sequences} can be restated as $x_{ik+m} \in [x_{m-1}, x_{(i-1)k+m})$ (these intervals are similar to the zones defined in~\cite{circle_cover_arc_minimization_1984_LEE1984109}).
\end{remark}

Proposition~\ref{prop:behavior_of_non-optimal_greedy_sequences} shows that there is an important structure in the greedy sequence when viewed locally, which we capture in the following definitions:

\begin{definition}[Local Greedy Sequence]
\label{def:local_greedy_sequence}
    Let $x \in \partial P$, consider the greedy sequence $(x_i)_{i = 0}^\infty$ starting at $x$. Then $(x_i^m)_{i = 0}^\infty = (x_{ik + m})_{i = 0}^\infty$ for $m = 1,\dots,k$ are \emph{local greedy sequences}.
\end{definition}

For non-optimal greedy sequences, its local greedy sequences move counterclockwise around $\partial P$ by Lemma~\ref{prop:behavior_of_non-optimal_greedy_sequences}.
We make this precise in the following definition.

\begin{definition}[Fingerprint]
\label{def:fingerprint}
    Let $(x_i^m)_{i=0}^\infty$ be the $m$'th local greedy sequence. For $i\geq 1$ we say that $x_i^m$ has a \emph{negative fingerprint} if $x_{i}^m \in [x_{m-1}, x_{i-1}^m)$. Otherwise, we say that $x_i^m$ has a \emph{positive fingerprint}.
\end{definition}

With the notion of fingerprints, we introduce \emph{combinatorial optimality conditions}, which if satisfied by a greedy sequence guarantees that it is optimal, that is, it contains an optimal solution. The contrapositive of Proposition~\ref{prop:behavior_of_non-optimal_greedy_sequences} gives the first such condition: A local greedy sequence element with a positive fingerprint is optimal.

\begin{corollary}[Positive fingerprint implies optimality]
\label{cor:positive_fingerprint_is_optimal}
    Let $x \in \partial P$. Assume that $x_i^m$ has a positive fingerprint. Then $x_i^m, x_i^{m+1}, x_i^{m+2}, \dots, x_{i+1}^{m-1}$ is an optimal solution.
\end{corollary}

\begin{proof}
    Proposition~\ref{prop:behavior_of_non-optimal_greedy_sequences} states that if the sequence $x_0,\dots,x_{ik+m}$ has no subsequence that constitutes an optimal solution, then all individual points have negative fingerprints. The contrapositive of this statement is that if a point $x_j$ in the sequence has a positive fingerprint, then it is a point in an optimal solution, which in turn, by Lemma~\ref{lem:One_point_is_enough} implies that $x_i^m, \dots, x_{i+1}^{m-1}$ is an optimal solution.
\end{proof}

The second combinatorial optimality condition states that if the local greedy sequence has moved out of the interval of the first revolution, then it is optimal.

\begin{corollary}[Escaping an interval implies optimality]
\label{cor:escaping_an_interval_implies_optimal}
    Let $x \in \partial P$. If an element $x_i^m$ of
    a local greedy sequence
    is not in $[x_{m-1}, x_m)$, then $x_i^m$ is the starting point of an optimal solution.
\end{corollary}

\begin{proof}
    As $[x_{m-1}, x_m) \subseteq [x_{m-1}, x_{ik + m})$ for all $i$, this is a direct consequence of Remark~\ref{rem:need_only_look_at_[x_m-1,x_(i-1)k+m]}.
\end{proof}

The third combinatorial optimality condition is perhaps the strongest; it states that if the greedy sequence has a periodic subsequence of any length, i.e., contains a repetition, then it is optimal.

\begin{corollary}[Periodicity implies optimality]
\label{cor:periodicity_implies_optimal}
    Let $x \in \partial P$ and $(x_i)_{i=0}^\infty$ be the greedy sequence starting at $x$. Assume there exists two indices $i < j$, such that $x_i = x_j$. Then $x_j, x_{j + 1}, \dots, x_{j + \opt - 1}$ is an optimal solution.
\end{corollary}

\begin{proof}
    Assume for contradiction that none of $x_0, \dots, x_j$ are optimal. Let for each $\ell > 0$
    \[ V_\ell = \bigcup_{m = 1}^{k} [y_{m - 1}, x_{\ell - 1}^m) . \]
    
    After $\ell k$ greedy steps, all future points in the greedy sequence will be contained in $V_\ell$ by Proposition~\ref{prop:behavior_of_non-optimal_greedy_sequences}, and since $V_\ell \cap \{x_t\mid t = 0,\dots,\ell k\} = \emptyset$, a point in the greedy sequence cannot be repeated. 
    This is a contradiction, so one of the $x_t$ must be an endpoint of an optimal solution. So by Corollary~\ref{cor:Once_optimal_always_optimal} $x_j, x_{j + 1}, \dots, x_{j + \opt - 1}$ is an optimal solution.
\end{proof}

We have established conditions in Corollary~\ref{cor:positive_fingerprint_is_optimal}, \ref{cor:escaping_an_interval_implies_optimal} and \ref{cor:periodicity_implies_optimal} that characterize when we can guarantee that the greedy sequence is optimal.
However, using only combinatorial insights, we still cannot distinguish between optimal and non-optimal greedy sequences (see Example~\ref{exa:two_greedy_sequences} in Appendix~\ref{apx:two_greedy_sequences}) nor guarantee that running \RepeatedGreedy long enough yields an optimal solution (See Figure~\ref{fig:greedy_convergence_figure}). Thus, we return to the geometric setting with Corollary~\ref{cor:positive_fingerprint_is_optimal}, \ref{cor:escaping_an_interval_implies_optimal} and \ref{cor:periodicity_implies_optimal} as optimality conditions.

\begin{figure}[ht]
    \centering

    \begin{tikzpicture}[scale=0.75]
        \draw (0, 0) circle (3.2);
                
            \draw[blue] (3.471, 1.128) node {$x_{0}$};
            \filldraw[color=blue] (1.362, -2.895) circle (3pt); % angle = -64.800
            \draw[blue] (1.554, -3.303) node {$x_{1}$};
            \filldraw[color=blue] (-2.723, -1.681) circle (3pt); % angle = -148.320
            \draw[blue] (-3.106, -1.917) node {$x_{2}$};
            \filldraw[color=blue] (-1.949, 2.538) circle (3pt); % angle = -232.488
            \draw[blue] (-2.223, 2.895) node {$x_{3}$};
            \filldraw[color=blue] (2.349, 2.173) circle (3pt); % angle = 42.761
            \draw[blue] (2.680, 2.478) node {$x_{4}$};
            \filldraw[color=blue] (2.359, -2.163) circle (3pt); % angle = -42.515
            \draw[blue] (2.690, -2.467) node {$x_{5}$};
            \filldraw[color=blue] (-1.982, -2.513) circle (3pt); % angle = -128.264
            \draw[blue] (-2.260, -2.866) node {$x_{6}$};
            \filldraw[color=blue] (-2.639, 1.810) circle (3pt); % angle = -214.437
            \draw[blue] (-3.010, 2.064) node {$x_{7}$};
            \filldraw[color=blue] (1.648, 2.743) circle (3pt); % angle = 59.006
            \draw[blue] (1.880, 3.129) node {$x_{8}$};
            \filldraw[color=blue] (2.828, -1.497) circle (3pt); % angle = -27.894
            \draw[blue] (3.226, -1.708) node {$x_{9}$};
            \filldraw[color=blue] (-1.358, -2.898) circle (3pt); % angle = -115.105
            \draw[blue] (-1.649, -3.305) node {$x_{10}$};
            \filldraw[color=blue] (-2.954, 1.229) circle (3pt); % angle = -202.594
            \draw[blue] (-3.470, 1.402) node {$x_{11}$};
            \filldraw[color=blue] (1.112, 3.001) circle (3pt); % angle = 69.665
            \draw[blue] (1.268, 3.423) node {$x_{12}$};
            \filldraw[color=blue] (3.038, -1.005) circle (3pt); % angle = -18.301
            \draw[blue] (3.565, -1.146) node {$x_{13}$};
            \filldraw[color=blue] (-0.907, -3.069) circle (3pt); % angle = -106.471
            \draw[blue] (-1.035, -3.500) node {$x_{14}$};
            \filldraw[color=blue] (-3.093, 0.819) circle (3pt); % angle = -194.824
            \draw[blue] (-3.629, 0.934) node {$x_{15}$};
            \filldraw[color=blue] (0.738, 3.114) circle (3pt); % angle = 76.658
            \draw[blue] (0.542, 3.551) node {$x_{16}$};
            \filldraw[color=blue] (3.130, -0.666) circle (3pt); % angle = -12.008
            \draw[blue] (3.670, -0.759) node {$x_{17}$};
            \filldraw[color=blue] (-0.600, -3.143) circle (3pt); % angle = -100.807
            \draw[blue] (-0.384, -3.585) node {$x_{18}$};
            \filldraw[color=blue] (-3.154, 0.541) circle (3pt); % angle = -189.726
            \draw[blue] (-3.698, 0.567) node {$x_{19}$};
            \filldraw[color=blue] (0.487, 3.163) circle (3pt); % angle = 81.246
            \filldraw[color=blue] (3.170, -0.439) circle (3pt); % angle = -7.878
            \filldraw[color=blue] (-0.395, -3.176) circle (3pt); % angle = -97.090
            \filldraw[color=blue] (-3.180, 0.356) circle (3pt); % angle = -186.381
            \filldraw[color=blue] (0.320, 3.184) circle (3pt); % angle = 84.257
            \filldraw[color=blue] (3.187, -0.288) circle (3pt); % angle = -5.169
            \filldraw[color=blue] (-0.260, -3.189) circle (3pt); % angle = -94.652
            \filldraw[color=blue] (-3.191, 0.234) circle (3pt); % angle = -184.187
            \filldraw[color=blue] (0.210, 3.193) circle (3pt); % angle = 86.232
            \filldraw[color=blue] (3.194, -0.189) circle (3pt); % angle = -3.391
            \filldraw[color=blue] (-0.170, -3.195) circle (3pt); % angle = -93.052
            \filldraw[color=blue] (-3.196, 0.153) circle (3pt); % angle = -182.747
            \filldraw[color=blue] (0.138, 3.197) circle (3pt); % angle = 87.528
            \filldraw[color=blue] (3.198, -0.124) circle (3pt); % angle = -2.225
            \filldraw[color=blue] (-0.112, -3.198) circle (3pt); % angle = -92.003
            \filldraw[color=blue] (-3.198, 0.101) circle (3pt); % angle = -181.802
            \filldraw[color=blue] (0.091, 3.199) circle (3pt); % angle = 88.378
            \filldraw[color=blue] (3.199, -0.082) circle (3pt); % angle = -1.460
            \filldraw[color=blue] (-0.073, -3.199) circle (3pt); % angle = -91.314
            \filldraw[color=blue] (-3.199, 0.066) circle (3pt); % angle = -181.182
            \filldraw[color=blue] (0.059, 3.199) circle (3pt); % angle = 88.936
            \filldraw[color=blue] (3.200, -0.053) circle (3pt); % angle = -0.958
            \filldraw[color=blue] (-0.048, -3.200) circle (3pt); % angle = -90.862
            \filldraw[color=blue] (-3.200, 0.043) circle (3pt); % angle = -180.776
            \filldraw[color=blue] (0.039, 3.200) circle (3pt); % angle = 89.302
            \filldraw[color=blue] (3.200, -0.035) circle (3pt); % angle = -0.628
            \filldraw[color=blue] (-0.032, -3.200) circle (3pt); % angle = -90.566
            \filldraw[color=blue] (-3.200, 0.028) circle (3pt); % angle = -180.509
            \filldraw[color=blue] (0.026, 3.200) circle (3pt); % angle = 89.542
            \filldraw[color=blue] (3.200, -0.023) circle (3pt); % angle = -0.412
            \filldraw[color=blue] (-0.021, -3.200) circle (3pt); % angle = -90.371
            \filldraw[color=blue] (-3.200, 0.019) circle (3pt); % angle = -180.334
            \filldraw[color=blue] (0.017, 3.200) circle (3pt); % angle = 89.699
            \filldraw[color=blue] (3.200, -0.015) circle (3pt); % angle = -0.271
            \filldraw[color=blue] (-0.014, -3.200) circle (3pt); % angle = -90.243
            \filldraw[color=blue] (-3.200, 0.012) circle (3pt); % angle = -180.219
            \filldraw[color=blue] (0.011, 3.200) circle (3pt); % angle = 89.803
            \filldraw[color=blue] (3.200, -0.010) circle (3pt); % angle = -0.177
            \filldraw[color=blue] (-0.009, -3.200) circle (3pt); % angle = -90.160
            \filldraw[color=blue] (-3.200, 0.008) circle (3pt); % angle = -180.144
        
            \filldraw[color=black] (0.000, 3.200) circle (3pt); % angle = 90.000
            \draw[black] (0.000, 2.750) node {$y_0$};
            \filldraw[color=black] (3.200, 0.000) circle (3pt); % angle = 0.000
            \draw[black] (2.750, 0.000) node {$y_1$};
            \filldraw[color=black] (0.000, -3.200) circle (3pt); % angle = -90.000
            \draw[black] (0.000, -2.750) node {$y_2$};
            \filldraw[color=black] (-3.200, -0.000) circle (3pt); % angle = -180.000
            \draw[black] (-2.750, -0.000) node {$y_3$};
            \filldraw[color=blue] (3.043, 0.989) circle (3pt); % angle = 18.000

    \end{tikzpicture}
    
    \caption{Visualization of a greedy sequence, that never reaches a combinatorial optimality condition since the local greedy sequences converge to their respective $y$'s but never reach them.}
    \label{fig:greedy_convergence_figure}
\end{figure}

% --------------------------------------------------------------------
\section{\RepeatedGreedy, a geometric viewpoint}
\label{sec:the_greedy_algorithm_geometry}
% --------------------------------------------------------------------

In this section, we explore the behavior of Algorithm~\ref{alg:greedy} in the geometric setting.
We observe the behavior of the (local) greedy sequences by adapting the combinatorial optimality conditions Corollaries~\ref{cor:positive_fingerprint_is_optimal}, \ref{cor:escaping_an_interval_implies_optimal} and \ref{cor:periodicity_implies_optimal} to include geometric observations and finally show that the sequence contains an optimal solution within a polynomial number of revolutions.

First, Corollary~\ref{cor:positive_fingerprint_is_optimal}~and~\ref{cor:escaping_an_interval_implies_optimal} can be related in the geometric to when local greedy sequences move onto new edges:

\begin{definition}[Edge jumps]
    Let $(x_i^m)_{i = 0}^\infty$ be a local greedy sequence. We say that $(x_i^m)_{i = 0}^\infty$ is an \emph{edge jump} at $i$ if $x_i^m$ is a vertex of $P$ or $x_i^m$ and $x_{i+1}^m$ are in different edges of $P$ and neither is a vertex of $P$.
\end{definition}

\begin{lemma}[Maximal edge jumping]
\label{lem:maximal_edge_jumping}
    Let $(x_i^1)_{i =0}^\infty, \dots, (x_i^k)_{i=0}^\infty$ be the local greedy sequences. 
    If there are at least $n + 1$ edge jumps across all the local greedy sequences then some local greedy sequence $(x_i^m)_{i=0}^\infty$ will lie outside $[x_{m-1}, x_m]$ after these edge jumps.
\end{lemma}

\begin{proof}
    Let $n_m$ be the number of vertices of $P$ in $(x_{m-1},x_m)$ and $n_{k+1}$ the number of vertices in $(x_k,x_0)$. Then $\sum_{m = 1}^k n_m \leq \sum_{m=1}^{k+1} n_m \leq n$. 
    If the local greedy sequence $(x_i^m)_{i = 0}^\infty$ jumps $n_m + 1$ times it will at some point have passed $x_{m-1}$ and thus one of the points must lie in $[x_{m-2}, x_{m-1})$. 
    The lemma follows by the pigeonhole principle.
\end{proof}

This leads to a new progress condition based on edge jumps:

\begin{corollary}[Many edge jumps implies optimal] \label{cor:too_many_edge_jumps_will_force_greedy_sequence_to_be_optimal}
     Let $(x_i)_{i=0}^N$ be a greedy sequence whose local greedy subsequences $n + 1$ edge jumps. 
     Then the revolution after $x_N$ will be optimal.
\end{corollary}

\begin{proof}
    The proof is immediate from Lemma~\ref{lem:maximal_edge_jumping} and Corollary~\ref{cor:escaping_an_interval_implies_optimal}.
\end{proof}

\begin{remark}[Geometric progress conditions]
\label{rem:geometric_progresss}
    We establish three progress conditions on the basis of our combinatoric optimality conditions and extend them with geometric results. If one of these are satisfied, we discover an optimal solution or make progress towards one.
    \begin{enumerate}
    \item If we see a positive finger print, the greedy sequence is optimal (Corollary~\ref{cor:positive_fingerprint_is_optimal}).
    \item If we see a repetition in the greedy sequence, the sequence is optimal (Corollary~\ref{cor:periodicity_implies_optimal}).
    \item If we see $n+1$ edge jumps, the greedy sequence is optimal (Corollary~\ref{cor:too_many_edge_jumps_will_force_greedy_sequence_to_be_optimal}).
    \end{enumerate}
\end{remark}

The goal of the remainder of this section is to prove that we eventually make progress:

\begin{restatable}[\RepeatedGreedy will reach a progress condition]{theorem}{thmJumpEdgesOrRepeat}
\label{thm:repeatedgreedy_will_jump_edges_or_repeat}
    Running \Greedy for $\Oh{kn^2}$ revolutions guarantees the occurrence of a progress condition.
\end{restatable}

We will prove Theorem~\ref{thm:repeatedgreedy_will_jump_edges_or_repeat} by considering the local behavior of $\Greedy$ when applied to one local greedy sequence. We will need terminology to describe this; see Figure~\ref{fig:geometry_notation_overview}.

\begin{figure}[ht]
    \centering
    
    \begin{tikzpicture}
        \draw[color=black] (0,4) -- (0,1) -- (1,1.2) -- (2,1) -- (3,1.5) -- (4,1.5) -- (5,2) -- (4.5, 4.5);
        \draw[color=black, left] (0,4) node {$\ell_z$};
        \draw[color=black, right] (4.5, 4.5) node {$\ell_y$};

        \draw[color=black] (0.5, 4) -- (1, 2.5) -- (1.3, 4);
        \draw[color=black] (3.5, 4.4) -- (4, 2.7) -- (3.8, 4.2);

        \filldraw[color = black] (4.94, 2.32) circle (1.5pt);
        \draw[color = black, right] (4.94, 2.22) node {$y_i$};
        \filldraw[color = black] (0,2.18) circle (1.5pt);
        \draw[color = black, left] (0,2.18) node {$z_i$};

        \filldraw[color = guard] (2.43, 2.96) circle (1.5pt);
        \draw[color = guard, above] (2.43, 2.96) node {$g_i$};
        \filldraw[color = gray] (4.89, 2.55) circle (1.5pt);
        \draw[color = gray, right] (4.89, 2.65) node {$y_i'$};

        \draw[color = guard, dashed] (0,2.18) -- (2.43, 2.96) -- (4.89, 2.55);
        \draw[color = guard, dashed] (2.43, 2.96) -- (4.94, 2.32);

        \filldraw[color = black] (1, 2.5) circle (1.5pt);
        \draw[color = black, below] (1, 2.5) node {$C_{z,i}$};
        \filldraw[color = black] (4, 2.7) circle (1.5pt);
        \draw[color = black, below, yshift=-3pt] (4, 2.7) node {$C_{y,i}$};

        \draw[color = black, dashed] (-0.2,2.42) -- (1,2.5) -- (4,2.7) -- (5,2.766);
        \draw[black] (2.5, 2.3) node {$\ell_i$};

        \filldraw[color = black] (4.8, 3) circle (1.5pt);
        \draw[color = black, right] (4.8, 3.1) node {$y_{i+1}$};
    \end{tikzpicture}   
    
    \caption{A greedy step from $y_i$ on edge $\ell_y$ to $z_i := \Gsc(y_i)$ on edge $\ell_z$ guarded by $g_i$, however $g_i$ can see from $y_i'$. The pivot point $C_{z,i}$ blocks $g_i$ from seeing further than $z_i$ and the other pivot $C_{y,i}$ blocks $g_i$ in the other direction, so $g_i$, $C_{y,i}$ and $y_i'$ lie on a line. The relative position of the line $\ell_i := \lineextension{C_{z,i}}{C_{y,i}}$ between the pivot points to the guard $g_i$ will be vital for our analysis.}
    \label{fig:geometry_notation_overview}
\end{figure}

We let $(y_i)_{i=0}^\infty$ and $(z_i)_{i=0}^\infty$ be local greedy sequences with $z_i = \Gsc(y_i)$. We let $g_i$ be the guard found by \Greedy, that sees $[y_i, z_i]$. 
If $(y_i)_{i=0}^\infty$ or $(z_i)_{i=0}^\infty$ jumps edges, we will have a progress condition, so we assume that this does not happen. Let $(y_i)_{i=0}^\infty$ be contained on $\ell_y$ and $(z_i)_{i=0}^\infty$ be contained on $\ell_z$. We also know, that $(y_i)_{i=0}^\infty$ will move up along $\ell_y$ and $(z_i)_{i=0}^\infty$ will move down along $\ell_z$, as we otherwise would get a positive finger print.
Of all the $z_i$ it is only possible for $z_0$ to be a vertex of $P$. This is important, as we now know that all other $z_i$ will be defined from a \emph{blockage}, i.e. $g_i$ can see no further than $z_i$, because there is some vertex of $P$ in the way (see Figure~\ref{fig:horizon_vs_blockage}).

\begin{figure}[ht]
    \centering
    
    \begin{tikzpicture}
        \draw[color=black] (0.5,0) -- (1,2) -- (2,0) -- (4,0.5) -- (5,3);
        \draw[color=guard, above] (3,3) node {$g$}; 
        \filldraw[guard] (3,3) circle (2pt);
        \draw[color=guard, dashed] (1,2) -- (3,3) -- (5,3);
    \end{tikzpicture}%
    \hspace{5em}%
    \begin{tikzpicture}
        \draw[color=black] (0,1.5) -- (1,0) -- (2,0) -- (4,0.5) -- (5,3);
        \draw[color=black] (0,3) -- (1,1.5) -- (2,3);
        \draw[color=guard, above] (3,3) node {$g$}; 
        \filldraw[guard] (3,3) circle (2pt);
        \draw[color=guard, dashed] (0.33333,1) -- (3,3) -- (5,3);
    \end{tikzpicture}
    
    \caption{Types of end points with (left) a horizon and (right) a blockage end point}
    \label{fig:horizon_vs_blockage}
\end{figure}

So we have at least one blockage between $z_i$ and $g_i$, the closest to $z_i$ we denote $C_{z,i}$.
If $y_{i+1}$ could see $g_i$, then $g_i$ could see $[y_{i+1},z_i]$, thus we would repeat in the greedy sequence (as $z_i = z_{i+1}$). Assuming this does not happen, there must be something blocking $g_i$ from $y_{i+1}$. We let $y_i'$ be the point on $\ell_y$ farthest up, which is still visible from $g_i$. As it is farther up, there must be some blockage between $y_i'$ and $g_i$, the closest to $y_i'$ we denote $C_{z,i}$.

We will refer to the points $C_{y,i}$ and $C_{z,i}$ as \emph{pivot points} and $\ell_i :=\lineextension{C_{y,i}}{C_{z,i}}$ the \emph{pivot line} for a given $i$. It is clear that $C_{y,i} \neq C_{z,i}$ so the pivot line is well defined.
The pivot points are all vertices of the polygon, thus there are at most $\Oh{n^2}$ possible pivot lines. In the following, we will look only at one pivot line at a time, and then run enough revolutions to see the same pivot line multiple times by the pigeonhole principle.

Finally, we track where we can place guards. For this, we define feasible regions:

\begin{definition}[Feasible region]\label{def:feasible_region}
    Let $y, z \in \partial P$and $\I{y}{z}$ the interval of $P$ from $y$ to $z$. 
    Then the \emph{feasible region} $\FeasibleRegion{\I{y}{z}}$ is the subset of $P$ of all possible guard placements that see $\I{y}{z}$.
    Specifically, $\FeasibleRegion{\I{y}{z}} \neq \emptyset$ if and only if$\I{y}{z} \in \Fc$.
\end{definition}

In the following, for a greedy step $\GscNEW{y_i} = z_i$ where $y_i$ and $z_i$ are on edges $\ell_y$ and $\ell_z$, we define $F := \FeasibleRegion{\I{y}{z}}$, where $y$ and $z$ are the endpoints of edges $\ell_y$ and $\ell_z$ contained in $[y_i, z_i]$.

% --------------------------------------------------------------------
\subsubsection*{Strategy}
% --------------------------------------------------------------------

We briefly outline our approach, which is based on the relationship between the feasible region and the pivot line in a greedy step, $\GscNEW{y_i} = z_i$. 
In Section~\ref{sec:F_above_pivot_line}, we show that the case where $F$ is entirely above $\ell_i$ will lead to a progress condition. 
In Section~\ref{sec:F_under_pivot_line}, we show that if the feasible region $F$ contains points below $\ell_i$, then within the next revolution a different pivot line $\widehat{\ell_i}$ is visited, with corresponding feasible region $\widehat{F}$ entirely above $\widehat{\ell_i}$.

% --------------------------------------------------------------------
\subsection{Supporting lemmas}
% --------------------------------------------------------------------

Throughout this section, we will need the following two properties:

\begin{restatable}{lemma}{lemFeasibleRegionIsConnected}
\label{lem:feasible_region_is_connected}
    For any $a, b \in \partial P$, the feasible region $\FeasibleRegion{[a, b]}$ is connected. Furthermore, for any convex subset $C \subset P$, $C \cap \FeasibleRegion{[a,b]}$ is convex.
\end{restatable}

\begin{proofsketch}
    The proof follows by induction in the number of vertices in $[a,b]$. In the base case, the feasible region is a single visibility polygon, hence it is connected. In the induction step, we intersect the previous feasible region with a visibility polygon from a vertex $v_{i+1}$. If this new feasible region $F$ were to be disconnected, we find two points $S, T$ in different components of $F$, where $\linesegment{S}{T}$ is contained in a triangle completely visible to $v_{i+1}$, i.e. completely contained inside $P$. This will contradict Lemma~\ref{lem:convex_viewing_lemma} in Appendix~\ref{apx:feasible_region_connected} and thus $F$ will be connected. A similar argument shows in convexity claim.
    
    For the detailed proof, see Appendix~\ref{apx:feasible_region_connected}.
\end{proofsketch}

\begin{restatable}[Blockings happen outside visible area]{lemma}{lemBlockingCornerNotInView}
\label{lem:Blocking_corner_not_in_view}
    Let $y \in \partial P$ and $z = \Gsc(y)$ and assume that neither $y$ nor $z$ is a vertex of $P$. Let $g$ be a guard seeing $[y,z]$ and let $c$ be a vertex blocking the view from $g$ to $z$. Then $c \not \in [y, z]$. 
\end{restatable}

\begin{proofsketch}
    One considers what $g$ can see around $C$. It can be shown, that $g$ cannot see everything in the vicinity of $C$ (Lemma~\ref{lem:what_G_sees_around_C} in Appendix~\ref{apx:yz_cannot_block_z}). Assume for contradiction that $C \in [y,z]$, then either some part of $[y, z]$ is not visible from $g$ or $z$ can be moved to a visible vertex, yielding a contradiction.
    
    For the detailed proof, see Appendix~\ref{apx:yz_cannot_block_z}.
\end{proofsketch}

% --------------------------------------------------------------------
\subsection{Feasible region strictly above pivot line}
\label{sec:F_above_pivot_line}
% --------------------------------------------------------------------

In this section, we analyze configurations in which the feasible region lies strictly above a pivot line. The main result of this section is to show that the guard will have a unique optimal position in $F$. This will imply a progress condition under the assumption that the same pivot line is chosen twice. In Figure~\ref{fig:smaller_angles_see_more} we see that the angles between the pivot points and the points above the pivot line are a good measure of the quality of prospective guard positions. In Figure~\ref{fig:comparison_of_angles_in_different_scenarios} we use this insight to define the \emph{shadow} of a point, which is used in Lemma~\ref{lem:F_in_Ta} to show that the feasible region is contained in the shadow of one of its unique point closest to $\ell_i$. Finally, in Proposition~\ref{prop:F_over_l_gives_repetition_or_edge_jumping} we show that whenever the feasible region is above the pivot line, a progress condition will be satisfied.

\begin{figure}[ht]
    \centering
    
    \begin{tikzpicture}
        \draw[color=black] (0,4) -- (0,1);
        \draw[color=black] (5,1) -- (4.5, 4.5);
        \draw[left] (0,3) node {$\ell_z$}; 
        \draw[right] (4.7, 3.5) node {$\ell_y$};
        
        \filldraw[color = extracolor] (2.93, 2.96) circle (1.5pt);
        \draw[color = extracolor, above] (2.93, 2.96) node {$w$};
        \filldraw[color = feasible region] (1.63, 3.06) circle (1.5pt); 
        \draw[color = feasible region, above] (1.63, 3.06) node {$v$};

        \draw[color=extracolor] (0, 2.3) -- (2.93, 2.96) -- (4.81, 2.32);
        \draw[color=feasible region] (0, 1.68) -- (1.63, 3.06) -- (4.78, 2.51);
        
        \coordinate (A) at (2.93, 2.96);
        \coordinate (B) at (1, 2.525);
        \coordinate (C) at (3.7, 2.7);
        \coordinate (D) at (1.63, 3.06);
        \pic [draw, -, extracolor, angle radius = 0.6cm] {angle=C--B--A};
        \pic [draw, -, extracolor, angle radius = 0.6cm] {angle=A--C--B};
        \pic [draw, -, feasible region, angle radius = 0.5cm] {angle=C--B--D};
        \pic [draw, -, feasible region, angle radius = 0.5cm] {angle=D--C--B};
        
        \draw[color=black] (0.5, 4) -- (1, 2.545) -- (1.3, 4);
        \draw[below] (1.1,2.545) node {$C_{z,i}$};
        \draw[color=black] (3.5, 4.4) -- (3.7, 2.7) -- (3.8, 4.2);
        \draw[below] (3.7, 2.7) node {$C_{y,i}$};
        
        \draw[color=gray] (1,2.525) -- (3.7,2.7);
        \draw[color = gray] (2.3, 2.35) node {$\ell_i$};
    \end{tikzpicture}
    
    \caption{In the setup of Figure~\ref{fig:geometry_notation_overview}, a guard above the pivot line $\ell_i$ placed in point $v$ can see more of the $\ell_y$ edge than if it was placed in $w$, and vice-versa, analogously to Figure~\ref{fig:CAG_tradeoff_infinite_different}.}
    \label{fig:smaller_angles_see_more}
\end{figure}

Fixing $v$, it is clear that the area where a guard will have a smaller angle than $
v$ with respect to $\ell_{i}$ is below $\lineextension{v}{C_{z,i}}$, and a larger angle above $\lineextension{v}{C_{z,i}}$ see Figure~\ref{fig:comparison_of_angles_in_different_scenarios}.

The intersection of the half-planes above $\lineextension{a}{C_{z,i}}$ and $\lineextension{v}{C_{y,i}}$ is a quarter plane where every guard placement is worse than $v$, both with respect to $\ell_y$ and $\ell_z$. For a point $v \in F$ we denote this quarter plane $S(v)$ and call it the \emph{shadow} of $v$; see Figure~\ref{fig:comparison_of_angles_in_different_scenarios} (right).

\begin{figure}[ht]
    \centering

    \begin{tikzpicture}
        \draw[color=black] (0,4) -- (0,1);
        \draw[color=black] (5,1) -- (4.5, 4.5);
        \draw[left] (0,3) node {$\ell_z$}; 
        
        \draw[right] (4.7, 3.5) node {$\ell_y$};

        \filldraw[shadow] (1.3,4) -- (1,2.525) -- (3.59,3.64) -- (3.5,4.4) -- cycle;
        \draw[shadowcolor] (2.1, 3.6) node {\scriptsize{Larger angle}};
        \filldraw[shadow2] (1,2.525) -- (3.59,3.64) -- (3.7, 2.7) -- cycle;
        \draw[shadowcolor!80] (3.1, 3.1) node {\scriptsize{Smaller}};
        \draw[shadowcolor!80] (3.1, 2.8) node {\scriptsize{angle}};

         \draw[color=shadowcolor, dashed] (0, 2.09) -- (3.59,3.64);
            
        \draw[color=black] (0.5, 4) -- (1, 2.545) -- (1.3, 4);
        \draw[color=black] (3.5, 4.4) -- (3.7, 2.7) -- (3.8, 4.2);

        \filldraw[color = feasible region] (2.33, 3.1) circle (1.5pt); 
        \draw[color = feasible region, below] (2.33, 3.1) node {$v$};
        
        \draw[color=gray] (1,2.525) -- (3.7,2.7);
        \draw[color = gray] (2.3, 2.35) node {$\ell_i$};
    \end{tikzpicture}%
    \hspace{5em}%
    \begin{tikzpicture}
        \draw[color=black] (0,4) -- (0,1); 
        \draw[color=black] (5,1) -- (4.5, 4.5);
        \draw[left] (0,3) node {$\ell_z$}; 
        \draw[right] (4.7, 3.5) node {$\ell_y$};

        \filldraw[shadow] (1.3, 4) -- (1.18, 3.43) -- (2.33, 3.1) -- (3.59, 3.64) -- (3.5, 4.4) -- cycle;
        \draw[shadowcolor] (2.4, 3.7) node {\scriptsize{$S(v)$}};

        \draw[color=shadowcolor, dashed] (0, 2.09) -- (3.59,3.64);
        \draw[color=shadowcolor, dashed] (4.8, 2.38) -- (1.18, 3.43);
            
        \draw[color=black] (0.5, 4) -- (1, 2.545) -- (1.3, 4);
        \draw[color=black] (3.5, 4.4) -- (3.7, 2.7) -- (3.8, 4.2);

        \filldraw[color = feasible region] (2.33, 3.1) circle (1.5pt); 
        \draw[color = feasible region, below] (2.33, 3.1) node {$v$};

        \draw[color=gray] (1,2.525) -- (3.7,2.7);
        \draw[color = gray] (2.3, 2.35) node {$\ell_i$};
    \end{tikzpicture}
        
    \caption{Left, points with smaller angle than $v$ with $\ell_i$ see more of $\ell_z$ and are better guards. Right, placing a guard strictly in the shadow $S(v)$ of $v$ leads to a strictly worse guard.}
    \label{fig:comparison_of_angles_in_different_scenarios}
\end{figure}

\begin{remark}[$z$ above $\ell_i$]
\label{rem:z_over_li}
    Let $z$ be the vertex of $\ell_z$ in $[y_i, z_i]$ and $a$ is the vertex of $F$ closest to $\ell_i$. If $z$ lies above $\lineextension{a}{C_{z,i}}$, as illustrated in Figure~\ref{fig:z_above_li}, we see some weird behavior:

    For $C_{z,i}$ to be the pivot, $z_i$ has to lie below $\ell_i$, thus $a$ cannot see any part of $\ell_z$ except for $z$.
    Thus it makes sense to force the points in $F$ to see a small part of $\ell_y$ and $\ell_z$. This can be done by setting $F = F([y - \varepsilon, z+\varepsilon])$, where $y- \varepsilon$ and $z+\varepsilon$ refers to points on $\ell_y$ respectively $\ell_z$ very close to $y$ respectively $z$. 

    Making this change will not impact the other proofs other than fixing some special cases in this subsection. By doing this, $z$ cannot lie above $\lineextension{a}{C_{z,i}}$, which is important for future proofs.
    The analogous behavior holds true if $y$ lies above $\lineextension{C_{y,i}}{a}$.
\end{remark}

\begin{figure}[ht]
    \centering
    
    \begin{tikzpicture}[scale = 0.75]
        \draw[color = black] (-2.02, 3.99) -- (-0.33, 0.84) -- (-0.81, 5.31);
        \draw[color = black] (5.34, 5) -- (4.88, 0.64) -- (6.59, 4.49);
        
        \draw[color = black, dashed] (-0.33, 0.84) -- (4.88, 0.64);
        \draw[color = black, below left, xshift=5pt, yshift=5pt] (-0.33, 0.84) node {$C_{z,i}$};
        \draw[color = black, below] (4.88, 0.64) node {$C_{y,i}$};
        
        \draw[color = shadowcolor, dashed] (-0.33, 0.84) -- (5.16, 3.31);
        \draw[color = shadowcolor, dashed] (4.88, 0.64) -- (-0.63, 3.7);
        
        \filldraw[color = feasible region] (2.35, 2.04) circle (1.5pt);
        \draw[color = feasible region, below] (2.35, 2.04) node {$a$};

        \draw[color = black] (-0.6, -1) -- (0.65, 2.85) -- (0.85, 1.69) -- (1.05, 1.69) -- (1.4, 1.34) -- (1.4, 0) -- (2.4, 0) -- (3, -0.1) -- (4.5, -0.3) -- (5.6, 0) -- (5.8, 1);
        \draw[color = black, below] (3, -0.1) node {$[y_i, z_i]$};

        \filldraw[color = gray, dashed] (0.65, 2.85) -- (2.35, 2.04);
        \draw[color = black, left] (0.65, 2.85) node {$z$};
        \filldraw[color = black, dashed] (0.65, 2.85) circle (1.5pt);

        \draw[left] (-0.6, -1) node {$z_i$};
        \filldraw (-0.6, -1) circle (1.5pt);

        \draw[dashed] (0.65, 2.85) -- (1.3, 4.84);
    \end{tikzpicture}
    
    \caption{$a$ can see nothing on $\ell_i$ but $z$. Thus picking $a$ as a guard would lead to edge jumps immediately}
    \label{fig:z_above_li}
\end{figure}

\begin{lemma}[Successor and predecessor edges to $a \in F$ lie in $S(a)$]
\label{lem:successor_and_predecessor_edges_to_a_will_lie_in_Ta}
    Assume $F$ lies strictly above $\ell_i$ and let $a$ be a vertex of $F$ closest to $\ell_i$. Let $e$ and $f$ denote the edges of $F$ connected to $a$. Then $e, f \subseteq S(a)$.
\end{lemma}

\begin{proof}
    We will start by determining what characterizes edges in $F$. 
    Since $F$ is the intersection of the visibility polygons of vertices of $\partial P$, an edge of $F$ has to be an part of an edge of the visibility polygon to a vertex $v$.

    The edges of the visibility polygon from $v$ come in two ways: Either they come from shooting rays from $v$ through other vertices of $P$ or as edges of $P$.
    
    Assume for contradiction that $e$ or $f$ is not contained in $S(a)$. As the proof is symmetric, we assume $e \not \subseteq S(a)$, $e$'s other endpoint is $b$ and $b$ is closer to $C_{y,i}$ than $C_{z,i}$ (again by symmetry). Furthermore let $z$ be the last vertex of $[y_i, z_i]$, thus $z$ can see $F$ and hence also points $a$ and $b$. Now $e$ must either be a segment contained in a ray from a vertex $v \in [y_i, z_i]$ or $e$ is a part of $\partial P$.

    \begin{case}
        If $e \subseteq \ray{v}{a}$ we must have $v \in \lineextension{b}{a}$. For $z$ to see $a$, $z$ must lie above $\lineextension{b}{a}$ as $[v,z]$ would block it otherwise. However now we would need $z$ above $\lineextension{a}{C_{z,i}}$, which we disallow by Remark~\ref{rem:z_over_li}. Thus we arrive at a contradiction.
    \end{case}
    
    \begin{case}
        If $e$ is a part of $\partial P$, we can assume $F$ lies above $\lineextension{b}{a}$, since $a$ is a lowest point of $F$; see Figure~\ref{fig:z_cannot_see_a_and_b} (right). However, now $z$ cannot see $b$ as the space directly below $\linesegment{a}{b}$ is not inside $P$ and Remark~\ref{rem:z_over_li} forcing $z$ below $\lineextension{b}{a}$, leading to a contradiction.
        \popQED
    \end{case}
\end{proof}

\begin{figure}[ht]
    \centering

    \begin{tikzpicture}[scale = 0.75]
        \begin{scope}
            \draw[color = black] (-2.02, 3.99) -- (-0.33, 0.84) -- (-0.81, 5.31);
            \draw[color = black] (5.34, 5) -- (4.88, 0.64) -- (6.59, 4.49);
    
            \draw[color = black, dashed] (-0.33, 0.84) -- (4.88, 0.64);
            \draw[color = black, below left, xshift=5pt, yshift=5pt] (-0.33, 0.84) node {$C_{z,i}$};
            \draw[color = black, below] (4.88, 0.64) node {$C_{y,i}$};
    
            \draw[color = shadowcolor, dashed] (-0.33, 0.84) -- (5.16, 3.31);
            \draw[color = shadowcolor, dashed] (4.88, 0.64) -- (-0.63, 3.7);
            
            \draw[color = black, dotted] (-0.39, 1.41) -- (2.35, 2.045)  (3.87, 2.395) -- (5.1, 2.68);

            \draw[feasible region] (2.35, 2.045) -- (3.87, 2.395);

            \draw[below, feasible region] (3.1, 2.2) node {$e$};
            
            \draw[color = black] (-0.6, -1) -- (0.65, 2.85) -- (0.85, 1.69) -- (1.05, 1.69) -- (1.4, 1.34) -- (1.4, 0) -- (2.4, 0) -- (3, -0.1) -- (4.5, -0.3) -- (5.6, 0) -- (5.8, 1);
            \draw[color = black, below] (3, -0.1) node {$[y_i, z_i]$}; 
            \filldraw[color = black] (0.85, 1.69) circle (1.5pt);
            \draw[color = black, above] (0.95, 1.69) node {$v$};
            
            \filldraw[color = feasible region] (2.35, 2.04) circle (1.5pt);
            \draw[color = feasible region, below] (2.35, 2.04) node {$a$};
    
            \filldraw[color = feasible region] (3.87, 2.4) circle (1.5pt);
            \draw[color = feasible region, below] (3.87, 2.4) node {$b$};
            
            \filldraw[color = gray, dashed] (0.65, 2.85) -- (2.35, 2.04);
            \draw[color = black, left] (0.65, 2.85) node {$z$};
            \filldraw[color = black, dashed] (0.65, 2.85) circle (1.5pt);
            
            \draw[left] (-0.6, -1) node {$z_i$};
            \filldraw (-0.6, -1) circle (1.5pt);
            
            \draw[dashed] (0.65, 2.85) -- (1.3, 4.84);
        \end{scope}

        \begin{scope}[xshift=9.5cm]
            \draw[color = black] (-2.02, 3.99) -- (-0.33, 0.84) -- (-0.81, 5.31);
            \draw[color = black] (5.34, 5) -- (4.88, 0.64) -- (6.59, 4.49);

            \filldraw[Fbare] (1.8, 3.0) -- (2.35, 2.04) -- (3.87, 2.4) -- cycle;
            \draw[feasible region] (2.6, 3) node {$F$};

            \draw[color = black, dashed] (-0.33, 0.84) -- (4.88, 0.64);
            \draw[color = black, below] (-0.33, 0.84) node {$C_{z,i}$};
            \draw[color = black, below] (4.88, 0.64) node {$C_{y,i}$};

            \draw[color = notvisible, dashed] (-0.5, -0.45) -- (3.87, 2.4);
            
            \draw[color = shadowcolor, dashed] (-0.33, 0.84) -- (5.16, 3.31);
            \draw[color = shadowcolor, dashed] (4.88, 0.64) -- (-0.63, 3.7);
            
            \draw[color = feasible region] (1.8, 3.0) -- (2.35, 2.04) -- (3.87, 2.4);

            \filldraw[color = feasible region] (2.35, 2.04) circle (1.5pt);
            \draw[color = feasible region, below] (2.35, 2.04) node {$a$};

            \filldraw[color = feasible region] (3.87, 2.4) circle (1.5pt);
            \draw[color = feasible region, below] (3.87, 2.4) node {$b$};
            
            \draw[below, feasible region] (3, 2.3) node {$e$};

            \draw[color = black, left] (-0.5, -0.45) node {$z$};
            \filldraw[color = black, dashed] (-0.5, -0.45) circle (1.5pt);
        \end{scope}
    \end{tikzpicture}
    
    \caption{Left, if $z$ sees $a$, $z$ is above $\lineextension{b}{a}$ and $z_i$ is below $\ell_i$ for $C_{z,i}$ to be a pivot. Right, if nothing blocks visibility between $z$ and $b$ then $a$ is not closest to $\ell_i = \linesegment{C_{y,i}}{C_{z,i}}$.}
    \label{fig:z_cannot_see_a_and_b}
\end{figure}

\begin{lemma}[The feasible region is contained in the shadow]
\label{lem:F_in_Ta}
    Assume that $F$ lies above $\ell_i$ and let $a$ be a vertex of $F$ closest to $\ell_i$. Then $F \subseteq S(a)$.
\end{lemma}

\begin{proof}
    Assume for contradiction there are points in $F$ outside of $S(a)$. Let $b$ be a point in $F$ such that $\angle baC_{y,i}$ or $\angle C_{z,i}ab$ is minimal (depending on whether $b$ is closer to $C_{y,i}$ or $C_{z,i}$). Assume w.l.o.g. $b$ is closer to $C_{y,i}$ than $C_{z,i}$. Because of the minimality of $\angle baC_{y,i}$, we now know all of $F$ is above $L(b,a)$.
    
    By Lemma~\ref{lem:successor_and_predecessor_edges_to_a_will_lie_in_Ta} we know $b$ is not the successor or predecessor to $a$, thus not all of $\linesegment{a}{b}$ will be in $F$. So for at least one vertex $v$, $a$ and $b$ will be visible, but not all of $\linesegment{a}{b}$. Let $d$ be a point on $\linesegment{a}{b}$ not visible to $v$.

    \begin{figure}[ht]
        \centering

        \begin{tikzpicture}[scale = 0.75]
            \draw[color = black] (-2.02, 3.99) -- (-0.33, 0.84) -- (-0.81, 5.31);
            \draw[color = black] (5.34, 5) -- (4.88, 0.64) -- (6.59, 4.49);

            \draw[color = black, dashed] (-0.33, 0.84) -- (4.88, 0.64);
            \draw[color = black, below] (-0.33, 0.84) node {$C_{z,i}$};
            \draw[color = black, below] (4.88, 0.64) node {$C_{y,i}$};

            \draw[color = shadowcolor, dashed] (-0.33, 0.84) -- (5.16, 3.31);
            \draw[color = shadowcolor, dashed] (4.88, 0.64) -- (-0.63, 3.7);
            
            \draw[color = notvisible, dotted] (-0.39, 1.41) -- (2.35, 2.04);
            \draw[color = notvisible, dotted] (2.35, 2.04) -- (3.87, 2.4);
            \draw[color = notvisible, dotted] (3.87, 2.4) -- (5.1, 2.68);
            
            \filldraw[color = feasible region] (2.35, 2.04) circle (1.5pt);
            \draw[color = feasible region, below] (2.35, 2.04) node {$a$};

            \filldraw[color = feasible region] (3.87, 2.4) circle (1.5pt);
            \draw[color = feasible region, below] (3.87, 2.4) node {$b$};

            \draw[color = black, left] (2.19, -0.65) node {$v$};
            \filldraw[color = black, dashed] (2.19, -0.65) circle (1.5pt);
            \draw[dashed] (2. 19, -0.65) -- (2.76, 2.13);
            \draw[dashed, notvisible]  (2.76, 2.13) --  (3.34, 5);
            
            \filldraw (2.76, 2.13) circle (1.5pt);
            \draw (3.16, 1.93) node {$d$};
        \end{tikzpicture}
        
        \caption{Ray $\ray{v}{d}$ and $\lineextension{b}{a}$ divides the feasible region $F$ into disconnected components }
        \label{fig:F_divided_into_disconnected_components}
    \end{figure}
    
    It must now hold that $F$ cannot intersect the ray $\ray{v}{d}$. 
    Furthermore, $F$ cannot cross $\lineextension{b}{a}$ from the minimality of the angle (see Figure~\ref{fig:F_divided_into_disconnected_components}). 
    Thus $a$ and $b$ must be in disconnected components of $F$. By Lemma~\ref{lem:feasible_region_is_connected} $F$ must be connected, leading to a contradiction.
\end{proof}

We can now show that when $F$ lies strictly above $\ell_i$, we get a progress condition:

\begin{proposition}[Feasible region over repeated pivot line implies a progress condition]
\label{prop:F_over_l_gives_repetition_or_edge_jumping}
    If $\ell_i = \ell_j$ for some $i < j$ and $F$ lies above $\ell_i$, then we get a progress condition.
\end{proposition}

\begin{proof}
    Let $v$ be a vertex of $F$ closest to $\ell_i$. By Lemma~\ref{lem:F_in_Ta} we know $F \subseteq S(v)$. 
    
    Consider $y_i$. If it can see any part of $F$, it can also see $v$, since $F \subseteq S(v)$, and any point in $F$ can see no farther than $v$ on $\ell_z$, again because $F \subseteq S(a)$. Thus, $z_i$ is the farthest point on $\ell_z$ visible from $v$. If $y_i$ could not see any point in $F$, then $z_i$ could not lie on $\ell_z$ which is an edge jump, hence a progress condition.

    Now consider $y_j$. We get the exact same considerations as above, so $z_j$ is the farthest point on $\ell_z$ visible from $v$ and $z_i = z_j$, a repetition, or $z_j \not \in \ell_z$, and edge jump, both are progress conditions, as wanted.
\end{proof}

% --------------------------------------------------------------------
\subsection{Feasible region below pivot line}
\label{sec:F_under_pivot_line}
% --------------------------------------------------------------------

In this section, we analyze what happens when $F$ is not strictly above $\ell_i$. Contrary to what we have seen in Section~\ref{sec:F_above_pivot_line}, when $g_i$ can lie below $\ell_i$, it does not hold that there is a unique optimal guard placement. Thus we will show that when $F$ is on or below $\ell_i$, we will be able to find some other guard, which is found by \Greedy, that lies strictly above its pivot line, thus forcing a unique optimal guard placement at some other point in the polygon. 
We denote the polygon constructed by the edges of $[y'_i, z_i]$, $\linesegment{z_i}{g_i}$ and $\linesegment{g_i}{y'_i}$ by $Q_i$, see Figure~\ref{fig:Q_i}. Recall that $y_i'$ is defined considering the last point on $\ell_y$ visible to $g_i$.

\begin{figure}[ht]
    \centering
    
    \begin{tikzpicture}
        \draw[color=black] (0,4) -- (0,1) -- (1,1.2) -- (2,1) -- (3,1.5) -- (4,1.5) -- (5,2) -- (4.5, 4.5);
        \filldraw[qi] (0,2.9) -- (0,1) -- (1,1.2) -- (2,1) -- (3,1.5) -- (4,1.5) -- (5,2) -- (4.8, 3.1) -- (2.43, 2) -- cycle;
        
        \draw[color=black] (0.5, 4) -- (1, 2.525) -- (1.3, 4);
        \draw[color=black] (3.5, 4.4) -- (4, 2.7) -- (3.8, 4.2);

        \draw[guard, dashed] (0,2.9) -- (2.43, 2) -- (4.8, 3.1);

        \filldraw[color = black] (4.8, 3.1) circle (1.5pt);
        \draw[color = black, right] (4.8, 3.1) node {$y_i'$};
        \filldraw[color = black] (0,2.9) circle (1.5pt);
        \draw[color = black, left] (0,2.9) node {$z_i$};

        \filldraw[color = guard] (2.43, 2) circle (1.5pt);
        \draw[color = guard, above] (2.43, 2) node {$g_i$};

        \draw[Qi] (1.5, 1.7) node {$Q_i$};
    \end{tikzpicture}
    
    \caption{The polygon $Q_i$ defined by $[y'_i, z_i]$, $\linesegment{z_i}{g_i}$ and $\linesegment{g_i}{y'_i}$.}
    \label{fig:Q_i}
\end{figure}

\begin{lemma}
\label{lem:G_is_line_of_sight_is_clear}
    $[z_i, y_i']$ does not cross the boundary of $Q_i$.
\end{lemma}

\begin{proof}
    The guard $g_i$ must be able to see $y'_i$ and $z_i$, so $\linesegment{g_i}{y'_i}$ and $\linesegment{g_i}{z_i}$ must not be obstructed. Since all the boundary segments in $[y_i', z_i]$ cannot be obstructed either as $P$ is simple, no part of $[z_i, y_i']$ (i.e. the interval not guarded by $g_i$) can cross the boundary of $Q_i$.
\end{proof}

If $F$ lies above and below $\ell_i$ we could potentially have guards below and above. 
However, the following lemma shows that guards are below the pivot line if possible. 

\begin{lemma}[Guards are placed on or below the pivot lines if possible]
\label{lem:better_to_have_Gi_below_pivot_line}
    Let $\ell_i^-$ be the closed half-plane below the $i$'th pivot line.
    If $F \cap \ell_i^- \neq \emptyset$ then $g_i$ will lie on or below $\ell_i$.
\end{lemma}

\begin{proof}
    Assume for contradiction that $g_i$ lies strictly above $\ell_i$. 
    It holds that $F$ is connected by Lemma~\ref{lem:feasible_region_is_connected}. 
    Combining this with the $F \cap \ell^-_i \neq \emptyset$ condition, then there must be a point on $\linesegment{C_{y,i}}{C_{z,i}}\subset\ell_i$ also contained in $F$ ($\ell_i$ is the extended line, where $\linesegment{C_{y,i}}{C_{z,i}}$ is only the line segment). Since $\linesegment{C_{y,i}}{C_{z,i}} \subseteq Q_i$ we know that $F \cap \ell^-_i \cap Q_i \neq \emptyset$. Let $g_i' \in F \cap \ell^-_i \cap Q_i$.

    \begin{figure}[ht]
        \centering
    
        \begin{tikzpicture}
            \begin{scope}
                \draw[color=black] (0,4) -- (0,1) -- (1,1.2) -- (2,1) -- (3,1.5) -- (4,1.5) -- (5,2) -- (4.5, 4.5);
                \filldraw[qi] (0,2.18) -- (0,1) -- (1,1.2) -- (2,1) -- (3,1.5) -- (4,1.5) -- (5,2) -- (4.89, 2.55) -- (2.43, 2.96) -- cycle;
                \filldraw[F] (1.5, 2) -- (2,2.5) -- (2, 3) -- (2.4, 3.3) -- (2.6, 3) -- (3, 2) -- (3, 1.7) -- cycle;
                
                \draw[color=black] (0.5, 4) -- (1, 2.525) -- (1.3, 4);
                \draw[color=black] (3.5, 4.4) -- (4, 2.7) -- (3.8, 4.2);
        
                \filldraw[color = black] (4.89, 2.55) circle (1.5pt);
                \draw[color = black, right] (4.89, 2.65) node {$y_i'$};;
                \filldraw[color = black] (0,2.18) circle (1.5pt);
                \draw[color = black, left] (0,2.18) node {$z_i$};
                
                \filldraw[color = guard] (2.43, 2.96) circle (1.5pt);
                \draw[color = guard, above, xshift=3pt] (2.63, 2.96) node {$g_i$};
    
                \draw[color = black, dashed] (1,2.5) -- (4,2.7);
                \filldraw[color = guard] (2.5, 2.6) circle (1.5pt);
                \draw[color = guard, below] (2.5, 2.6) node {$g_i'$};
                \draw[Qi] (3.5, 2) node {$Q_i$};
                \draw[feasible region] (2, 2.1) node {$F$};
                \draw (3.5, 3) node {$l_i$};
            \end{scope}
    
            \begin{scope}[xshift=7cm]
                \draw[color=black] (0,4) -- (0,1) -- (1,1.2) -- (2,1) -- (3,1.5) -- (4,1.5) -- (5,2) -- (4.5, 4.5);
                \filldraw[qi] (0,2.18) -- (0,1) -- (1,1.2) -- (2,1) -- (3,1.5) -- (4,1.5) -- (5,2) -- (4.89, 2.55) -- (2.43, 2.96) -- cycle;
                \filldraw[F] (0.5, 2.342) -- (1, 2.5) -- (2, 3) -- (2.4, 3.3) -- (2.6, 3) -- (3, 2.7) -- (2.4, 2.65) -- (1, 2.5) -- cycle;
        
                \draw[color=black] (0.5, 4) -- (1, 2.525) -- (1.3, 4);
                \draw[color=black] (3.5, 4.4) -- (4, 2.7) -- (3.8, 4.2);
        
                \draw[color = black, dashed] (1,2.5) -- (4,2.7);
                \draw[color=guard, dashed] (0.75, 2.42) -- (0,2.8);
                
                \filldraw[color = black] (4.89, 2.55) circle (1.5pt);
                \draw[color = black, right] (4.89, 2.65) node {$y_i'$};
                \filldraw[color = black] (0,2.18) circle (1.5pt);
                \draw[color = black, left] (0,2.18) node {$z_i$};
    
                \filldraw[color = black] (0,2.8) circle (1.5pt);
                \draw[color = black, left] (0,2.8) node {$z_i'$};
        
                \filldraw[color = guard] (2.43, 2.96) circle (1.5pt);
                \draw[color = guard, above, xshift=3pt] (2.63, 2.96) node {$g_i$};
        
                \draw[red] (3.5, 2) node {$Q_i$};
                \draw[feasible region] (1.6, 3) node {$F$};
                \draw (3.5, 3) node {$l_i$};
        
                \filldraw[guard] (0.75, 2.42) circle (1.5pt);
                \draw[guard, above] (0.7, 2.42) node {$g_i'$};
                
            \end{scope}
        \end{tikzpicture}
        
        \caption{Left, the setup of Lemma~\ref{lem:better_to_have_Gi_below_pivot_line} with feasible region $F$, guard $g_i$ that can see all of $Q_i$ along with placement of $g_i'$. Right, $g_i$ cannot be optimal as $g_i'$ can see further than $g_i$.}
        \label{fig:Q_and_F_under_li}
    \end{figure}

    As $g_i' \in F$, $g_i'$ can see every vertex in $[y_i,z_i]$, and since $g_i' \in Q_i$, $g_i'$ can also see $y_i$ and $z_i$ by Lemma~\ref{lem:G_is_line_of_sight_is_clear} and Lemma~\ref{lem:Blocking_corner_not_in_view}, 
    see Figure~\ref{fig:Q_and_F_under_li} left.
    Thus $g_i'$ must also see all of $[y_i, z_i]$.
    However, we know from analyzing \Greedy (in Remark~\ref{rem:_multiple_optimal_guards_lie_on_line_with_endpoint}), that multiple optimal guards must be collinear with $z_i$, so $g_i' \in \lineextension{z_i}{g_i} = \lineextension{z_i}{C_{z,i}}$ and since $g_i' \in \ell_i^-$, $g_i'$ must lie in $\linesegment{z_i}{C_{z,i}}$.
    But now no vertex of $P$ is between $g_i'$ and $z_i$ (since $C_{z,i}$ is the vertex closest to $z_i$ between $z_i$ and $g_i$), so $g_i'$ must be able to see more than $z_i$ making $g_i$ not optimal, leading to a contradiction. 
\end{proof}

Now we know that the placement of the guard must be below $\ell_i$, if possible. We next prove that the points $y_i$ and $z_i$ will both be above $\ell_i$ if we see $\ell_i$ again:

\begin{lemma}[$y_j$ and $z_j$ will lie above $\ell_j$]
\label{lem:yj_and_zj_above_li}
    If $F \cap \ell_i^- \neq \emptyset$ and $\ell_i = \ell_j$ with $i<j$. Then $y_j$ and $z_j$ will both lie above $\ell_i$ or we reach a progress condition.
\end{lemma}

\begin{figure}[ht]
    \centering
    
    \begin{tikzpicture}
        \draw[color=black] (0,4) -- (0,1) -- (1,1.2) -- (2,1) -- (3,1.5) -- (4,1.5) -- (5,2) -- (4.5, 4.5);
        \draw[color=black] (0.5, 4) -- (1, 2.525) -- (1.3, 4);
        \draw[color=black] (3.5, 4.4) -- (4, 2.7) -- (3.8, 4.2);

        \filldraw[color = black] (4.94, 2.32) circle (1.5pt);
        \draw[color = black, right] (4.94, 2.22) node {$y_i$};
        \filldraw[color = black] (4.79, 3.07) circle (1.5pt);
        \draw[color = black, right] (4.79, 3.07) node {$y'_i$};
        \filldraw[color = black] (4.72, 3.4) circle (1.5pt);
        \draw[color = black, right] (4.72, 3.45) node {$y_j$};

        \filldraw[color = black] (0,2.92) circle (1.5pt);
        \draw[color = black, left] (0,2.92) node {$z_i$};
        \filldraw (0, 2.62) circle (1.5pt);
        \draw[left] (0, 2.62) node {$z_j$};

        \filldraw[color = guard] (2.43, 1.96) circle (1.5pt);
        \draw[color = guard, below] (2.33, 1.96) node {$g_i$};

        \coordinate (g) at (3.42, 2.2);
        
        \filldraw[color = guard] (g) circle (1.5pt);
        \draw[color = guard, below] (g) node {$g_j$};

        \draw[color = guard, dashed] (0, 2.62) -- (g) -- (4.72, 3.4);
        
        \draw[color = guard, dashed] (0,2.92) -- (2.43, 1.96) -- (4.79, 3.07);
        \draw[color = guard, dashed] (0,2.47) -- (4.85,2.75);

        \draw (2.6, 2.9) node {$\ell_i$};
    \end{tikzpicture}
    
    \caption{If the pivot line $\ell_i$ is reused at step $i < j$ then both $y_j$ and $z_j$ are above it.}
    \label{fig:yj_and_zj_above_li_main}
\end{figure}

\begin{proof}

    We also have that $C_{y,i}$ lies on $\linesegment{g_i}{y_i'}$, thus $y_i'$ lies above $\ell_i$, and since $i < j$, either $y_j$ will also lie above $\ell_i$ or we get a positive fingerprint.

    Analogously, $C_{y,i}$ lies on $\linesegment{g_i}{y_i'}$, thus $y_i'$ lies above $\ell_i$, and since $i < j$ (and we assume negative fingerprint), $y_j$ will also lie above $\ell_i$.
\end{proof}

We are now ready to give the main result of the subsection, essential in proving Theorem~\ref{thm:repeatedgreedy_will_jump_edges_or_repeat}.

\begin{proposition}[Guard placed below pivot implies other guard placed strictly above pivot]
\label{prop:guards_below_pivots_implies_other_guards_above_pivots}
    Let $g_i$ lies below the pivot line and $y_i, z_i$ lie above the pivot line, then one of the following will occur:
    \begin{enumerate}
        \item In the next revolution, we will see a progress condition.
        \item $z_{i+1}$ will lie below the pivot line $\ell_i$.
        \item There is some guard $\widehat{g}$ found in the next revolution, which is below its pivot line $\widehat{\ell}$.
    \end{enumerate}
\end{proposition}

\begin{proof}
    Consider the triangle $\triangle gC_{z,i}C_{y,i}$. Let $x$ be a point on $\partial P \cap \triangle gC_{z,i} C_{y,i}$ which is furthest below $\ell_i$ (it may be that $x = C_{z,i}$ or $x=C_{y,i}$). During the next $k+1$ greedy steps, we will at some point have a guard, that sees $x$ as part of its interval. If $x$ is the endpoint of such an interval, we have an edge jump. So we assume $x$ is not an endpoint. Next, notice that $x \in [C_{z,i}, C_{y,i}]$, since if $x \in [y_i', z_i]$, $g_i$ would not be able to see $y_i'$ and $z_i$ and if $x \in [z_i, C_{z,i}]$ or $[C_{y,i}, y_i']$ $x \in [z_i, C_{z,i}]$ or $x\in[C_{y,i}, y_i']$, that piece of the boundary would need to enter $Q_i$ contradicting Lemma~\ref{lem:G_is_line_of_sight_is_clear}.

    \begin{figure}[ht]
        \centering

        \begin{tikzpicture}
            \draw[color=black] (0,4) -- (0,1) -- (1,1.2) -- (2,1) -- (3,1.5) -- (4,1.5) -- (5,2) -- (4.5, 4.5);
            \filldraw[qi] (0,3.06) -- (0,1) -- (1,1.2) -- (2,1) -- (3,1.5) -- (4,1.5) -- (5,2) -- (4.7, 3.48) -- (2.43, 1.76) -- cycle;
            \filldraw[color = gray, fill = gray!20] (2.43, 1.76) -- (1, 2.525) -- (4, 2.95) -- cycle;
            
            \draw[color=black] (0.86, 4.05) -- (1.51, 2.76) -- (1, 2.525) -- (1.8, 2.4) -- (2.08, 3.09) -- (1.88, 4.26);
            \draw[color=black] (3.5, 4.4) -- (4, 2.95) -- (3.8, 4.2);

            \draw[guard, dashed] (0,3.06) -- (2.43, 1.76) -- (4.7, 3.48);
    
            \filldraw[color = black] (4.77, 3.17) circle (1.5pt);
            \draw[color = black, right] (4.77, 3.17) node {$y_i$}; 
            \filldraw[color = black] (4.7, 3.48) circle (1.5pt);
            \draw[color = black, right] (4.7, 3.48) node {$y'_i$};
    
            \filldraw[color = black] (0,3.06) circle (1.5pt);
            \draw[color = black, left] (0,3.06) node {$z_i$};
            
            \filldraw (1, 2.525) circle (1.5pt);
            \draw (1, 2.525) node[above] {$C_{z,i}$};
            
            \filldraw (4, 2.95) circle (1.5pt);
            \draw (4, 2.95) node[below right=-4pt] {$C_{y,i}$};
            
            \filldraw[color=black] (1.8, 2.4) circle (1.5pt);
            \draw[color=black] (1.8, 2.4) node[below right=-2pt] {$x$}; 
            
            \filldraw[color = guard] (2.43, 1.76) circle (1.5pt);
            \draw[color = guard, below] (2.43, 1.76) node {$g_i$};
            
            \draw[color = black, right] (2.3, 3) node {$\ell_i$};
        \end{tikzpicture}
        
        \caption{A point $x\in\partial P$ furthest below $\ell_i$ must lie in the triangle $\triangle g_iC_{z,i} C_{y,i}$.}
        \label{fig:where_to_place_x}
    \end{figure}
    
    Let $\widehat{g}$ be the guard that sees $x$. We introduce the same terminology around $\widehat{g}$ as for $g_i$: $\widehat{g}$ guards $[\widehat{y},\widehat{z}]$ and $\widehat{g}$ is blocked by $\widehat{C}_{z}$ and $\widehat{C}_y$ with $\widehat{y}\,'$ being the furthest $\widehat{g}$ can see when going backwards. Let $\widehat{\ell}= L(\widehat{C}_y, \widehat{C}_z)$ be the pivot line associated with $\widehat{g}$. 

    In the following, we will consider where $\widehat{g}, \widehat{C}_{z}$ and $\widehat{C}_y$ can be placed in relation to each other. We will see that in all these cases one of the three criteria in the proposition will be satisfied.
    
    Let $\widehat{F}$ be the feasible region for the vertex $x$ and the edges of $P$ connected to $x$. $\widehat{F}$ is obtained by extending the edges connected to $x$ as in Figure~\ref{fig:G_hat}.
    
    Finally, let $z$ be the vertex on $[y_i, z_i]$ just before $z_i$.
    
    \begin{figure}[ht]
        \centering

        \begin{tikzpicture}
            \draw[color=black] (0,4) -- (0,1) -- (1,1.2) -- (2,1) -- (3,1.5) -- (4,1.5) -- (5,2) -- (4.5, 4.5);
            \filldraw[qi] (0,3.06) -- (0,1) -- (1,1.2) -- (2,1) -- (3,1.5) -- (4,1.5) -- (5,2) -- (4.7, 3.48) -- (2.43, 1.76) -- cycle;
            
            \draw[color=black] (0.86, 4.05) -- (1.51, 2.76) -- (1, 2.525) -- (1.8, 2.4) -- (2.08, 3.09) -- (1.88, 4.26);
            \draw[color=black] (3.5, 4.4) -- (4, 2.95) -- (3.8, 4.2);
            
            \filldraw[color = black] (0,3.06) circle (1.5pt);
            \draw[color = black, left] (0,3.06) node {$z_i$};
            
            \filldraw[F] (1.8,2.4) -- (4.85, 1.92) -- (4,1.5) -- (3,1.5) -- (2,1) -- (1.29, 1.14) -- cycle;
            
            \filldraw (1,2.525) circle (1.5pt);
            \draw (1, 2.525) node[above] {$C_{z,i}$};
            
            \filldraw[color=black] (1.8, 2.4) circle (1.5pt);
            \draw[color = black, right] (1.8, 2.6) node {$x$}; 
            \filldraw[color = black] (4.77, 3.17) circle (1.5pt);
            \draw[color = black, right] (4.77, 3.17) node {$y_i$}; 
            \filldraw[color = black] (4.7, 3.48) circle (1.5pt);
            \draw[color = black, right] (4.7, 3.48) node {$y'_i$};
            
            \filldraw[color = guard] (2.43, 1.76) circle (1.5pt);
            \draw[color = guard, below] (2.43, 1.76) node {$g_i$};
            
            \filldraw[guard] (3, 1.8) circle (1.5pt);
            \draw[right, guard] (3,1.8) node {$\widehat{g}$};
    
            \draw[color = feasible region,right] (2.6,1.1) node {$\widehat{F}$};
            
            \filldraw (0,1) circle (1.5pt);
            \draw[left] (0,1) node {$z$};
        \end{tikzpicture}
        
        \caption{The guard $\widehat{g}$ that guards $x$ must lie in the feasible region $\widehat{F}$ of the two edges that share $x$ as an endpoint.}
        \label{fig:G_hat}
    \end{figure}
    
    Since $x$ is chosen lowest in $\triangle C_{z,i}C_{y,i}g_i\cap \partial P$, both edges in $\partial P$ connected to $x$ will point downwards (where $\ell_i$ is horizontal). Thus $\widehat{F}$, whose boundary is the continuation of these edges, will be contained below $\ell_i$. 
    
    We now look at the case, where $\widehat{g}$ is not in $Q_i$:
    
    If $\widehat{g}$ is not in $Q_i$, it lies above $L(y_i',g_i)$ and $L(g_i, z_i)$, as they are extensions of borders of $Q_i$. From the considerations above, we know that $\widehat{g}$ must also lie below $\ell_i$, thus $\widehat{g} \in \triangle g_iC_{z,i}C_{y,i}$.   
    Furthermore, we know that $x$ is the lowest point of $\partial P$ inside $\triangle C_{z,i}C_{y,i}g_i$, and $\widehat{g}$ is below $x$, thus below $\widehat{g}$ there is no part of $\partial P$ inside $\triangle gC_{z,i} C_{y,i}$. 
    
    We now assume for contradiction, that on of the pivots for $\widehat{g}$, $\widehat{C}_y$ or $\widehat{C}_z$, lies below $\widehat{g}$. Then $\widehat{y}\,'$ or $\widehat{z} \in [y_i', z_i]$, since $\widehat{y}\,'$, respectively $\widehat{z}$, lie on the ray $\ray{\widehat{g}}{\widehat{C}_y}$, respectively the ray $\ray{\widehat{g}}{\widehat{C}_z}$, and these rays will be contained in the region of $\triangle g_i C_{z,i} C_{y,i}$ below $\widehat{g}$ and $Q_i$ (until they hit $\partial P$), where no part of $[z_i, y_i']$ can enter by Lemma~\ref{lem:G_is_line_of_sight_is_clear} (see Figure~\ref{fig:G_hat_not_in_Qi} left).
    
    \begin{figure}[ht]
        \centering

        \begin{tikzpicture}
            \begin{scope}
                \draw[color=black] (0,4) -- (0,1) -- (1,1.2) -- (2,1) -- (3,1.5) -- (4,1.5) -- (5,2) -- (4.5, 4.5);
                \filldraw[qi] (0,3.06) -- (0,1) -- (1,1.2) -- (2,1) -- (3,1.5) -- (4,1.5) -- (5,2) -- (4.7, 3.48) -- (2.43, 1.76) -- cycle;
                \filldraw[color = gray, fill = gray!20] (2.43, 1.76) -- (1, 2.525) -- (4, 2.95) -- cycle;
    
                \draw[color=black] (0.86, 4.05) -- (1.51, 2.76) -- (1, 2.525) -- (1.8, 2.4) -- (2.08, 3.09) -- (1.88, 4.26);
                \draw[color=black] (3.5, 4.4) -- (4, 2.95) -- (3.8, 4.2);
    
                \filldraw[color = black] (4.77, 3.17) circle (1.5pt);
                \draw[color = black, right] (4.77, 3.17) node {$y_i$}; 
                \filldraw[color = black] (4.7, 3.48) circle (1.5pt);
                \draw[color = black, right] (4.7, 3.48) node {$y'_i$};
    
                \filldraw[color = black] (0,3.06) circle (1.5pt);
                \draw[color = black, left] (0,3.06) node {$z_i$};
    
                \filldraw (1,2.525) circle (1.5pt);
                \draw (1, 2.525) node[above] {$C_{z,i}$};
    
                \filldraw (4, 2.95) circle (1.5pt);
                \draw (3.5, 2.95) node[above] {$C_{y,i}$};
    
                \filldraw[color = guard] (2.43, 1.76) circle (1.5pt);
                \draw[color = guard, below] (2.43, 1.76) node {$g_i$};
    
                \filldraw[guard] (2.4,2.1) circle (1.5pt);
                \draw[guard] (2.4,2.1) node[above] {$\widehat{g}$};
    
                \filldraw[color=black] (0.6, 1.7) circle (1.5pt);
                \draw (0.6, 1.7) node[above] {$\widehat{C}_{y}$};
    
                \draw[dashed, guard] (2.4,2.1) -- (0,1.56);
                \filldraw (0,1.56) circle (1.5pt);
                \draw (0,1.56) node[left] {$\widehat{y}\,'$};
    
                \filldraw (0,1) circle (1.5pt);
                \draw[left] (0,1) node {$z$};
            \end{scope}
    
            \begin{scope}[xshift=7cm]
                \draw[color=black] (0,4) -- (0,1) -- (1,1.2) -- (2,1) -- (3,1.5) -- (4,1.5) -- (5,2) -- (4.5, 4.5);
                \filldraw[qi] (0,3.06) -- (0,1) -- (1,1.2) -- (2,1) -- (3,1.5) -- (4,1.5) -- (5,2) -- (4.7, 3.48) -- (2.43, 1.76) -- cycle;
                \filldraw[color = gray, fill = gray!20] (2.43, 1.76) -- (1, 2.525) -- (4, 2.95) -- cycle;
    
                \draw[color=black] (0.86, 4.05) -- (1.51, 2.76) -- (1, 2.525) -- (1.8, 2.4) -- (2.08, 3.09) -- (1.88, 4.26);
                \draw[color=black] (3.5, 4.4) -- (4, 2.95) -- (3.8, 4.2);
    
                \filldraw[color = black] (4.77, 3.17) circle (1.5pt);
                \draw[color = black, right] (4.77, 3.17) node {$y_i$}; 
                \filldraw[color = black] (4.7, 3.48) circle (1.5pt);
                \draw[color = black, right] (4.7, 3.48) node {$y'_i$};

                \filldraw (1,2.525) circle (1.5pt);
                \draw (1, 2.6) node[above] {$C_{z,i}$};
    
                \filldraw (4, 2.95) circle (1.5pt);
                \draw (3.5, 2.95) node[above] {$C_{y,i}$};
    
                \filldraw[color = guard] (2.43, 1.76) circle (1.5pt);
                \draw[color = guard, below] (2.43, 1.76) node {$g_i$};

                \draw[dashed, guard] (2.4,2.1) -- (0,2.83);
                \filldraw (0,2.83) circle (1.5pt);
    
                \draw[dashed, red, thick] (2.4,2.1) -- (0,3.06);
    
                \filldraw[color = black] (0,3.06) circle (1.5pt);
                \draw[color = black, left] (0,3.06) node {$z_i$};
                
                \filldraw[guard] (2.4,2.1) circle (1.5pt);
                \draw[guard] (2.4,2.1) node[above] {$\widehat{g}$};
                
                \filldraw (0,1) circle (1.5pt);
                \draw[left] (0,1) node {$z$};
            \end{scope}
        \end{tikzpicture}
        
        \caption{Assuming $\widehat{g}$ is not in $Q_i$ and $\widehat{C}_y$ below $\widehat{g}$, we must have $\widehat{y}$ in $[y'_i, z_i]$ (left). However $C_{y,i}$ blocks the view for $\widehat{g}$ to see $z_i$ (right), so $[\widehat{y},x]$ cannot both be visible from $\widehat{g}$.}
        \label{fig:G_hat_not_in_Qi}
    \end{figure}
    
    This implies that we would need $\widehat{g}$ to be able to see all of $[z_i, x]$ or $[x, y_i']$. However this is not possible, as $\widehat{g}$ is strictly above the lines $L(C_{y,i}, g_i)$ and $L(g_i, C_{z,i})$ and thus the pivots $C_{y,i}$ and $C_{z,i}$ will block $\widehat{g}$ from seeing $y_i'$ and $z_i$. This yields a contradicition. 

    Now both $\widehat{C}_y$ and $\widehat{C}_z$ are above $\widehat{g}$. From the assumption that $\widehat{g}$ sees $x$, we know that $x\in[\widehat{y}\,', \widehat{z}]$, thus $\widehat{C}_y$ will lie to the left of $\linesegment{\widehat{g}}{x}$ and $\widehat{C}_z$ to the right. So now $\widehat{\ell}$ must lie below $\widehat{g}$ and we have showed what we wanted (see Figure~\ref{fig:G_hat_not_in_Qi_and_Cy_above_G}).
    
    \begin{figure}
        \centering
        \begin{tikzpicture}
            \draw[color=black] (0,4) -- (0,1) -- (1,1.2) -- (2,1) -- (3,1.5) -- (4,1.5) -- (5,2) -- (4.5, 4.5);
            \filldraw[qi] (0,3.06) -- (0,1) -- (1,1.2) -- (2,1) -- (3,1.5) -- (4,1.5) -- (5,2) -- (4.7, 3.48) -- (2.43, 1.76) -- cycle;
            \filldraw[color = gray, fill = gray!20] (2.43, 1.76) -- (1, 2.525) -- (4, 2.95) -- cycle;
    
            \draw[color=black] (0.86, 4.05) -- (1.51, 2.76) -- (1, 2.525) -- (1.8, 2.4) -- (2.08, 3.09) -- (1.88, 4.26);
            \draw[color=black] (3.5, 4.4) -- (4, 2.95) -- (3.8, 4.2);
    
            \filldraw[color = black] (4.77, 3.17) circle (1.5pt);
            \draw[color = black, right] (4.77, 3.17) node {$y_i$}; 
            \filldraw[color = black] (4.7, 3.48) circle (1.5pt);
            \draw[color = black, right] (4.7, 3.48) node {$y'_i$};
    
            \filldraw[color = black] (0,3.06) circle (1.5pt);
            \draw[color = black, left] (0,3.06) node {$z_i$};
    
            \filldraw (1, 2.525) circle (1.5pt);
            \draw (1, 2.525) node[above] {$C_{z,i}$};
    
            \filldraw (4, 2.95) circle (1.5pt);
            \draw (3.5, 2.95) node[above] {$C_{y,i}$};
    
            \filldraw[color=black] (0.6, 2.4) circle (1.5pt);
            \draw (0.6, 2.4) node[left] {$\widehat{C}_{y}$};
    
            \filldraw[color = black] (4.2,2.6) circle (1.5pt);
            \draw[color = black, right] (4.2,2.8) node {$\widehat{C}_z$};
    
            \draw[dashed] (0.6, 2.4) -- (4.2,2.6);
            \draw[thick] (0,1.76) -- (4.91, 2.46);
            \draw[thick] (2.4,2.1) -- (1.8,2.4);

            \filldraw[color = guard] (2.43, 1.76) circle (1.5pt);
            \draw[color = guard, below] (2.43, 1.76) node {$g_i$};
    
            \filldraw[guard] (2.4,2.1) circle (1.5pt);
            \draw[guard] (2.42,2.2) node[above] {$\widehat{g}$};
    
            \filldraw (1.8,2.4) circle (1.5pt);
            \draw (1.7,2.4) node[above] {$x$};
            \filldraw (0,1) circle (1.5pt);
            \draw[left] (0,1) node {$z$};
        \end{tikzpicture}
        
        \caption{If $\widehat{C}_y$ and $\widehat{C}_z$ are above $\widehat{g}$, and $x \in [\widehat{y}\,', \widehat{z}]$ then we must have $\widehat{g}$ above $\widehat{\ell}$.}
        \label{fig:G_hat_not_in_Qi_and_Cy_above_G}
    \end{figure}

    Now we assume that $\widehat{g}$ lies in $Q_i$, and again we consider, where $\widehat{C}_y$ and $\widehat{C}_z$ can be located. Since $x$ is in $[\widehat{y}\,', \widehat{z}]$, we cannot have $\widehat{z} \in (z_i, x)$, as this would contradict the maximality of greedy step from $y_i$ to $z_i$. Likewise we cannot have $\widehat{y}\,' \in (x, y_i)$, again due to maximality. The possible placements of $\widehat{z}$ and $\widehat{y}$, along with the possible placements of $\widehat{C}_z$ and $\widehat{C}_y$ are marked on Figure~\ref{fig:place_z_and_y_hat}. 

    \begin{figure}[ht]
        \centering

        \begin{tikzpicture}
            \begin{scope}
                \filldraw[extra] (1.7, 1.8) -- (1.8, 2.4) -- (2.08, 3.09) -- (1.88, 4.26) -- (3.5, 4.4) -- (4, 2.7) -- (3.8, 4.2) --  (4.5, 4.5) -- (4.79, 3.07) -- (5,2) -- (4,1.5) -- (3,1.5) -- (2,1) -- (1,1.2) -- (0,1) -- (0,3.06) -- cycle;
                
                \draw (0,3.06) -- (0,4);
                \draw[color=black] (0.86, 4.05) -- (1.51, 2.76) -- (1, 2.525) -- (1.8, 2.4);
                \draw[color = extracolor] (1.8, 2.4) -- (2.08, 3.09) -- (1.88, 4.26);
                \draw[color=extracolor] (0,3.06) -- (0,1) -- (1,1.2) -- (2,1) -- (3,1.5) -- (4,1.5) -- (5,2) -- (4.5, 4.5);
    
                \draw[color = extracolor, dashed] (0,3.06) -- (1.7, 1.8) -- (1.8, 2.4);
    
                \draw[color=extracolor] (3.5, 4.4) -- (4, 2.7) -- (3.8, 4.2);
    
                \filldraw[color = black] (4.77, 3.17) circle (1.5pt);
                \draw[color = black, right] (4.77, 3.17) node {$y_i$};  
                \filldraw[color = black] (4.7, 3.48) circle (1.5pt);
                \draw[color = black] (4.7, 3.48) node[above right=-2pt] {$y'_i$};
    
                \filldraw[color = black] (0,3.06) circle (1.5pt);
                \draw[color = black, left] (0,3.06) node {$z_i$};
    
                \filldraw (1.8,2.4) circle (1.5pt);
                \draw (1.8,2.4) node[right] {$x$};
    
                \filldraw[guard] (1.7, 1.8) circle (1.5pt);
                \draw[guard] (1.7, 1.8) node[below] {$\widehat{g}$};
                \filldraw (0,1) circle (1.5pt);
                \draw[left] (0,1) node {$z$};
            \end{scope}
    
            \begin{scope}[xshift=7cm]
                \filldraw[extra] (1.7, 1.8) --  (4.77, 3.17) -- (5,2) -- (4,1.5) -- (3,1.5) -- (2,1) -- (1,1.2) -- (0,1) -- (0,4) -- (0.86, 4.05) -- (1.51, 2.76) -- (1, 2.525) -- (1.8, 2.4) -- cycle;
                
                \draw (4.5,4.5) --  (4.77, 3.17);
                \draw[color=extracolor] (0.86, 4.05) -- (1.51, 2.76) -- (1, 2.525) -- (1.8, 2.4);
                \draw[color = black] (1.8,2.4) -- (2.08, 3.09) -- (1.88, 4.26);
                \draw[color=extracolor] (0,4) -- (0,1) -- (1,1.2) -- (2,1) -- (3,1.5) -- (4,1.5) -- (5,2) --  (4.77, 3.17);
    
                \draw[color = extracolor, dashed]  (4.77, 3.17) -- (1.7, 1.8) -- (1.8, 2.4);
    
                \draw[color=black] (3.5, 4.4) -- (4, 2.75) -- (3.8, 4.2);
    
                \filldraw[color = black] (4.77, 3.17) circle (1.5pt);
                \draw[color = black, right] (4.77, 3.17) node {$y_i$};   
                \filldraw[color = black] (4.7, 3.48) circle (1.5pt);
                \draw[color = black] (4.7, 3.48) node[above right=-2pt] {$y'_i$};
    
                \filldraw[color = black] (0,2.92) circle (1.5pt);
                \draw[color = black, left] (0,2.92) node {$z_i$};
                
                \filldraw (1.8,2.4) circle (1.5pt);
                \draw (1.8,2.4) node[right] {$x$};
                
                \filldraw[guard] (1.7, 1.8) circle (1.5pt);
                \draw[guard] (1.7, 1.8) node[below] {$\widehat{g}$};
                \filldraw (0,1) circle (1.5pt);
                \draw[left] (0,1) node {$z$};
            \end{scope}
        \end{tikzpicture}
        
        \caption{Left, $\widehat{z}$ must lie in $[x,z_i]$ and $\widehat{C}_z$ in the polygon defined by $[x,z_i]$, $\linesegment{z_i}{\widehat{g}}$ and $\linesegment{\widehat{g}}{x}$. Right, analogously for $\widehat{y}$ and $\widehat{C}_y$.}
        \label{fig:place_z_and_y_hat}
    \end{figure}

    \begin{case}
        If $\widehat{z} \in (y_i', y_i]$, then $\widehat{z}$ will be in the same local greedy sequence as $y_i$, hence $\widehat{z} = y_{i+1}$, but now $g_i$ sees $[y_{i+1}, z_i]$, thus we must repeat endpoints and the greedy sequence repeats, which is one of the desired conditions.
    \end{case}
    
    \begin{case}
        If $\widehat{z} \in (y_i, z_i]$, then $\widehat{z}$ must be in the same local greedy sequence as $z_i$, i.e. $\widehat{z} = z_{i+1}$, implying that $\widehat{y}$ must be in the same local greedy sequence as $y_i$, so $\widehat{y} = y_{i+1}$, but $\widehat{y}$ is in $[y_i', C_{z,i}]$, so $y_{i+1}$ either jumps edges, has a positive fingerprint or $z_{i+1} = z_i$, a repetition.
    \end{case}
    
    \begin{case}
        Now $\widehat{z} \in [x, y_i']$. Consider the subcases where $g_i$ lies above or below $L(y_i', \widehat{g})$: If $g_i$ lies above, $\widehat{C}_z$ must lie above $L(y_i', \widehat{g})$, since we otherwise would have the ray $\ray{\widehat{g}}{\widehat{C}_z}$ contained in $Q_i$ (until it hits $\partial P$), since $\widehat{g} \in Q_i$, making $\widehat{z} \in (y_i', z_i]$. If $g_i$ lies below $L(y_i',\widehat{g})$, we must have $\widehat{C}_z$ above $L(y_i', g_i)$ and $L(g_i, \widehat{g})$ for the same reason (see Figure~\ref{fig:place_Cz_hat_depending_on_G_i}).     
    
        \begin{figure}[ht]
            \centering

            \begin{tikzpicture}
                \begin{scope}
                    \filldraw[extra] (1.67,1.88) --  (1.8, 2.4) -- (2.08, 3.09) -- (1.88, 4.26) --  (3.5, 4.4) -- (4, 2.95) -- (3.8, 4.2) --  (4.5, 4.5) -- (4.7, 3.48) -- (4.75, 3.24) -- cycle;
                    
                    \draw (0,2.92) -- (0,4);
                    \draw[color=black] (0,2.92) -- (0,1) -- (1,1.2) -- (2,1) -- (3,1.5) -- (4,1.5) -- (5,2) -- (4.75, 3.24);
                    \draw[extracolor] (4.5, 4.5) -- (4.75, 3.24);
                    \draw[color = extracolor, dashed] (4.75, 3.24)  -- (1.67,1.88) -- (1.8,2.4);
        
                    \draw[color=black] (0.86, 4.05) -- (1.51, 2.76) -- (1, 2.525) -- (1.8, 2.4);
                    \draw[color = extracolor] (1.8, 2.4) -- (2.08, 3.09) -- (1.88, 4.26);
                    
                    \draw[color=extracolor] (3.5, 4.4) -- (4, 2.95) -- (3.8, 4.2);
        
                    \filldraw[color = black] (4.83, 2.84) circle (1.5pt);
                    \draw[color = black, right] (4.83, 2.84) node {$y_i$};   
                    \filldraw[color = black] (4.75, 3.24) circle (1.5pt);
                    \draw[color = black, right] (4.75, 3.44) node {$y'_i$};
        
                    \filldraw[color = black] (0,2.89) circle (1.5pt);
                    \draw[color = black, left] (0,2.89) node {$z_i$};
                
                    \filldraw[color = guard] (1.98, 2.17) circle (1.5pt);
                    \draw[color = guard, right] (1.88, 2.37) node {$g_i$};
                    \filldraw (1.8,2.4) circle (1.5pt);
                    \draw (1.7,2.4) node[above] {$x$};
        
                    \filldraw[guard] (1.67,1.88) circle (1.5pt);
                    \draw[guard] (1.67,1.88) node[left] {$\widehat{g}$};
                    \filldraw (0,1) circle (1.5pt);
                    \draw[left] (0,1) node {$z$};
                \end{scope}
        
                \begin{scope}[xshift=7cm]
                    \filldraw[extra] (1.7, 1.8) --  (1.8, 2.4) -- (2.08, 3.09) -- (1.88, 4.26) --  (3.5, 4.4) -- (4, 2.9) -- (3.8, 4.2) --  (4.5, 4.5) -- (4.7, 3.48) -- (2.43, 1.76) -- cycle;
                    
                    \draw (0,2.92) -- (0,4);
                    \draw[color=black] (0,2.92) -- (0,1) -- (1,1.2) -- (2,1) -- (3,1.5) -- (4,1.5) -- (5,2) -- (4.7, 3.48);
                    \draw[extracolor] (4.5, 4.5) -- (4.7, 3.48);
                    \draw[color = extracolor, dashed] (4.7, 3.48) -- (2.43, 1.76) -- (1.7, 1.8) -- (1.8,2.4);
        
                    \draw[color=black] (0.86, 4.05) -- (1.51, 2.76) -- (1, 2.525) -- (1.8, 2.4);
                    \draw[color = extracolor] (1.8, 2.4) -- (2.08, 3.09) -- (1.88, 4.26);
                    
                    \draw[color=extracolor] (3.5, 4.4) -- (4, 2.95) -- (3.8, 4.2);
        
                    \filldraw[color = black] (4.77, 3.17) circle (1.5pt);
                    \draw[color = black, right] (4.77, 3.17) node {$y_i$};   
                    \filldraw[color = black] (4.7, 3.48) circle (1.5pt);
                    \draw[color = black] (4.7, 3.48) node[above right=-2pt] {$y'_i$};
        
                    \filldraw[color = black] (0,3.06) circle (1.5pt);
                    \draw[color = black, left] (0,3.06) node {$z_i$};
                
                    \filldraw[color = guard] (2.43, 1.76) circle (1.5pt);
                    \draw[color = guard, below] (2.43, 1.76) node {$g_i$};
                    \filldraw (1.8,2.4) circle (1.5pt);
                    \draw (1.8,2.4) node[right] {$x$};
        
                    \filldraw[guard] (1.7, 1.8) circle (1.5pt);
                    \draw[guard] (1.7, 1.8) node[left] {$\widehat{g}$};
                    \filldraw (0,1) circle (1.5pt);
                    \draw[left] (0,1) node {$z$};
                \end{scope}
            \end{tikzpicture}
            
            \caption{Further restrictions to Figure~\ref{fig:place_z_and_y_hat}. On the left, the possible placement of $\widehat{z}$ and $\widehat{C}_z$ are marked in red, when $g_i$ lies above $\lineextension{y_i'}{\widehat{g}}$. On the right, the possible placements when $g_i$ lies below $\lineextension{y_i'}{\widehat{g}}$}
            \label{fig:place_Cz_hat_depending_on_G_i}
        \end{figure}
    
        We now look at where we can place $\widehat{C}_y$ so that $\widehat{g}$ lies below $\widehat{\ell}$ in both of these cases:
        
        \begin{subcase}[$g_i$ above $L(y'_i, \widehat{g})$]
            If $g_i$ lies above $L(y'_i, \widehat{g})$, we must place $\widehat{C}_y$ below $L(y_i', \widehat{g})$ for $\widehat{\ell}$ to lie above $\widehat{g}$. Combined with the previous restrictions, we see on Figure~\ref{fig:G_i_over_L_where_to_put_Cy}, the possible placements for $\widehat{C}_y$.
        
            \begin{figure}
                \centering
                
                \begin{tikzpicture}
                    \filldraw[extra] (1.67,1.88) --  (0,1.15) -- (0,1) -- (1,1.2) -- (2,1) -- (3,1.5) -- (4,1.5) -- (5,2) -- (4.83, 2.84) -- cycle;
                    
                    \draw (0,1.15) -- (0,4);
                    \draw[color=extracolor] (0,1.15) -- (0,1) -- (1,1.2) -- (2,1) -- (3,1.5) -- (4,1.5) -- (5,2) -- (4.83, 2.84);
                    \draw[black] (4.5, 4.5) -- (4.83, 2.84);
                    \draw[color = extracolor, dashed] (4.83, 2.84)  -- (1.67,1.88) -- (0,1.15);
                    \draw[color = extracolor, loosely dashdotted] (1.67,1.88) -- (4.75, 3.24);
            
                    \draw[color=black] (0.86, 4.05) -- (1.51, 2.76) -- (1, 2.525) -- (1.8, 2.4) -- (2.08, 3.09) -- (1.88, 4.26);
                    
                    \draw[color=black] (3.5, 4.4) -- (4, 2.95) -- (3.8, 4.2);
            
                    \filldraw[color = black] (4.83, 2.84) circle (1.5pt);
                    \draw[color = black, right] (4.83, 2.84) node {$y_i$};  
                    \filldraw[color = black] (4.75, 3.24) circle (1.5pt);
                    \draw[color = black, right] (4.75, 3.44) node {$y'_i$};
            
                    \filldraw[color = black] (0,2.89) circle (1.5pt);
                    \draw[color = black, left] (0,2.89) node {$z_i$};
                
                    \filldraw[color = guard] (1.98, 2.17) circle (1.5pt);
                    \draw[color = guard, right] (1.88, 2.37) node {$g_i$};
            
                    \filldraw[guard] (1.67,1.88) circle (1.5pt);
                    \draw[guard] (1.67,1.88) node[above] {$\widehat{g}$};
                    
                    \filldraw (0,1) circle (1.5pt);
                    \draw[left] (0,1) node {$z$};
                \end{tikzpicture}
                \caption{When $g_i$ is over $\lineextension{y_i'}{\widehat{g}}$ then $\widehat{y}$ and $\widehat{C}_y$ are restricted by $y_i$ and $\lineextension{\widehat{g}}{y_i'}$} \label{fig:G_i_over_L_where_to_put_Cy}
            \end{figure}
            
            Now the ray $\ray{\widehat{g}}{\widehat{C}_y}$ will be contained in $Q_i$, so $\widehat{y}\,' \in [y_i, z_i]$. This means, that $\widehat{y}$ will be in the same local greedy sequence as $z_i$, thus $\widehat{y} = z_i$. If $\widehat{y}\,' \in [y_i, z]$, either $z_{i+1} \in [y_i, \widehat{y}\,']$ and we get an edge jump, or $z_{i+1} \in [\widehat{y}\,', z_i]$ and we will repeat since $\Gsc(z_i) = \Gsc(z_{i+1})$, both are terminal conditions of the proposition. Thus, it remains to analyze the case $\widehat{y}\,' \in (z, z_i]$. 
        
            We have $y_i'$ above $\ell_i$ and $\widehat{g}$ below $\ell_i$. Thus $\widehat{y}\,'$ will also lie below $\ell_i$. If $z_{i+1} \in [\widehat{y}\,', z_i]$ we will again get a repetition, so we need $z_{i+1}\in [z, \widehat{y}]$, and $z_{i+1}$ is below $\ell_i$, which is the third terminal condition of the proposition.
        \end{subcase}
    
        \begin{subcase}[$g_i$ below $L(y'_i, \widehat{g})$]
            If $g_i$ lies below $L(y_i', \widehat{g})$ we again consider where to place $\widehat{C}_y$. To the right of $\widehat{g}$, it must be below $L(y_i, \widehat{g})$ by the previous arguments. To the left of $\widehat{g}$ it must be below $L(g_i, \widehat{g})$ for $\widehat{g}$ to lie below $\widehat{\ell}$ (see Figure~\ref{fig:G_i_below_L_where_to_put_Cy}).
    
            \begin{figure}[ht]
                \centering
                
                \begin{tikzpicture}
                    \filldraw[extra] (1.7, 1.8) -- (0,1.88) -- (0,1) -- (1,1.2) -- (2,1) -- (3,1.5) -- (4,1.5) -- (5,2) -- (4.77, 3.17) -- cycle;
                        
                    \draw (0,1.88) -- (0,4);
                    \draw[color=extracolor] (0,1.88) -- (0,1) -- (1,1.2) -- (2,1) -- (3,1.5) -- (4,1.5) -- (5,2) -- (4.77, 3.17);
                    \draw[black] (4.5, 4.5) -- (4.77, 3.17);
                    \draw[color = extracolor, dashed] (0,1.88) -- (1.7,1.8) -- (4.77, 3.17);
                
                    \draw[color = extracolor, dashdotted] (2.43, 1.76) -- (1.7, 1.8);
                
                    \draw[color=black] (0.86, 4.05) -- (1.51, 2.76) -- (1, 2.525) -- (1.8, 2.4) -- (2.08, 3.09) -- (1.88, 4.26);
                    
                    \draw[color=black] (3.5, 4.4) -- (4, 2.95) -- (3.8, 4.2);
                
                    \filldraw[color = black] (4.77, 3.17) circle (1.5pt);
                    \draw[color = black, right] (4.77, 3.17) node {$y_i$};  
                    \filldraw[color = black] (4.7, 3.48) circle (1.5pt);
                    \draw[color = black] (4.7, 3.48) node[above right=-2pt] {$y'_i$};
                
                    \filldraw[color = black] (0,3.06) circle (1.5pt);
                    \draw[color = black, left] (0,3.06) node {$z_i$};
                    
                    \filldraw[color = guard] (2.43, 1.76) circle (1.5pt);
                    \draw[color = guard, below] (2.43, 1.76) node {$g_i$};
                    
                    \filldraw[guard] (1.7, 1.8) circle (1.5pt);
                    \draw[guard] (1.7, 1.8) node[below] {$\widehat{g}$};
                
                    \filldraw (0,1.88) circle (1.5pt);
                    \draw (0,1.88) node[left] {$z_{\max}$};
                    \filldraw (0,1) circle (1.5pt);
                    \draw[left] (0,1) node {$z$};
                \end{tikzpicture}
                
                \caption{Case where $g_i$ is below $\lineextension{y_i'}{\widehat{g}}$: Where to put $\widehat{y}$ and $\widehat{C}_y$}
                \label{fig:G_i_below_L_where_to_put_Cy}
            \end{figure}
            
            Since $\widehat{g}$ lies in $Q_i$, we must have that the ray $\ray{\widehat{g}}{g_i}$ is contained in $Q_i$. Let the intersection of this ray with $l_z$ be $z_{\max}$. Thus $\widehat{y}\,' \in [y_i, z_{\max}]$ and we again have $\widehat{y} = z_i$, likewise, if $\widehat{y}\,' \in [y_i, z]$ we will get an edge jump or repetition as in case 3.1. So we assume $\widehat{y} \in (z, z_i]$.
            
            Now we must have $\widehat{C}_y$ inside the triangle $\triangle \widehat{g}z_{\max} z$, which is contained in $Q_i$. Thus $\widehat{C}_y \in [y_i, z_i]$. Since $\widehat{C}_y$ lies below $L(g_i, \widehat{g})$ and $\widehat{y}\,'$ lies on the ray $\widehat{g}\widehat{C}_y$, we must have that $\widehat{y}\,'$ lies below $L(g_i, \widehat{g})$ and since $\widehat{C}_y \in \linesegment{\widehat{g}}{\widehat{y}\,'}$, $\widehat{C}_y$ must lie above $L(g_i, \widehat{y}\,')$. It is also clear, that $z$ must lie below $L(g_i, \widehat{y}\,')$, but $\widehat{C}_y$ and $z$ are connected by edges in $\partial P$. This now means, that $g_i$ cannot see $\widehat{y}\,'$, contradicting the fact that $g_i$ is a guard that sees $[y_i, z_i]$ (see Figure~\ref{fig:Cy_blocks_Gi_view}). 
     
            \begin{figure}[ht]
                \centering
                
                \begin{tikzpicture}[scale = 2]
                    \filldraw[extra] (1,1) -- (1.7, 1.95) -- (1, 2.13) -- cycle;
                    \draw[extracolor, dashed] (1,1) -- (1.7, 1.95) -- (1, 2.13);
                    \draw[extracolor, loosely dashdotted] (1.7, 1.95) -- (2.82, 1.72);
                    \draw (1,1) -- (1,3);
                    \draw[dashed] (2.82, 1.72) -- (1, 1.56);
                    
                    \filldraw (1,1) circle (0.75pt);
                    \draw[left] (1,1) node {$z$};
                    \draw[dashed] (1.7, 1.95) -- (1, 1.56);
            
                    \filldraw[black] (1, 1.56) circle (0.75pt);
                    \draw[left] (1, 1.56) node {$\widehat{y}\,'$};
            
                    \filldraw (1.28, 1.72) circle (0.75pt);
                    \draw[above] (1.28, 1.72) node {$\widehat{C}_y$};
            
                    \filldraw[guard] (2.82, 1.72) circle (0.75pt);
                    \draw[guard, below] (2.82, 1.72) node {$g_i$};
            
                    \filldraw[guard] (1.7, 1.95) circle (0.75pt);
                    \draw[above, guard] (1.7, 1.95) node {$\widehat{g}$};
            
                    \filldraw (1, 2.13) circle (0.75pt);
                    \draw[left] (1, 2.13) node {$z_{\max}$};
            
                    \filldraw (1,3) circle (0.75pt);
                    \draw[left] (1,3) node {$z_i$};
                    
                \end{tikzpicture}
                
                \caption{$\widehat{C}_y$ will block $g_i$'s view to $\widehat{y}\,'$}
                \label{fig:Cy_blocks_Gi_view}
            \end{figure}
        \end{subcase}
    \end{case}

    Considering all the above cases, we show that when $\widehat{g}$ lies below $\widehat{\ell}$ it implies either an edge jump, a repetition, a positive finger print or $z_{i+1}$ being below $\ell_i$, as wanted.
\end{proof}

We now have all the tools needed to prove Theorem~\ref{thm:repeatedgreedy_will_jump_edges_or_repeat}, which was the goal of this section. 

\thmJumpEdgesOrRepeat*

\begin{proof}
    Assume for contradiction that we see no progress conditions in the first $\Oh{kn^2}$ revolutions of \RepeatedGreedy. 
    Let $(x^m_i)_{i = 0}^\infty$ for $m = 1, \dots, k$ be the local greedy sequences. 
    All these are contained on their respective edges (since we see no edge jumps).
    Let $F^m$ be the feasible region for the vertices of $\partial P$ in the interval $[x_{m-1}, x_m]$.     
    Furthermore let $\ell^m_i$ be the pivot line for the pair $(x_i^{m-1}, x_i^m)$. 
    Consider the choice of $\ell^1_i$ for $i = 1,\dots,\Oh{k n^2}$ when \RepeatedGreedy makes $\Oh{k n^2}$ revolutions. 
    We assume for contradiction that we see no progress conditions and consider the following possible cases:

    \begin{case}
        If $F^1$ lies strictly above $\ell_i^1$, we cannot repeat $\ell_i^1$, as this leads to a progress condition (Proposition~\ref{prop:F_over_l_gives_repetition_or_edge_jumping}). 
        Thus this can occur at most $\Oh{n^2}$ times (i.e. once for each possible pivot line).
    \end{case}
    
    \begin{case}
        If $F^1$ lies at or below $\ell_i^1$ we consider whether it is the first time we have seen this pivot line or not:

        \begin{subcase}
            If it is the first time $\ell_i^1$ is a pivot line (i.e. $\ell_i^1 \neq \ell_j^1$ for all $j < i$), we cannot deduce anything. 
            However, this can happen at most $\Oh{n^2}$ times (i.e. once for each possible pivot line).
        \end{subcase}
        
        \begin{subcase}
            If it is not the first time, Lemma~\ref{lem:yj_and_zj_above_li} states that $x_{i-1}^k$ and $x_i^1$ lie above $\ell_i^1$. 
            Now the conditions of Proposition~\ref{prop:guards_below_pivots_implies_other_guards_above_pivots} are satisfied, so we either obtain a progress condition within one revolution, $x_{i+1}^1$ lies below $\ell_i^1$, or some guard $g_i^m$ that lies above $\ell_i^m$. We consider the last two cases.

            \begin{subsubcase}
                If $x_{i+1}^1$ lies below $\ell_i$, we cannot use $\ell_i$ again as this would break Lemma~\ref{lem:yj_and_zj_above_li}. Thus this case can happen at most $\Oh{n^2}$ times.
            \end{subsubcase}

            We have shown, that at most $\Oh{n^2}$ of the $\Oh{kn^2}$ $i$'s can fall into the above cases, hence we still have $\Oh{kn^2}$ $i$'s left for this final case:

            \begin{subsubcase}
                Finally some other $g_i^m$ lies above its pivot line $\ell_i^m$. 
                Since there are $k-1$ other feasible regions, and each of these has $\Oh{n^2}$ pivot lines, at some point, we will have chosen the same feasible region with the same pivot line pair twice, by the pigeonhole principle. Since the guard is above this pivot line, the entire feasible region will also be above by contraposition of Lemma~\ref{lem:better_to_have_Gi_below_pivot_line}. 
                Now we have seen the same pivot line twice with the same feasible region strictly above it, so Proposition~\ref{prop:F_over_l_gives_repetition_or_edge_jumping} yields a progress condition.
                \popQED
            \end{subsubcase}
        \end{subcase}
    \end{case}
\end{proof}

% --------------------------------------------------------------------
\subsection{Proof of main theorem}
\label{sec:Combining_it_all}
% --------------------------------------------------------------------

Finally, we prove Theorem~\ref{thm:CAG_is_in_P}:

\thmCAGIsInP*

\begin{proof}
    Running \RepeatedGreedy for $\Oh{k n^2} = \Oh{\opt n^2}$ revolutions guarantees a progress condition i.e. a positive fingerprint, a repetition, or an edge jump by Theorem~\ref{thm:repeatedgreedy_will_jump_edges_or_repeat}.
    Thus, repeating $n+1$ times guarantees at least one positive fingerprint, at least one repetition, or $n+1$ edge jumps. 
    A positive fingerprint or repeatetition implies optimality by Corollary~\ref{cor:positive_fingerprint_is_optimal} and \ref{cor:periodicity_implies_optimal}. 
    If we have $n+1$ edge jumps, we are optimal by Corollary~\ref{cor:too_many_edge_jumps_will_force_greedy_sequence_to_be_optimal}.
    Now after $\Oh{\opt n^3}$ revolutions \RepeatedGreedy will find an optimal endpoint, and by Corollary~\ref{cor:Once_optimal_always_optimal} output an optimal solution. 
    The algorithm makes a revolution in $\Oh{n^2 \log n}$ arithmetic operations by Corollary~\ref{cor:greedy_traverses_P_in_n3_log_n}, the total number of arithmetic operations is $\Oh{\opt n^5 \log n}$. 
\end{proof}

% --------------------------------------------------------------------
\section{Bit Complexity of \RepeatedGreedy}
\label{sec:bit_complexity}
% --------------------------------------------------------------------

In this section we show how to bound the bit complexity of the points calculated by the \Greedy algorithm, to then bound the overall bit complexity of the arithmetic operations of both \Greedy and \RepeatedGreedy.
We define the bit complexity as follows, following the definition of Gr{\"{o}}tschel, Lov{\'{a}}sz and Schrijver~\cite{bit_complexity_book}.

\newcommand{\com}[1]{\left\langle #1 \right\rangle}
\newcommand{\maxcom}[2]{\max \{ \com{#1}, \com{#2} \}}
\begin{definition}[Bit complexity]
\label{def:bit_complexity}
    Let $\com{v}$ denote the bit complexity of value $v$.
    \begin{itemize}
        \item If $v$ is an integer, then $\com{v} = 1 + \lceil \log_2 (|v| + 1) \rceil$, i.e. the number of bits used to encode $v$, including the sign bit.
        \item If $v$ is a fraction $\frac{v^n}{v^d}$, then $\com{v} = \com{v^n} + \com{v^d}$.
        \item If $v$ is a point $(v_x, v_y)$, then $\com{v} = \com{v_x} + \com{v_y}$.
    \end{itemize}
\end{definition}

Note that $\com{v}$ denotes the number of bits used to represent the value $v$ up to a constant factor, as it does not include bits used to separate, i.e., the numerator and denominator of a fraction.
The notation is stronger than $\mathcal{O}$-notation, as it does not allow to hide a possible constant, to ensure that the hidden constant in the $\mathcal{O}$-notation truly is a constant.
In the following lemmas, we will need the following rewriting properties of bit complexity on integers.

\begin{lemma}[Properties of bit complexity]
\label{lem:integer_bit_complexity_prop}
    Let $a$ and $b$ be integers. Then the following holds:
    \begin{itemize}
        \item $\com{a \cdot b} \le \com{a} + \com{b}$.
        \item $\com{a + b} \le 1 + \max \{ \com{a}, \com{b} \}$.
    \end{itemize}
\end{lemma}
\begin{proof}
    Follows from Definition~\ref{def:bit_complexity} and logarithm rules.
\end{proof}

In this section we shall assume, that the input polygon $P$ is encoded in $N$ bits as a list of points, i.e., $4n$ integers encoded in binary. The polygon may be encoded using fewer bits by a different clever encoding, however, as the main goal is to show that \cagp is in the complexity class \textsc{P}, this simple encoding will suffice.

To bound the complexity of the points produced by \Greedy, we need the notion of a \emph{scalar point}, which is defined as a point using two vertices of the input polygon and a fractional scalar, which denotes at what fractional point along the line defied by the two vertices the scalar point is located. Note that the scalar point does not need to be on the segment between the two points that defines it.
Crucially, the intersection between two lines defined from three input vertices and a scalar point is a scalar point, as shown in the following lemma, which also bounds the bit complexity of the new scalar point as a function of the old.

\begin{lemma}[Bit complexity of segment intersection]
\label{lem:segment_intersection_complexity}
    Let $N$ be the number of bits used to represent the input polygon.
    Let $V$, $W$, $B$, $C$, and $D$ be input vertices, with $V = \left( \frac{V_x^n}{V_x^d}, \frac{V_y^n}{V_y^d} \right)$, and similarly for the other input vertices, where all values $v_x^n,v_x^d,v_y^n,v_y^d$ of the points are integers.
    Let $s$ be a scalar $\frac{s^n}{s^d}$, with $s^n$ and $s^d$ being integers, where $\com{s}$ may be of arbitrary complexity. Let $A$ be the scalar point $s V + (1 - s) W$.
    Assume $\lineextension{A}{B}$ and $\lineextension{C}{D}$ are non-parallel lines and let the point $E$ be their intersection.
    Then there exists a scalar $t = \frac{t^n}{t^d}$, such that the intersection $E$ is a scalar point of the form $t C + (1 - t) D$, where $\maxcom{t^n}{t^d} \le \maxcom{s^n}{s^d} + \Oh{N}$.
\end{lemma}
\begin{proof}
    \begin{figure}
        \centering
        
        \tikzstyle{point} = [circle, fill=blue, inner sep=0pt, minimum size=1ex]
        \begin{tikzpicture}[scale=0.6]
            \node[point, label={above left:$V$}] (V) at (0, 0) {};
            \node[point, label={right:$W$}] (W) at (6, -1) {};
            \draw[dashed] (V) -- (W);
            
            \node[point, fill=red, label={3:$A = s V + (1 - s) W$}] (A) at (2, -0.333) {};
            \node[point, label={above right:$B = \left( \frac{B_x^n}{B_x^d}, \frac{B_y^n}{B_y^d} \right)$}] (B) at (4.5, 7) {};
            \node[point, label={above left:$C$}] (C) at (-1, 3.5) {};
            \node[point, label={above right:$D$}] (D) at (5.5, 4.5) {};
            \draw (A) -- (B);
            \draw (C) -- (D);

            \node[point, fill=red, label={280:$E = t C + (1 - t) D$}] (E) at (3.55, 4.2) {};
        \end{tikzpicture}
        
        \caption{%
        Setup of the intersection.
        Blue points denote input vertices and the red points denote scalar points.
        Note that it is not a requirement that the segments $\linesegment{A}{B}$ and $\linesegment{C}{D}$ intersect, only that their lines intersect.}
        \label{fig:line_segment_intersection_gadget}
    \end{figure}

    The setup is illustrated in Figure~\ref{fig:line_segment_intersection_gadget}.
    Using a formula for the intersection of the lines defined by the segments $\linesegment{A}{B}$ and $\linesegment{C}{D}$ as stated by Goldman~\cite{GOLDMAN1990}, the values of $t^n$ and $t^d$ can be expressed as follows:
    \begin{align*}
        t^n &=& s^n (&
            + B_x^d B_y^d C_x^d C_y^d D_x^d D_y^n V_x^n V_y^d W_x^d W_y^d
            + B_x^d B_y^d C_x^d C_y^d D_x^n D_y^d V_x^d V_y^d W_x^d W_y^n
        \\ & & &
            + B_x^d B_y^n C_x^d C_y^d D_x^d D_y^d V_x^d V_y^d W_x^n W_y^d
            + B_x^n B_y^d C_x^d C_y^d D_x^d D_y^d V_x^d V_y^n W_x^d W_y^d
        \\ & & &
            - B_x^d B_y^d C_x^d C_y^d D_x^d D_y^n V_x^d V_y^d W_x^n W_y^d
            - B_x^d B_y^d C_x^d C_y^d D_x^n D_y^d V_x^d V_y^n W_x^d W_y^d
        \\ & & &
            - B_x^d B_y^n C_x^d C_y^d D_x^d D_y^d V_x^n V_y^d W_x^d W_y^d
            - B_x^n B_y^d C_x^d C_y^d D_x^d D_y^d V_x^d V_y^d W_x^d W_y^n
        ) \\
        & & + s^d (&
            + B_x^d B_y^d C_x^d C_y^d D_x^d D_y^n V_x^d V_y^d W_x^n W_y^d
            + B_x^d B_y^n C_x^d C_y^d D_x^n D_y^d V_x^d V_y^d W_x^d W_y^d
        \\ & & &
            + B_x^n B_y^d C_x^d C_y^d D_x^d D_y^d V_x^d V_y^d W_x^d W_y^n
            - B_x^d B_y^d C_x^d C_y^d D_x^n D_y^d V_x^d V_y^d W_x^d W_y^n
        \\ & & &
            - B_x^d B_y^n C_x^d C_y^d D_x^d D_y^d V_x^d V_y^d W_x^n W_y^d
            - B_x^n B_y^d C_x^d C_y^d D_x^d D_y^n V_x^d V_y^d W_x^d W_y^d
        ) \\ \\
        t^d &=& s^n (&
            + B_x^d B_y^d C_x^d C_y^d D_x^d D_y^n V_x^n V_y^d W_x^d W_y^d
            + B_x^d B_y^d C_x^d C_y^d D_x^n D_y^d V_x^d V_y^d W_x^d W_y^n
        \\ & & &
            + B_x^d B_y^d C_x^d C_y^n D_x^d D_y^d V_x^d V_y^d W_x^n W_y^d
            + B_x^d B_y^d C_x^n C_y^d D_x^d D_y^d V_x^d V_y^n W_x^d W_y^d
        \\ & & &
            - B_x^d B_y^d C_x^d C_y^d D_x^d D_y^n V_x^d V_y^d W_x^n W_y^d
            - B_x^d B_y^d C_x^d C_y^d D_x^n D_y^d V_x^d V_y^n W_x^d W_y^d
        \\ & & &
            - B_x^d B_y^d C_x^d C_y^n D_x^d D_y^d V_x^n V_y^d W_x^d W_y^d
            - B_x^d B_y^d C_x^n C_y^d D_x^d D_y^d V_x^d V_y^d W_x^d W_y^n
        ) \\
        & & + s^d (&
            + B_x^d B_y^d C_x^d C_y^d D_x^d D_y^n V_x^d V_y^d W_x^n W_y^d
            + B_x^d B_y^d C_x^n C_y^d D_x^d D_y^d V_x^d V_y^d W_x^d W_y^n
        \\ & & &
            + B_x^d B_y^n C_x^d C_y^d D_x^n D_y^d V_x^d V_y^d W_x^d W_y^d
            + B_x^n B_y^d C_x^d C_y^n D_x^d D_y^d V_x^d V_y^d W_x^d W_y^d
        \\ & & &
            - B_x^d B_y^d C_x^d C_y^d D_x^n D_y^d V_x^d V_y^d W_x^d W_y^n
            - B_x^d B_y^d C_x^d C_y^n D_x^d D_y^d V_x^d V_y^d W_x^n W_y^d
        \\ & & &
            - B_x^d B_y^n C_x^n C_y^d D_x^d D_y^d V_x^d V_y^d W_x^d W_y^d
            - B_x^n B_y^d C_x^d C_y^d D_x^d D_y^n V_x^d V_y^d W_x^d W_y^d
        )
    \end{align*}
    Following Lemma~\ref{lem:integer_bit_complexity_prop} it holds that $\com{t^n}$ and $\com{t^d}$ both are bounded by $\maxcom{s^n}{s^d} + \Oh{N}$, as the point values of the input vertices are bounded trivially by $N$, and therefore $\maxcom{t^n}{t^d} \le \maxcom{s^n}{s^d} + \Oh{N}$.
\end{proof}

Further, the bit complexity of a scalar point, when computed into a point, is bounded.

\begin{lemma}[Bit complexity of scalar point]
\label{lem:scalar_point_complexity}
    Let $N$ be the number of bits used to represent the input polygon.
    Let $V$ and $W$ be input vertices, with $V = \left( \frac{V_x^n}{V_x^d}, \frac{V_y^n}{V_y^d} \right)$, and similarly for $W$, where each component of the points are integers.
    Let $s$ be a scalar $\frac{s^n}{s^d}$ of arbitrary complexity, with $s^n$ and $s^d$ being integers, and let $A$ be the scalar point $s V + (1 - s) W$.
    Then $\com{A} \le 4 \maxcom{s^n}{s^d} + \Oh{N}$.
\end{lemma}
\begin{proof}
    Let $A = (A_x, A_y)$. It then holds that
    \[ A_X = \frac{V_x^d W_x^n s^d - V_x^d W_x^n s^n + V_x^n W_x^d s^n}{V_x^d W_x^d s^d} \; , \]
    and similarly for $A_y$.
    Applying Definition~\ref{def:bit_complexity} and Lemma~\ref{lem:integer_bit_complexity_prop} immediately concludes the proof.
\end{proof}

Using the lemmas on scalar points and intersections, we can show that when the input point to \Greedy is a scalar point, then the output point is a scalar point, which scalar complexity is bounded by the scalar complexity of the input point.

\begin{lemma}[Bound output bit complexity of \Greedy]
\label{lem:bit_complexity_greedy}
    Let $P$ be a simple polygon encoded in $N$ bits and let $x$ be a scalar point on $\partial P$. Let $V$ and $W$ be vertices of $P$ and $s$ be a scalar $\frac{s^n}{s^d}$, s.t. $x = s V + (1 - s) W$.
    Then there exists vertices $V'$ and $W'$ of $P$ and a scalar $s' = \frac{s'^n}{s'^d}$, s.t. $\Gsc(x) = s' V' + (1 - s') W'$, and it holds that $\maxcom{s'^n}{s'^d} \le \maxcom{s^n}{s^d} + \Oh{N}$.
\end{lemma}
\begin{proof}
    \Greedy proceeds in two phases. First, the feasible region $F$ of $x$ and a maximal number of contiguous edges from $x$ is computed. Next, the guard seeing the furthest along the next edge of the boundary is then a corner of $F$, and $\Gsc(x)$ is found.
    We therefore need to bound the complexity of the corners of $F$ to then bound the final output point.

    The feasible region $F$ is found by the intersection of the visibility polygons from $x$ and some of the input vertices.
    It must therefore hold that each vertex of $F$ is the intersection of lines from the visibility polygons.
    Each line of a visibility polygon is either from a segment of $P$, which is therefore a line from a input vertex to a input vertex, a line from a reflex vertex to a input vertex, which is a input vertex to a input vertex line, or from the point $x$ to a reflex vertex.
    Two lines both containing $x$ must intersect in $x$, and such point is by definition a scalar point between input vertices.
    Any other intersection is then either between four input vertices, or three input vertices and the point $x$.
    By Lemma~\ref{lem:segment_intersection_complexity}, it holds that each corner of the feasibility region $F$ is a scalar point on the form $s'' V'' + (1 - s'') W''$, for some input vertices $V''$ and $W''$, and scalar $s'' = \frac{s''^n}{s''^d}$, and it holds that $\maxcom{s''^n}{s''^d} \le \maxcom{s^n}{s^d} + \Oh{N}$.

    Next, the point $\Gsc(x)$ is found as the intersection of the line segment of some edge of $P$, and the line between a guard placement and a reflex vertex of $P$. As the guard placement is a corner of $F$, then the intersection is between three input vertices and a scalar point. Let the guard point be $s'' V'' + (1 - s'') W''$.
    By Lemma~\ref{lem:segment_intersection_complexity}, then there are vertices $V'$ and $W'$, which is the endpoints of the edge $\Gsc(x)$ is on, and a scalar $s' = \frac{s'^n}{s'^d}$, s.t. $\Gsc(x) = s' V' + (1 - s') W'$, and it holds that $\maxcom{s'^n}{s'^d} \le \maxcom{s''^n}{s''^d} + \Oh{N} \le \maxcom{s^n}{s^d} + \Oh{N}$.
    This therefore concludes the lemma.
\end{proof}

Similarly, it can be shown that the internal arithmetic operations \Greedy performs are bounded by the complexity of the input point.

\begin{lemma}[Bound internal bit complexity of \Greedy]
\label{lem:bit_complexity_internal_greedy}
    Let $P$ be a simple polygon encoded in $N$ bits and let $x$ be a point on $\partial P$. Then the internal arithmetic operations of \Greedy with input $x$ is polynomially bounded in complexity by $N$ and $\com{x}$.
\end{lemma}
\begin{proof}
    Each internal arithmetic computation of \Greedy are computed by a constant size straight line program, which therefore is computeable in time linear in $N$ and $\com{x}$. By Theorem~\ref{thm:running_time_alg_greedy_segment} \Greedy performs polynomially many such computations in $N$, which concludes the lemma.
\end{proof}

We now have the building blocks necessary to show the main theorem.

\thmCAGIsTruelyInP*

\begin{proof}
    \RepeatedGreedy repeatably applies \Greedy to make revolutions around the polygon.
    By Lemma~\ref{lem:bit_complexity_internal_greedy} the internal arithmetic operations of \Greedy are bounded by the complexity of the input point and the polygon. It therefore suffices to argue for the bit complexity of each intermediate point on $\partial P$ where \Greedy is applied to.

    \RepeatedGreedy starts at some point on $\partial P$. Let this starting point be some vertex $V$ of $P$. Let $W$ be any other vertex of $P$. It then holds that the starting point is on the form $sV + (1-s)W$ for scalar $s=1/1$.
    Let $s_t = \frac{s_t^n}{s_t^d}$ be the scalar after $t$ applications of \Greedy. By induction and Lemma~\ref{lem:bit_complexity_greedy} it holds that $\maxcom{s_t^n}{s_t^d} \le \maxcom{s_0^n}{s_0^d} + t \cdot \Oh{N} = \Oh{t N}$.
    By Theorem~\ref{thm:CAG_is_in_P} the number of applications of \Greedy is $\Oh{\opt^2 n^3}$, which bounds $t$.
    As each scalar $s_t$ corresponds to a scalar point is between some input vertices of $P$, then by Lemma~\ref{lem:scalar_point_complexity} it therefore holds that the bit complexity of any point computed by \Greedy on the boundary $\partial P$ is bounded by $\Oh{\opt^2 n^3 \cdot N}$.
    This, by Lemma~\ref{lem:bit_complexity_internal_greedy}, concludes the proof.
\end{proof}

% --------------------------------------------------------------------
\section{Open problems}
\label{sec:open_problems}
% --------------------------------------------------------------------

We showed that the \cagp is solvable in $\Oh{\opt n^5 \log n}$ time by bounding the number of revolutions before \RepeatedGreedy finds an optimal solution to $\Oh{\opt n^3}$ and showing that the bit complexity is bounded.
We conjecture that $\Oh{1}$ revolutions are sufficient. 
To provide evidence for this we simulated more than 2.000.000 random art galleries using the provided \Cpp implementation, and in all instances, it found an optimal solution (a repetition even) within $4$ revolutions.
If this conjecture is true, the complexity of \RepeatedGreedy becomes $\Oh{n^2 \log n}$. 
Analyzing the behavior of \RepeatedGreedy on axis-aligned input polygons appears to be a good step towards proving this conjecture.

We leave it as an intriguing open problem to decide whether the \cagp with holes is polynomial-time solvable, in contrast to the other art gallery variants we considered (see Table~\ref{tab:table_of_related_work}).

If an unrestricted guard may cover $h$ intervals ($h > 1$), we believe the problem is \textsc{NP}-hard. 
This is the case for $h=n$ since this is exactly the edge-covering art gallery problem.

% --------------------------------------------------------------------
%\bibliographystyle{plainurl}
%\bibliography{refs}

% --------------------------------------------------------------------

% --------------------------------------------------------------------
\clearpage
\appendix
% --------------------------------------------------------------------

% --------------------------------------------------------------------
\section{Vertex restricted contiguous art gallery}
\label{apx:vertex_restricted_variants}
% --------------------------------------------------------------------

In the following, we show how to efficiently compute an optimal solution to the \cagp, when the problem is vertex restricted, as mentioned in Section~\ref{sec:limitations_of_existing_approaches} and Table~\ref{tab:table_of_related_work}. There are two variants of this restriction, one is when the \emph{intervals are restricted to vertices} and the other is when \emph{guards are restricted to vertices}, which we cover in Appendix~\ref{apx:vertex_restrict_intervals} and \ref{apx:vertex_restrict_guards}, respectively.
We let for the following the input be a simple polygon $P$ containing $n$ vertices.

% --------------------------------------------------------------------
\subsection{Intervals are restricted to vertices}
\label{apx:vertex_restrict_intervals}
% --------------------------------------------------------------------

If the contiguous interval seen by a guard is restricted to start and end at a vertex, the problem can be solved as follows.
In a polygon of $n$ vertices, there are $\OhTheta{n^2}$ contiguous intervals of the boundary, that is, from any vertex to any other. However, not all of these may be valid as there may not exist a guard that can see the whole interval.
When all valid intervals have been found, the algorithm of Lee and Lee~\cite{circle_cover_arc_minimization_1984_LEE1984109} can be used to find an optimal solution.
Note that if a valid interval is contained entirely in another valid interval, then it can be pruned by Lemma~\ref{lem:properties_of_Fc}.1, without removing optimality.
For the valid intervals generated, we report only the intervals starting at a vertex and being maximal clockwise.

This is exactly equal to computing \Greedy from a starting vertex and restricting the final point to the last vertex on the computed interval.
An iteration of \Greedy, when the interval contains $e$ edges, runs in $\Oh{e n \log n} = \Oh{n^2 \log n}$ time (Theorem~\ref{thm:running_time_alg_greedy_segment}), leading to total $\Oh{n^3 \log n}$ time used to compute the intervals. Computing the optimal solution of these intervals takes $\Oh{n}$ time, as the intervals are generated in sorted order, which leads to $\Oh{n^3 \log n}$ time in total.

Note that if the interval $[v_i, v_j]$ is computed from vertex $v_i$, then the interval starting at $v_{i + 1}$ must be able to reach at least $v_j$. Recomputing \Greedy for each vertex does not use this fact, and leads to the following optimization.
If the visibility polygon of a vertex can be efficiently removed from the feasible polygon \Greedy used to determine the guard location and therefore also to decide whether an interval is maximal, then recomputation time can be optimized.
The intersection algorithm is capable of intersecting both visibility polygons and feasible regions. Computing the feasible polygon of a set of visibility regions is, therefore, solvable by a simple divide-and-conquer algorithm.

The result of Overmars~\cite{Overmars_static_to_dynamic} show how to convert such static divide-and-conquer algorithms into a dynamic version, allowing for both insertion and deletion of visibility polygons from the set.
The intersection algorithm takes two objects and computes their intersection. Let $m$ be the total number of underlying visibility polygons on which the intersection is computed, with $m_1$ and $m_2$ visibility polygons contained in the two objects to intersect.
The intersection is computed in $\Oh{(n + h) \log n}$ time, with $h$ describing the number of intersections between the two objects, due to an algorithm by Martínez, Rueda, and Feito~\cite{martinez_2009_clipping_MARTINEZ20091177}, see Section~\ref{sec:intersections}. Note that intersections must be either some vertex of an input object intersecting the edge of the other input object, or a proper intersection between edges of the input objects.
Every such proper intersection must lead to a vertex in the output object. The complexity of the output object is $\Oh{n}$ (Lemma~\ref{lem:complexity_of_feasible_region_vertices}), and as there similarly is $\Oh{n}$ and $\Oh{n}$ vertices in the input objects, then the number of total intersections $h \le \Oh{n} + \Oh{n} + \Oh{n} = \Oh{n}$.
The running time of the intersection is therefore $\Oh{n \log n}$. This, by Overmars~\cite[Theorem 3.4]{Overmars_static_to_dynamic}, yields a data structure over $m$ visibility polygons, allowing for both insertions and deletions of visibility polygons in $\Oh{n \log n \log m}$ time, while maintaining the feasible polygon over all current visibility polygons.

The optimized algorithm is then as follows.
Start at any vertex $v_i$, and insert the visibility polygon of this vertex in the data structure that maintains the feasible region polygon. Then repeatably insert the visibility polygon of the next vertex until the feasible polygon becomes empty. Let this lastly inserted vertex be $v_{j+1}$. Then there must exist a guard that can see the interval $[v_i, v_j]$, which is outputted.
Note that on an insertion, the visibility polygon of the vertex must be computed in $\Oh{n}$ time, using the algorithm by Gindy and Avis~\cite{visibility_polygon_1981_ELGINDY1981186}, see Section~\ref{sec:vis_pol}, to then be inserted.
To compute the interval starting at $v_{i+1}$, then $v_i$ is deleted from the data structure, and until the feasible polygon again becomes empty, vertices from $v_{j+2}$ and onward are inserted. This procedure then generates the interval starting at $v_{i+1}$.
This process continues to compute the maximal interval starting at every vertex.
During execution of this procedure, the number of vertices in the data structure is at most $n$, and therefore $m \le n$. Every vertex is inserted and deleted $\Oh{1}$ times.
It takes $n \cdot \Oh{n} = \Oh{n^2}$ total time to compute the visibility polygons of every vertex.
The overall running time is therefore $\Oh{n^2} + n \cdot \Oh{1} \cdot \Oh{n \log n \log m} = \Oh{n^2 \log^2 n}$.
This outputs $n$ maximal intervals. Computing an optimal solution from the intervals takes $\Oh{n}$ time, as the intervals are generated in sorted order, leading to an overall $\Oh{n^2 \log^2 n}$ time to compute an optimal solution.

% --------------------------------------------------------------------
\subsection{Guards are restricted to vertices}
\label{apx:vertex_restrict_guards}
% --------------------------------------------------------------------

If the guard locations are restricted to the vertices, the problem can be solved as follows.
If the polygon $P$ has $n$ vertices, then there is only $n$ possible placements of a guard.
Note that the visibility polygon from a point $v$ yields exactly the area that can be seen from $v$. By intersecting the visibility polygon with the boundary $\partial P$, the contiguous intervals of $\partial P$ can be computed. The visibility polygon from a vertex $v_i$ can be computed in $\Oh{n}$ time by the algorithm of Gindy and Avis~\cite{visibility_polygon_1981_ELGINDY1981186}, see Section~\ref{sec:vis_pol}. The intersection with $\partial P$ is computeable in $\Oh{n}$ time.
Each guard may see multiple contiguous intervals of $\partial P$, however, there is at most $\Oh{n}$ intervals for each guard.
In total, all contiguous intervals visible to a guard located in any vertex of $P$ are computeable in $n \cdot \Oh{n} = \Oh{n^2}$ time yielding $\Oh{n^2}$ intervals.
The intervals are not computed in sorted order, so the Lee and Lee~\cite{circle_cover_arc_minimization_1984_LEE1984109} algorithm computes an optimal solution in $\Oh{n^2 \log n}$.

% --------------------------------------------------------------------
\section{Proof of supporting lemmas}
\label{apx:proof_of_supporting_lemmas}
% --------------------------------------------------------------------

% --------------------------------------------------------------------
\subsection{Feasible regions are connected}
\label{apx:feasible_region_connected}
% --------------------------------------------------------------------

An important property of the feasible region $F$ is that it is connected. To prove this, we show a generalization of Avis and Toussain~\cite[Lemma~1]{visibility_of_polygon_from_an_edge_convex_viewing_lemma_1981_1675729}:

\begin{lemma}[Convex viewing lemma]
\label{lem:convex_viewing_lemma}
    Let $a, b$ and $c$ be points in the interior of $P$ (or on $\partial P)$ such that $a$ sees both $b$ and $c$ and that there is some point $d\in P$ on $\linesegment{b}{c}$, which is not visible from $a$. Then $\partial P$ must intersect $\linesegment{b}{c}$.
\end{lemma}

\begin{proof}
    \begin{figure}[ht]
        \centering
        
        \begin{tikzpicture}
            \coordinate (A) at (0, 4);
            \coordinate (B) at (0, 1);
            \coordinate (C) at (2, 3);
            \coordinate (D) at (1, 2);
            \draw[color=black, dashed] (A) -- (B);
            \draw[color=black, dashed] (A) -- (C);
            \draw[color=black, dotted] (B) -- (C);
            \draw[color=notvisible, dashed] (A) -- (D);
    
            \draw[color=black] (2.1,2.6) -- (0.2, 3) -- (1.6, 2.1);
    
            \filldraw[color = red] (D) circle (1.5pt);
            \draw[color = red, right] (D) node {$d$};
            \draw[color = black, left] (A) node {$a$};
            \draw[color = black, left] (B) node {$b$};
            \draw[color = black, right] (C) node {$c$};
            \draw[color = black, below] (2,2.6) node {$\partial P$};
        \end{tikzpicture}
        
        \caption{When $a$ can see $b$ and $c$, but not $d$, some part of $\partial P$ intersects $\linesegment{b}{c}$.}
        \label{fig:dP_intersecting_BC}
    \end{figure}
    
    Since $a$ can see both $b$ and $c$, $\linesegment{a}{c}$ and $\linesegment{b}{c}$ are unobstructed.
    However, since $a$ cannot see $d$, $\partial P$ must intersect $\linesegment{a}{d}$. 
    Since $a, b$ and $c$ are all inside $P$ then $\partial P$ must intersect triangle $abc$ to enter, so it can block $d$ from $a$, however, this can only happen by intersecting $\linesegment{b}{c}$ (see Figure~\ref{fig:dP_intersecting_BC}). 
\end{proof}

The feasible regions are constructed as the intersection between \emph{visibility polygons} for single points on $\partial P$. 
We cover in detail how the \Greedy algorithm computes these in Section~\ref{sec:alg_greedy_interval}.

A visibility polygon from a point $a \in \partial P$ can be decomposed into \emph{visible triangles} by drawing lines from $a$ to all the vertices of $\partial P$ visible from $a$ (see Figure~\ref{fig:visible_triangles_from_A}).

\begin{figure}[ht]
    \centering
    
    \begin{tikzpicture}[scale = 0.8]
        \coordinate (A) at (0.85, -1.42);
        \coordinate (B) at (1.78, -0.52);
        \coordinate (C) at (3.43, -0.74);
        \coordinate (D) at (1.84, 0.86);
        \coordinate (E) at (3.09, 1.66);
        \coordinate (F) at (1.88, 3.79);
        \coordinate (G) at (0.91, 2.72);
        \coordinate (H) at (1.1, 3.57);
        \coordinate (I) at (-1.15, 3.48);
        \coordinate (J) at (-1.42, 2.7);
        \coordinate (K) at (-0.05, 2.57);
        \coordinate (L) at (-2, 1.02);
        \coordinate (M) at (-2.34, 3.62);
        \coordinate (N) at (-4, 2);
        \coordinate (O) at (-2.47, 1.38);
        \coordinate (P) at (-3.16, -2.65);
        \coordinate (Q) at (-2.01, 0.13);
        \coordinate (R) at (-1.15, -0.09);
        \coordinate (S) at (-1.07, -0.81);
        \coordinate (T) at (-2.57, 0.86);
        \coordinate (U) at (-0.27, 3.51);
        \coordinate (V) at (2.58, 2.55);
        \coordinate (W) at (2.5, 0.18);
        \coordinate (Z) at (0.92, 3.56);
        \coordinate (A1) at (-3.62,2.37);

        \filldraw[visible triangle] (A) -- (W) -- (D) -- (V) -- (F) -- (G) -- (Z) -- (U) -- (K) -- (L) -- (A1) -- (O) -- (T) -- (R) -- (S) -- cycle;
        \filldraw[nv] (B) -- (C) -- (W) -- cycle;
        \filldraw[nv] (D) -- (E) -- (V) -- cycle;
        \filldraw[nv] (G) -- (H) -- (Z) -- cycle;
        \filldraw[nv] (U) -- (I) -- (J) -- (K) -- cycle;
        \filldraw[nv] (T) -- (P) -- (Q) -- (R) -- cycle;
        \filldraw[nv] (O) -- (N) -- (M) -- (L) -- cycle;

        \draw[thick] (A) -- (B) -- (C) -- (1.855, 0.86) -- (E) -- (F) -- (0.925, 2.72) -- (H) -- (I) -- (J) -- (-0.13, 2.57) -- (L) -- (M) -- (N) -- (O) -- (P) -- (Q) -- (R) -- (S) -- cycle; 
        \draw[gray, thick] (A) -- (W);
        \draw[gray, thick] (A) -- (V);
        \draw[gray, thick] (A) -- (F);
        \draw[gray, thick] (A) -- (Z);
        \draw[gray, thick] (A) -- (U);
        \draw[gray, thick] (A) -- (A1);
        \draw[gray, thick] (A) -- (T);
        \draw[gray, thick] (A) -- (R);
        \draw[gray, thick] (A) -- (S);

        \draw[black, below] (A) node {$a$};
    \end{tikzpicture}
    
    \caption{The grey region is a visibility polygon $\visPol{P}{a}$ from point $a$ in polygon $P$. It is further subdivided into visible triangles and the red non-visible areas are pockets that connect to $\visPol{P}{a}$ at windows.}
    \label{fig:visible_triangles_from_A}
\end{figure}
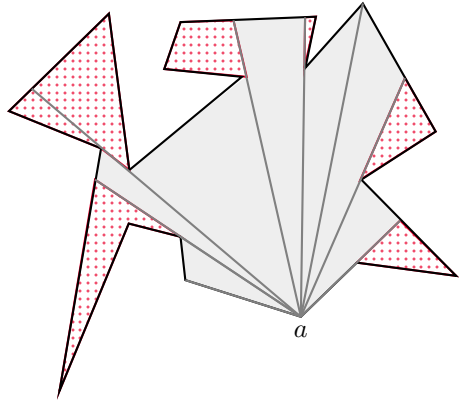

Note that these visible triangles may be degenerate (i.e. just a line). 
The areas marked in red, which are not visible to $a$, we denote \emph{pockets}. 
Each pocket is connected to the visibility polygon through line segments shared with one visible triangle. 
We denote these shared line segments as \emph{windows} (similar to Gindy and Avis~\cite{visibility_polygon_1981_ELGINDY1981186}). 
In Figure~\ref{fig:visible_triangles_from_A} each pocket has exactly one window. 
The following lemma shows that this is always true:

\begin{lemma}[Pockets have exactly one window]
\label{lem:pockets_are_connected_to_exactly_on_visible_triangle}
    Let $a \in \partial P$ and let $R$ be a pocket of $\visPol{P}{a}$. Then there is exactly one window for $R$.
\end{lemma}

\begin{proof}
    We show the lemma by contradiction. It is clear, that $R$ must be connected to at least one visible triangle by an edge, as $R$ otherwise would be disconnected from $P$ and hence not a part of $P$.

    Thus we assume there are at least 2 windows for $R$.
    Consider the connected components of $\partial P \cap \partial R$, i.e. the boundary of $R$ without the windows.
    As parts of $\partial P$ they are connected in a cycle with other parts of $\partial P$ and as parts of $\partial R$ they are connected in a cycle by windows.
    We consider a point $b$, which is a vertex of $R$, where one edge connected to $b$ is in a component of $\partial P \cap \partial R$, which we denote $f$ and the other is a window $w_b$.
    Furthermore, we assume $B$ is the closer point to $a$, of the endpoints of $w_b$. Let $c$ be the other endpoint $f$.
    See Figure~\ref{fig:obstructed_area_connected_to_multiple_triangles}.

    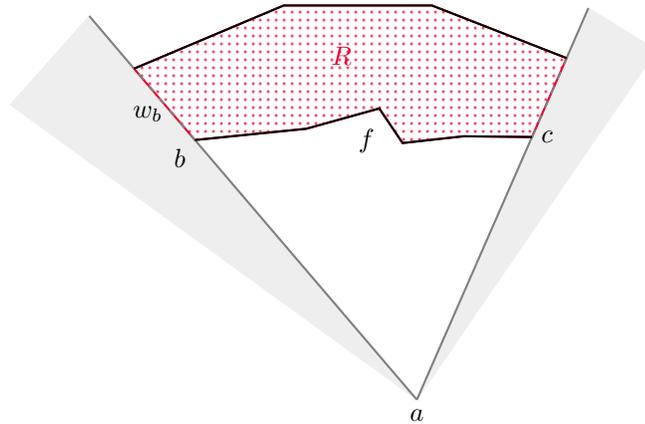
\begin{figure}[ht]
        \centering
        
        \begin{tikzpicture}[scale = 0.5]
            \coordinate (T) at (-11.59, 6.9);
            \coordinate (A) at (-10.4, 5.5);
            \coordinate (B) at (-6.42, 7.18);
            \coordinate (C) at (-2.48, 7.18);
            \coordinate (D) at (1.1, 5.78);
            \coordinate (E) at (0.18, 3.68);
            \coordinate (F) at (-1.64, 3.7);
            \coordinate (G) at (-3.26, 3.52);
            \coordinate (H) at (-3.88, 4.44);
            \coordinate (I) at (-5.82, 3.9);
            \coordinate (J) at (-8.78, 3.6);
            \coordinate (K) at (-2.88, -3.31);
            
            \filldraw[visible triangle] (K) -- (T) -- (-13.69, 4.54) -- cycle;
            \filldraw[visible triangle] (K) -- (1.68, 7.11) -- (3.5, 6) -- cycle;
            
            \filldraw[nv] (A) -- (B) -- (C) -- (D) -- (E) -- (F) -- (G) -- (H) -- (I) -- (J) -- cycle;
            \draw[below] (K) node {$a$};
            \draw[right] (E) node {$c$};
            \draw[below left, yshift=2pt] (J) node {$b$};
            \draw[notvisible] (-4.88, 5.84) node {$R$};
            \draw[left] (-3.76, 3.57) node {$f$};
            
            \draw[thick, gray] (K) -- (J);
            \draw[thick, gray,dashed] (J) -- node[midway, below left, yshift = 3pt, xshift = 3pt, black] {$w_b$} (A);
            \draw[thick, gray] (A) -- (T);
            \draw[thick, gray] (K) -- (E);
            \draw[thick, gray, dashed] (E) -- (D);
            \draw[thick, gray] (D) -- (1.68, 7.11);
            \draw[thick] (A) -- (B) -- (C) -- (D);
            \draw[thick] (E) -- (F) -- (G) -- (H) -- (I) -- (J);
        \end{tikzpicture}
        
        \caption{When a pocket is connected by multiple windows, the polygon cannot be simple.}
        \label{fig:obstructed_area_connected_to_multiple_triangles}
    \end{figure}
    
    We now consider where $\partial P$ can continue from $b$. We consider four cases:
    
    \begin{case}
        $\partial P$ can continue along $w_b$. However, this will contradict the assumption that $b$ is connected to the window.
    \end{case}
    
    \begin{case}
        $\partial P$ can continue into the visible triangle, but then the triangle will no longer be completely visible, a contradiction.
    \end{case}
    
    \begin{case}
        $\partial P$ can continue into $R$, however $R$ is completely contained in $P$, and if an edge of $\partial P$ lies in $R$, there will be some area on one side of the edge, which is not contained in $P$. Thus, we have a contradiction. \end{case}
    
    \begin{case}
        $\partial P$ can continue along $\linesegment{a}{b}$ or the white region below $R$ (see Figure~\ref{fig:obstructed_area_connected_to_multiple_triangles}). Once it enters here, it will become stuck inside the area bounded by the two visible triangles and $f$. It can exit through $C$, but this will create a loop in $\partial P$ not including the top of $R$, which would introduce a hole in $P$.
        
        It could also enter $a$. However, we could do the same case analysis for $C$, and it also need to enter into $a$, thus we again have a loop, and we are done.
        \popQED
    \end{case}
\end{proof}

With this, we are now ready to prove the main lemma of this subsection:

\lemFeasibleRegionIsConnected*

\begin{proof}
    The proof of connectivity follows by induction in the number of visible polygons we intersect in the \Greedy algorithm (see Section~\ref{sec:alg_greedy_interval}). 

    For the base case, i.e. $\FeasibleRegion{[a,a]}$ where $a\in \partial P$, the feasible region is the visibility polygon $\visPol{P}{a}$.
    Visibility polygons are star-shaped, hence connected.
    
    \begin{figure}[ht]
        \centering
        
        \begin{tikzpicture}[scale = 1.5]
            \coordinate (A) at (1.22, 1.11);
            \coordinate (B) at (1.68, 3.08);
            \coordinate (C) at (2.04, 4.63);
            \coordinate (D) at (0.82, 4.67);
            \coordinate (E) at (3.39, 3.23);
            \coordinate (F) at (3.37, 4.58);
            \coordinate (G) at (0.92, 3.81);
            \coordinate (H) at (-0.02, 3.3);
            \coordinate (I) at (1.35, 4.46);
            \coordinate (J) at (1.27, 4.24);
            \coordinate (K) at (2.52, 4.41);
            \coordinate (L) at (2.31, 3.66);
            \coordinate (M) at (2.21, 4.24);
            \coordinate (N) at (2.15, 3.97);
            \coordinate (O) at (1.45, 3.25);
            \coordinate (P) at (0.72, 3.55);
            \coordinate (Q) at (0.91, 2.56);
            \coordinate (R) at (0.1, 4.29);
            \coordinate (S) at (1.56, 4.36);
            \coordinate (T) at (1.33, 3.53);
            
            \filldraw[nv] (D) -- (R) -- (G) -- cycle;
            \filldraw[visible triangle] (A) -- (C) -- (D) -- (G) -- (H) -- cycle;
            \filldraw[nv] (B) -- (E) -- (F) -- (C) -- cycle;
    
            \draw[thick] (A) -- (B) -- (E) -- (F) -- (D) -- (R) -- (0.91, 3.81) -- (H) -- cycle;
            \draw[gray] (A) -- (C);
            \draw[gray] (A) -- (D);
            \draw[gray] (A) -- (H);
            
            \filldraw[F] (I) -- (J) -- (M) -- (N) -- (P) -- (Q) -- (O) -- (L) -- (K) -- cycle;
            
            \draw[notvisible, dashed] (S) -- (T);
            \filldraw[green] (S) circle (0.7pt);
            \filldraw[green] (T) circle (0.7pt);
    
            \draw[dashed, gray, nearly opaque] (A) -- (C);
            \draw[dashed, gray, nearly opaque] (A) -- (D);
    
            \draw[black, right] (A) node {\small{$v_{i+1}$}};
            \draw[black, right] (S) node {\small{$s$}};
            \draw[black, right] (T) node {\small{$t$}};
        \end{tikzpicture}
        
        \caption{In the induction step, we intersect $\visPol{P}{v_{i+1}}$ (in gray) with the feasible region $\FeasibleRegion{[a, v_i]}$ (in green) and assume for contradiction that this disconnects the feasible region. Taking points $s$ and $t$ in different components, we now use convex viewing lemma (Lemma~\ref{lem:convex_viewing_lemma}) to get a contradiction}
        \label{fig:when_feasible_region_gets_split_by_new_intersection}
    \end{figure}
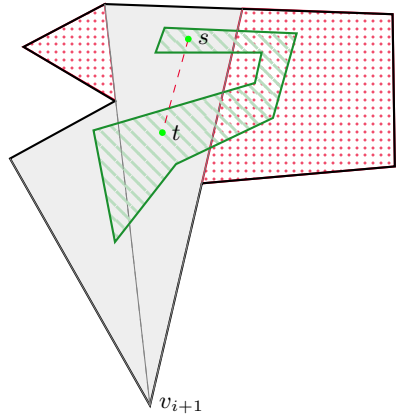

    For the induction step, we assume the feasible region seeing $[a, v_i]$ is connected and let $v_{i+1}$ be the next vertex along $\partial P$ or the point $b$ if no such vertex exits. 
    We assume for contradiction $\FeasibleRegion{[a,v_{i+1}]}$ is not connected. For this to happen, the feasible region $\FeasibleRegion{[a,v_i]}$ has to enter a pocket of the visibility polygon somewhere and exit elsewhere, so that the visibility polygon contains two disconnected components of $\FeasibleRegion{[a, v_i]}$ (see Figure~\ref{fig:when_feasible_region_gets_split_by_new_intersection}). Since each pocket only has one window by Lemma~\ref{lem:pockets_are_connected_to_exactly_on_visible_triangle}, the two components will intersect the same visible triangle. Let points $s$ and $t$ from different components of $\FeasibleRegion{[a,v_i]}$ in the same visible triangle.

    Since the previous feasible region is the intersection of visibility polygons for points $[a, v_i]$, there must be some point for which $s$ and $t$ are visible while some point on $\linesegment{s}{t}$ is not. 
    By Lemma~\ref{lem:convex_viewing_lemma}, there must be some part of $\partial P$ that intersects $\linesegment{s}{t}$, however, $\linesegment{s}{t}$ is contained in a visible triangle, where no part of $\partial P$ can lie. 
    Hence, the new intersected feasible region will continue to be connected, finishing the induction.

    The same considerations about $s$ and $t$ show the convexity claim.
\end{proof}

% --------------------------------------------------------------------
\subsection{Visible intervals cannot block their endpoints}
\label{apx:yz_cannot_block_z}
% --------------------------------------------------------------------

Since we only look at blockage (Remark~\ref{fig:horizon_vs_blockage}), it will be important to look at the point that blocks parts of edges from guards:

\begin{lemma}[What $g$ sees around blocking vertices]
\label{lem:what_G_sees_around_C}
    Let $g$ be a guard seeing $[y,z]$ and assume that $c$ is a vertex of $P$ that blocks $g$ from seeing more than $z$. Then either $g$ cannot see both edges connected to $c$, or one of these edges $\ell$ is a segment of $\lineextension{z}{g}$ and $g$ cannot see both of the edges connected to $\ell$.
\end{lemma}

\begin{proof}
    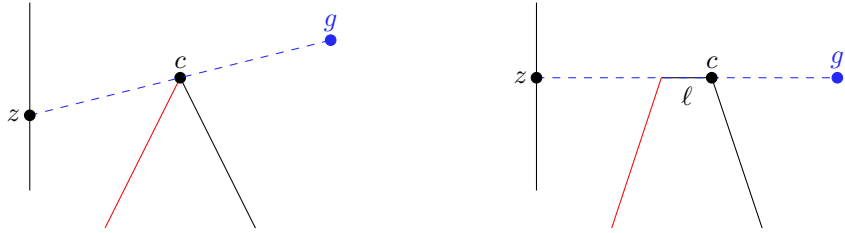
\begin{figure}[ht]
        \centering
        
        \begin{tikzpicture}
            \draw[color=black] (1,2) -- (2,0);
            \draw[color=red] (1,2) -- (0,0);
            \draw[color=black] (-1,3) -- (-1,0.5);
            \draw[color=guard, above] (3,2.5) node  {$g$}; 
            \filldraw[guard] (3,2.5) circle (2pt);
            \draw[color=black, above] (1,2) node {$c$};
            \filldraw[black] (1,2) circle (2pt);
            
            \draw[color=guard, dashed] (-1,1.5) --  (3,2.5);
            \draw[color = black, left] (-1,1.5) node {$z$};
            \filldraw[black] (-1,1.5) circle (2pt);
        \end{tikzpicture}%
        \hspace{6em}%
        \begin{tikzpicture}
            \draw[color=black] (0.66,2) --    (1.33,2) -- (2,0);
            \draw[color=red] (0,0) -- (0.66, 2);
            \draw[color=black] (-1,3) -- (-1,0.5);
            \draw[color=black, below] (1,2) node{$\ell$};
        
            \draw[color=guard, dashed] (-1,2) --(3,2);
    
            \draw[color=guard, above] (3,2) node {$g$}; 
            \filldraw[guard] (3,2) circle (2pt);
            \draw[color=black, above] (1.33,2) node    {$c$};
            \filldraw[black] (1.33,2) circle (2pt);
            \draw[color = black, left](-1,2) node {$z$};
            \filldraw[black] (-1,2) circle (2pt);
        \end{tikzpicture}
        
        \caption{Left, $c$ blocks $g$ and $g$ sees only one edge connected to $c$. Right, $g$ sees both edges, but one is a segment of $\lineextension{g}{c}$ and the edge after is not visible.}
        \label{fig:types_of_blocking}
    \end{figure}
    
    Assume $g$ can see both edges connected to $c$. Since $c$ blocks $g$ from seeing any more than $z$, we know $z$, $g$ and $c$ are collinear. Let $a$ and $b$ be points on the edge after, respectively before $c$, which are visible from $c$. We now show, that either $a$ and $b$ lie on different sides of $\lineextension{g}{c}$, or one of the points lies on $\lineextension{g}{c}$.
    
    So assume for contradiction $a$ and $b$ lie on the same side of $\lineextension{g}{c}$. Assume w.l.o.g. $\angle cga \leq \angle cgb$.

    Thus the ray $\ray{g}{a}$ intersects $\linesegment{b}{c}$, at some point which we denote $d$. Let $m$ be the midpoint of $\linesegment{a}{c}$ and $q$ the midpoint of $\linesegment{c}{d}$. We split into two cases: Either $a$ lies on $\linesegment{g}{d}$ or $d$ lies on $\linesegment{g}{a}$ (see Figure~\ref{fig:A_and_B_on_same_side_of_GC}).
    
    \begin{figure}[ht]
        \centering
        
        \begin{tikzpicture}[scale = 2]
            \draw[color=black] (0.6,0.9) -- (1.33,2) -- (1.6,1.2);
            \draw[notvisible, dashed] (1.5,1.5) -- (0.857,1.2857);
            \draw[color=guard, dashed] (1.5,1.5) -- (3,2);
            
            \draw[color=guard, dashed] (1,2) -- (3,2);
    
            \draw[color=guard, above] (3,2) node {$g$}; 
            \filldraw[guard] (3,2) circle (1pt);
            \draw[color=black, above] (1.33,2) node {$c$};
            \filldraw[black] (1.33,2) circle (1pt);
            \draw[color=black, above] (1.5,1.5) node {$a$};
            \filldraw[black] (1.5,1.5) circle (1pt);
            \draw[color=black, above] (0.66,1) node {$b$};
            \filldraw[black] (0.66,1) circle (1pt);
            \draw[color=black, above] (0.857,1.2857) node {$d$};
            \filldraw[black] (0.857,1.2857) circle (1pt);
            \draw[color=black, above] (1.0952, 1.6428) node {$q$};
            \filldraw[black] (1.0952, 1.6428) circle (1pt);
        \end{tikzpicture}%
        \hspace{5em}%
        \begin{tikzpicture}[scale = 2]
            \draw[color=black] (0.8,2.8) -- (1.33,2) -- (1.52,3.1);
            \draw[color= guard, dashed] (3,2) -- (1.33,2);
            \draw[color=notvisible, dashed] (1.4, 2.4) -- (1,2.5);
            \draw[color=guard, dashed] (3,2) -- (1.4, 2.4);
            \draw[color=guard, above] (3,2) node {$g$}; 
            \filldraw[guard] (3,2) circle (1pt);
            \draw[color=black, below] (1.33,2) node {$c$};
            \filldraw[black] (1.33,2) circle (1pt);
            \draw[color=black, above] (1,2.5) node {$a$};
            \filldraw[black] (1,2.5) circle (1pt);
            \draw[color=black, right] (1.5, 3) node {$b$};
            \filldraw[black] (1.5, 3) circle (1pt);
            \draw[color=black, right] (1.4, 2.5) node {$d$};
            \filldraw[black] (1.4, 2.4) circle (1pt);
            \draw[color=black, left] (1.166, 2.25) node {$m$};
            \filldraw[black] (1.166, 2.25) circle (1pt);
        \end{tikzpicture}
        
        \caption{$a$ and $b$ on the same side of $\lineextension{g}{c}$.}
        \label{fig:A_and_B_on_same_side_of_GC}
    \end{figure}
    
    \begin{case}
        If $a$ lies on $\linesegment{g}{d}$ then $g$ cannot see $q$ (Figure~\ref{fig:A_and_B_on_same_side_of_GC} (left)).
    \end{case}
    
    \begin{case}
        If $d$ lies on $\linesegment{g}{a}$ then $g$ cannot see $m$ (Figure~\ref{fig:A_and_B_on_same_side_of_GC} (right)).
    \end{case}
    
    Since we assumed that $g$ sees $a$, $b$, and $c$, $g$ must also be able to see $\linesegment{a}{c}$ and $\linesegment{b}{c}$, which is impossible as either $q$ or $m$ cannot be seen, thus $a$ and $b$ must either lie on different sides of $\lineextension{g}{c}$ or one must lie on $\lineextension{g}{c}$ (we cannot have both lie on $\lineextension{g}{c}$ since $c$ is a vertex).
    
    Assuming $a$ and $b$ are on either side of $\lineextension{g}{c}$, now the interval $[a,b]$ will block the view from $g$ to $z$ completely, thus we can discard this case. We assume w.l.o.g. $a$ lies on $\lineextension{g}{c}$. Let the endpoint of the edge containing $a$ different from $c$ be $j$ and let $m$ be the other edge connected to $j$. If $g$ can see no more of $m$ than $j$, we are done, so assume $g$ can see some point $u$ on $m$. $u$ cannot lie on $\lineextension{g}{c} = \lineextension{g}{j}$ since $j$ is a vertex. If $u$ and $b$ lie on the same side of $\lineextension{g}{c}$ we have the same problem as for $a$ and $b$ above (Figure~\ref{fig:something_breaks_if_G_sees_E} left). 
    If $k$ and $b$ lie on different sides of $\lineextension{g}{c}$ then $[b, u]$ will block the view from $g$ to $z$ (Figure~\ref{fig:something_breaks_if_G_sees_E}).
\end{proof}

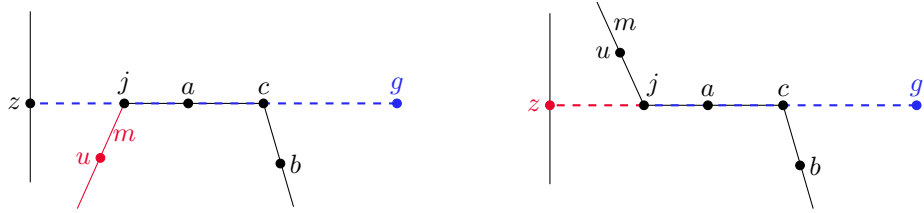
\begin{figure}[ht]
    \centering
    
    \begin{tikzpicture}[scale = 2.5]
        \draw (4, 2.25) -- (4, 3.16);
        \draw (4, 2.67) node[anchor=east] {$z$};
        \draw[guard] (5.95, 2.67) node[anchor=south] {$g$};
        \draw (5.4, 2.12) -- (5.24, 2.67) -- (4.5, 2.67);
        \draw[notvisible] (4.5, 2.67) -- (4.25, 2.11);
        \filldraw[color=guard] (5.95, 2.67) circle (0.66pt);
        \draw[guard, dashed, thick] (5.95, 2.67) -- (4, 2.67);
        \filldraw[color=black] (4, 2.67) circle (0.66pt);
        \filldraw (4.84, 2.67) circle (0.66pt);
        \draw (4.84, 2.67) node[anchor=south] {$a$};
        \filldraw (5.24, 2.67) circle (0.66pt);
        \draw (5.24, 2.67) node[anchor=south] {$c$};
        \filldraw (5.33, 2.35) circle (0.66pt);
        \draw (5.33, 2.35) node[anchor=west] {$b$};
        \filldraw (4.5, 2.67) circle (0.66pt);
        \draw (4.5, 2.67) node[anchor=south] {$j$};
        \filldraw[notvisible] (4.373, 2.38) circle (0.66pt);
        \draw[notvisible] (4.373, 2.38) node[anchor=east] {$u$};
        \draw[notvisible] (4.5, 2.5) node {$m$};
    \end{tikzpicture}%
    \hspace{4em}%
    \begin{tikzpicture}[scale = 2.5]
        \draw (4, 2.25) -- (4, 3.16);
        \filldraw[color=notvisible] (4, 2.67) circle (0.66pt);
        \draw[notvisible] (4, 2.67) node[anchor=east] {$z$};
        \draw[guard] (5.95, 2.67) node[anchor=south] {$g$};
        \draw (5.4, 2.12) -- (5.24, 2.67) -- (4.5, 2.67);
        \draw (4.5, 2.67) -- (4.25, 3.22);
        \filldraw[color=guard] (5.95, 2.67) circle (0.66pt);
        \draw[guard, dashed, thick] (5.95, 2.67) -- (4.5, 2.67);
        \draw[notvisible, dashed, thick] (4.5, 2.67) -- (4, 2.67);
        \filldraw (4.84, 2.67) circle (0.66pt);
        \draw (4.84, 2.67) node[anchor=south] {$a$};
        \filldraw (5.24, 2.67) circle (0.66pt);
        \draw (5.24, 2.67) node[anchor=south] {$c$};
        \filldraw (5.33, 2.35) circle (0.66pt);
        \draw (5.33, 2.35) node[anchor=west] {$b$};
        \filldraw (4.5, 2.67) circle (0.66pt);
        \draw (4.55, 2.67) node[anchor=south] {$j$};
        \filldraw (4.373, 2.95) circle (0.66pt);
        \draw (4.373, 2.95) node[anchor=east] {$u$};
        \draw (4.4, 3.1) node {$m$};
    \end{tikzpicture}
    
    \caption{Left, $u$ lies on same side as $b$ and is not visible from $g$. Right, $u$ lies on different side, but now $z$ is not visible from $g$.}
    \label{fig:something_breaks_if_G_sees_E}
\end{figure}

This can now be used to show where the blocking points are located in the figure:

\lemBlockingCornerNotInView*

\begin{proof}
    Assume for contradiction that $C \in [y, z]$. Since $c$ is a vertex of $P$, $c$ is neither $y$ nor $z$, thus $y$ and $z$ must lie on different sides of $c$ along $[y,z]$. From Lemma~\ref{lem:what_G_sees_around_C} we have two cases for what $g$ can see close to $c$ along $\partial P$.

    \begin{case}
        If $g$ cannot see both edges connected to $c$, either $[y,c]$ or $[c,z]$ is not visible to $g$ contradicting the fact that $g$ sees $[y,z]$.
    \end{case}
    
    \begin{case}
        If $g$ is able to see both edges connected to $c$, then we know one is contained in $\lineextension{g}{z}$ and the next edge is not visible from $g$, hence $z$ must lie on the edge connected to $c$, which is contained in $\lineextension{g}{z}$, for $[c,z]$ to be visible. However now $z$ can be moved until the endpoint of the edge containing $z$, hence $z$ is a vertex, which is contradicting the assumption (see Figure~\ref{fig:G_sees_more_than_just_z}).
        \popQED
    \end{case}
\end{proof}

\begin{figure}[ht]
    \centering
    \begin{tikzpicture}[scale = 2.5]
        \draw (4.27, 1.53) -- (4,2) -- (3,2);
        \draw[guard, dashed, thick] (4.75, 2) -- (3.5, 2);
        \draw[notvisible, dashed, thick] (3, 2) -- (3.5, 2);
        \filldraw[color=guard] (4.75, 2) circle (0.66pt);
        \draw[guard] (4.75, 2) node[anchor=west] {$g$};
        \filldraw (4, 2) circle (0.66pt);
        \draw (4, 2) node[anchor=south] {$c$};
        \filldraw (3.5, 2) circle (0.66pt);
        \draw (3.5, 2) node[anchor=south] {$z$};
    \end{tikzpicture}
    
    \caption{If $c$ is to act as blockage between $g$ and $z$, then $g$ should see no more than $z$. However, when $g, c$ and $z$ are collinear $g$ can see the entire edge containing $z$.}
    \label{fig:G_sees_more_than_just_z}
\end{figure}
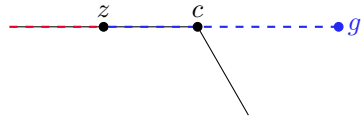

% --------------------------------------------------------------------
\section{Insightful examples}
\label{apx:counter_examples}
% --------------------------------------------------------------------

% --------------------------------------------------------------------
\subsection{Example of two combinatorially indistinguishable functions}
\label{apx:two_greedy_sequences}
% --------------------------------------------------------------------

\begin{example}
\label{exa:two_greedy_sequences}
    An examples of functions, which combinatorially look identically, but behave very different is given below:
    
    For this, we represent points on the circle by numbers in $[0,1)$.
    We then consider the two functions $\Gsc_1$ and $\Gsc_2$.
    The first (see Figure~\ref{fig:two_greedy_functions} (left)) is defined as $\Gsc_1(x) = \{x + \frac1k - \frac\varepsilon k\}$ where $\varepsilon$ is some fixed small number and $\{a\}$ is the decimal part of $a$ (i.e. $a = \lfloor a \rfloor + \{ a \}$). This will require $k + 1$ guards, but $\varepsilon^{-1}$ steps of \Greedy will be required to satisfy one of the optimality conditions (in this case we satisfy Corollary~\ref{cor:escaping_an_interval_implies_optimal}).
    
    Secondly, we define $\Gsc_2(x) = \{x + \frac1k - \alpha\frac{\{kx\}}{k}\}$, where $\alpha$ is some number in $(0,1)$ (see Figure~\ref{fig:two_greedy_functions} (right)). It has an optimal solution at $y_i = i/k$ with $k$ intervals. However starting anywhere other than in an optimal solution and running \RepeatedGreedy will never yield an optimal solution, since $\frac{\{xk\}}k$ is the distance $x$ needs to move backwards to hit the optimal solution, however we only move $\alpha$ times that distance each step.
    
    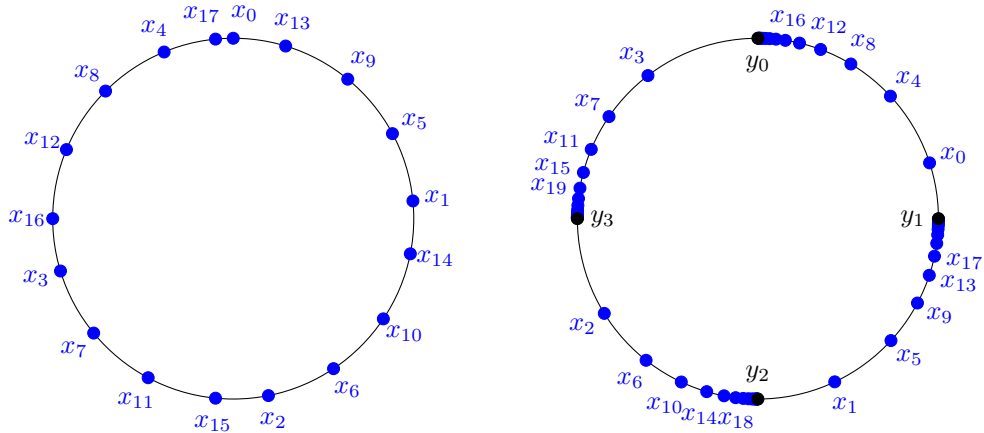
\begin{figure}[ht]
        \centering
    
        \begin{tikzpicture}[scale=0.75]
            \begin{scope}
                \node at (-4.5, 0) {};
    
                \draw (0, 0) circle (3.2);
                \filldraw[color=blue] (0.000, 3.200) circle (3pt); % angle = 90.000
                \draw[blue] (0.200, 3.650) node {$x_{0}$};
                \filldraw[color=blue] (3.185, 0.314) circle (3pt); % angle = 5.625
                \draw[blue] (3.632, 0.358) node {$x_{1}$};
                \filldraw[color=blue] (0.624, -3.139) circle (3pt); % angle = -78.750
                \draw[blue] (0.712, -3.580) node {$x_{2}$};
                \filldraw[color=blue] (-3.062, -0.929) circle (3pt); % angle = -163.125
                \draw[blue] (-3.493, -1.060) node {$x_{3}$};
                \filldraw[color=blue] (-1.225, 2.956) circle (3pt); % angle = -247.500
                \draw[blue] (-1.397, 3.372) node {$x_{4}$};
                \filldraw[color=blue] (2.822, 1.508) circle (3pt); % angle = -331.875
                \draw[blue] (3.219, 1.721) node {$x_{5}$};
                \filldraw[color=blue] (1.778, -2.661) circle (3pt); % angle = -416.250
                \draw[blue] (2.028, -3.035) node {$x_{6}$};
                \filldraw[color=blue] (-2.474, -2.030) circle (3pt); % angle = -500.625
                \draw[blue] (-2.821, -2.316) node {$x_{7}$};
                \filldraw[color=blue] (-2.263, 2.263) circle (3pt); % angle = -585.000
                \draw[blue] (-2.581, 2.581) node {$x_{8}$};
                \filldraw[color=blue] (2.030, 2.474) circle (3pt); % angle = -669.375
                \draw[blue] (2.316, 2.821) node {$x_{9}$};
                \filldraw[color=blue] (2.661, -1.778) circle (3pt); % angle = -753.750
                \draw[blue] (3.035, -2.028) node {$x_{10}$};
                \filldraw[color=blue] (-1.508, -2.822) circle (3pt); % angle = -838.125
                \draw[blue] (-1.721, -3.219) node {$x_{11}$};
                \filldraw[color=blue] (-2.956, 1.225) circle (3pt); % angle = -922.500
                \draw[blue] (-3.372, 1.397) node {$x_{12}$};
                \filldraw[color=blue] (0.929, 3.062) circle (3pt); % angle = -1006.875
                \draw[blue] (1.060, 3.493) node {$x_{13}$};
                \filldraw[color=blue] (3.139, -0.624) circle (3pt); % angle = -1091.250
                \draw[blue] (3.580, -0.712) node {$x_{14}$};
                \filldraw[color=blue] (-0.314, -3.185) circle (3pt); % angle = -1175.625
                \draw[blue] (-0.358, -3.632) node {$x_{15}$};
                \filldraw[color=blue] (-3.200, -0.000) circle (3pt); % angle = -1260.000
                \draw[blue] (-3.650, -0.000) node {$x_{16}$};
                \filldraw[color=blue] (-0.314, 3.185) circle (3pt); % angle = -1344.375
                \draw[blue] (-0.558, 3.632) node {$x_{17}$};
            \end{scope}
    
            \begin{scope}[xshift=9.3cm]
                \node at (4.5, 0) {};
                
                \draw (0, 0) circle (3.2);
                
                \draw[blue] (3.471, 1.128) node {$x_{0}$};
                \filldraw[color=blue] (1.362, -2.895) circle (3pt); % angle = -64.800
                \draw[blue] (1.554, -3.303) node {$x_{1}$};
                \filldraw[color=blue] (-2.723, -1.681) circle (3pt); % angle = -148.320
                \draw[blue] (-3.106, -1.917) node {$x_{2}$};
                \filldraw[color=blue] (-1.949, 2.538) circle (3pt); % angle = -232.488
                \draw[blue] (-2.223, 2.895) node {$x_{3}$};
                \filldraw[color=blue] (2.349, 2.173) circle (3pt); % angle = 42.761
                \draw[blue] (2.680, 2.478) node {$x_{4}$};
                \filldraw[color=blue] (2.359, -2.163) circle (3pt); % angle = -42.515
                \draw[blue] (2.690, -2.467) node {$x_{5}$};
                \filldraw[color=blue] (-1.982, -2.513) circle (3pt); % angle = -128.264
                \draw[blue] (-2.260, -2.866) node {$x_{6}$};
                \filldraw[color=blue] (-2.639, 1.810) circle (3pt); % angle = -214.437
                \draw[blue] (-3.010, 2.064) node {$x_{7}$};
                \filldraw[color=blue] (1.648, 2.743) circle (3pt); % angle = 59.006
                \draw[blue] (1.880, 3.129) node {$x_{8}$};
                \filldraw[color=blue] (2.828, -1.497) circle (3pt); % angle = -27.894
                \draw[blue] (3.226, -1.708) node {$x_{9}$};
                \filldraw[color=blue] (-1.358, -2.898) circle (3pt); % angle = -115.105
                \draw[blue] (-1.649, -3.305) node {$x_{10}$};
                \filldraw[color=blue] (-2.954, 1.229) circle (3pt); % angle = -202.594
                \draw[blue] (-3.470, 1.402) node {$x_{11}$};
                \filldraw[color=blue] (1.112, 3.001) circle (3pt); % angle = 69.665
                \draw[blue] (1.268, 3.423) node {$x_{12}$};
                \filldraw[color=blue] (3.038, -1.005) circle (3pt); % angle = -18.301
                \draw[blue] (3.565, -1.146) node {$x_{13}$};
                \filldraw[color=blue] (-0.907, -3.069) circle (3pt); % angle = -106.471
                \draw[blue] (-1.035, -3.500) node {$x_{14}$};
                \filldraw[color=blue] (-3.093, 0.819) circle (3pt); % angle = -194.824
                \draw[blue] (-3.629, 0.934) node {$x_{15}$};
                \filldraw[color=blue] (0.738, 3.114) circle (3pt); % angle = 76.658
                \draw[blue] (0.542, 3.551) node {$x_{16}$};
                \filldraw[color=blue] (3.130, -0.666) circle (3pt); % angle = -12.008
                \draw[blue] (3.670, -0.759) node {$x_{17}$};
                \filldraw[color=blue] (-0.600, -3.143) circle (3pt); % angle = -100.807
                \draw[blue] (-0.384, -3.585) node {$x_{18}$};
                \filldraw[color=blue] (-3.154, 0.541) circle (3pt); % angle = -189.726
                \draw[blue] (-3.698, 0.567) node {$x_{19}$};
                \filldraw[color=blue] (0.487, 3.163) circle (3pt); % angle = 81.246
                \filldraw[color=blue] (3.170, -0.439) circle (3pt); % angle = -7.878
                \filldraw[color=blue] (-0.395, -3.176) circle (3pt); % angle = -97.090
                \filldraw[color=blue] (-3.180, 0.356) circle (3pt); % angle = -186.381
                \filldraw[color=blue] (0.320, 3.184) circle (3pt); % angle = 84.257
                \filldraw[color=blue] (3.187, -0.288) circle (3pt); % angle = -5.169
                \filldraw[color=blue] (-0.260, -3.189) circle (3pt); % angle = -94.652
                \filldraw[color=blue] (-3.191, 0.234) circle (3pt); % angle = -184.187
                \filldraw[color=blue] (0.210, 3.193) circle (3pt); % angle = 86.232
                \filldraw[color=blue] (3.194, -0.189) circle (3pt); % angle = -3.391
                \filldraw[color=blue] (-0.170, -3.195) circle (3pt); % angle = -93.052
                \filldraw[color=blue] (-3.196, 0.153) circle (3pt); % angle = -182.747
                \filldraw[color=blue] (0.138, 3.197) circle (3pt); % angle = 87.528
                \filldraw[color=blue] (3.198, -0.124) circle (3pt); % angle = -2.225
                \filldraw[color=blue] (-0.112, -3.198) circle (3pt); % angle = -92.003
                \filldraw[color=blue] (-3.198, 0.101) circle (3pt); % angle = -181.802
                \filldraw[color=blue] (0.091, 3.199) circle (3pt); % angle = 88.378
                \filldraw[color=blue] (3.199, -0.082) circle (3pt); % angle = -1.460
                \filldraw[color=blue] (-0.073, -3.199) circle (3pt); % angle = -91.314
                \filldraw[color=blue] (-3.199, 0.066) circle (3pt); % angle = -181.182
                \filldraw[color=blue] (0.059, 3.199) circle (3pt); % angle = 88.936
                \filldraw[color=blue] (3.200, -0.053) circle (3pt); % angle = -0.958
                \filldraw[color=blue] (-0.048, -3.200) circle (3pt); % angle = -90.862
                \filldraw[color=blue] (-3.200, 0.043) circle (3pt); % angle = -180.776
                \filldraw[color=blue] (0.039, 3.200) circle (3pt); % angle = 89.302
                \filldraw[color=blue] (3.200, -0.035) circle (3pt); % angle = -0.628
                \filldraw[color=blue] (-0.032, -3.200) circle (3pt); % angle = -90.566
                \filldraw[color=blue] (-3.200, 0.028) circle (3pt); % angle = -180.509
                \filldraw[color=blue] (0.026, 3.200) circle (3pt); % angle = 89.542
                \filldraw[color=blue] (3.200, -0.023) circle (3pt); % angle = -0.412
                \filldraw[color=blue] (-0.021, -3.200) circle (3pt); % angle = -90.371
                \filldraw[color=blue] (-3.200, 0.019) circle (3pt); % angle = -180.334
                \filldraw[color=blue] (0.017, 3.200) circle (3pt); % angle = 89.699
                \filldraw[color=blue] (3.200, -0.015) circle (3pt); % angle = -0.271
                \filldraw[color=blue] (-0.014, -3.200) circle (3pt); % angle = -90.243
                \filldraw[color=blue] (-3.200, 0.012) circle (3pt); % angle = -180.219
                \filldraw[color=blue] (0.011, 3.200) circle (3pt); % angle = 89.803
                \filldraw[color=blue] (3.200, -0.010) circle (3pt); % angle = -0.177
                \filldraw[color=blue] (-0.009, -3.200) circle (3pt); % angle = -90.160
                \filldraw[color=blue] (-3.200, 0.008) circle (3pt); % angle = -180.144
            
                \filldraw[color=black] (0.000, 3.200) circle (3pt); % angle = 90.000
                \draw[black] (0.000, 2.750) node {$y_0$};
                \filldraw[color=black] (3.200, 0.000) circle (3pt); % angle = 0.000
                \draw[black] (2.750, 0.000) node {$y_1$};
                \filldraw[color=black] (0.000, -3.200) circle (3pt); % angle = -90.000
                \draw[black] (0.000, -2.750) node {$y_2$};
                \filldraw[color=black] (-3.200, -0.000) circle (3pt); % angle = -180.000
                \draw[black] (-2.750, -0.000) node {$y_3$};
                \filldraw[color=blue] (3.043, 0.989) circle (3pt); % angle = 18.000
    
            \end{scope}
        \end{tikzpicture}
        
        \caption{Left, $\Gsc_1$ with $k = 4$ is used to generate a greedy sequence, where $x_{17}$ shows, that we are optimal. 
        Right, $\Gsc_2$ with $k = 4$ and $\alpha = 0.1$ is used. Starting at $x_0 = 0.2$, we never see an optimal solution.}
        \label{fig:two_greedy_functions}
    \end{figure}
    
    Combinatorially, these two functions are identical (until $\varepsilon^{-1}$ rounds have passed), so we have no guarantee, that the algorithm terminates with a solution in polynomial time.
\end{example}

%-----------------------------------------------------------------------
\subsection{Polygons with holes}
\label{apx:polygons_with_holes}
%-----------------------------------------------------------------------
One natural generalization of the \cagp is the \cagp with holes. In this section, we study the variant in which the goal is to guard the external boundary of the polygon. 

Many parts of the geometric analysis break down when we introduce holes. The most glaring is the fact that the entire strategy of Proposition~\ref{prop:guards_below_pivots_implies_other_guards_above_pivots} does not work, as we are using the fact that the pivot points (and edges close to them) have to be guarded by a guard at some point and now these pivot points could be on a hole.

Furthermore, the algorithm given for $\Greedy$ in Section~\ref{sec:alg_greedy_interval} does not work as a point in $P$ can now see two vertices of $P$ without seeing the entire edge between them. This issue can be fixed, but even if it is, we show that \RepeatedGreedy can run in superpolynomial time in the number of vertices when $P$ has holes.

\begin{example}[Octagon with hole]\label{exa:octagon_with_holes}

Consider a regular octagon with a rotated regular octagon inside like shown on Figure~\ref{fig:octagon_in_octagon}, where the inner octagons edges line up exactly with the dashed line segments. Placing guards $g_1$ and $g_2$ will guard the entire boundary contiguously. 

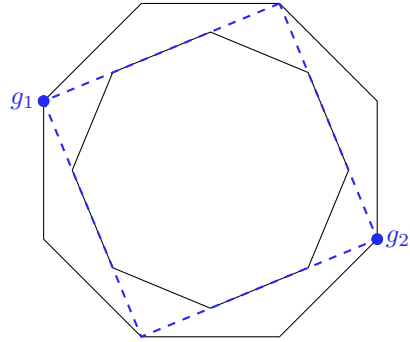
\begin{figure}[ht]
    \centering
    \begin{tikzpicture}[scale = 0.8]
        \foreach \x in {0,1,...,7}{
            \coordinate (A\x) at (45 * \x + 22.5: 3);
            \coordinate (B\x) at (45 * \x : 3 / 1.3065);
        }
        \foreach \L in {A, B}{
        \draw (\L0) -- (\L1) -- (\L2) -- (\L3) -- (\L4) -- (\L5) -- (\L6) -- (\L7) -- cycle;
        }

        \draw[dashed, guard, thick] (A1) -- (A3) -- (A5) -- (A7) -- cycle;

        \foreach \p in {A3, A7}{
        \node[draw, fill = guard, circle, inner sep = 0, minimum size = 4pt, color = guard] at (\p) {};
        }
        
        \newcommand{\drawNode}[3]
        {
            \node[draw, fill = guard, circle, inner sep = 0, minimum size = 4pt, color = guard, nearly opaque, label={[centered, color = guard, label distance = 0.2cm]#3:{$#2$}}] at (#1) {};
        }

        \drawNode{A3}{g_1}{west}
        \drawNode{A7}{g_2}{east}
    \end{tikzpicture}
    \caption{$P$ is an octagon with an octagonal hole in the center. Guards at $g_1$ and $g_2$ will guard the entire boundary contiguously.}\label{fig:octagon_in_octagon}
\end{figure}

Now we enlarge the inner polygon by a tiny $\varepsilon$, which will make this solution invalid. When we run \RepeatedGreedy in this slightly different polygon, it is evident, that the best guard is placed on the outer boundary of $P$ as far as the starting point can see (see Figure~\ref{fig:greedy_sequence_on_octagon}) and the greedy interval is found by taking the furthest point which this guard can see.

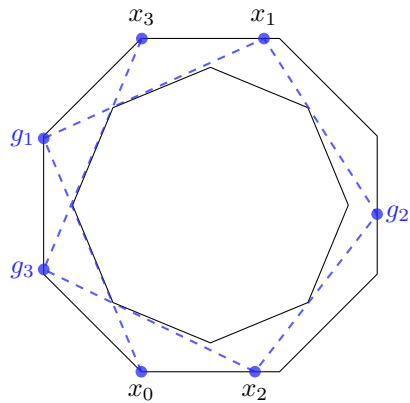
\begin{figure}[ht]
    \centering
    \begin{tikzpicture}[scale = 0.8]
        \foreach \x in {0,1,...,7}{
            \coordinate (A\x) at (45 * \x + 22.5: 3);
            \coordinate (B\x) at (45 * \x : 3 / 1.306562964 - 0.005);
        }
        \foreach \L in {A, B}{
            \draw (\L0) -- (\L1) -- (\L2) -- (\L3) -- (\L4) -- (\L5) -- (\L6) -- (\L7) -- cycle;}

        \coordinate (G1) at (-2.77,1.11);
        \coordinate (x1) at (0.89, 2.77);
        \coordinate (G2) at (2.77, -0.15);
        \coordinate (x2) at (0.74, -2.77);
        \coordinate (G3) at (-2.77, -1.07);
        \coordinate (x3) at (-1.14, 2.77);

        \newcommand{\drawNode}[4]
        {
            \node[draw, fill = guard, circle, inner sep = 0, minimum size = 4pt, color = guard, nearly opaque, label={[color = #4,centered, label distance = 0.2cm]#3:{$#2$}}] at (#1) {};
        }
        
        \drawNode{G1}{g_1}{west}{guard}
        \drawNode{G2}{g_2}{east}{guard}
        \drawNode{G3}{g_3}{west}{guard}
        \drawNode{A5}{x_0}{south}{black}
        \drawNode{x1}{x_1}{north}{black}
        \drawNode{x2}{x_2}{south}{black}
        \drawNode{x3}{x_3}{north}{black}

        \draw[dashed, guard, nearly opaque, thick] (A5) -- (G1) -- (x1) -- (G2) -- (x2) -- (G3) -- (x3) ;
    
    \end{tikzpicture}
    \caption{We have here enlarged the inner polygon with $\varepsilon = 0.001$, we now start \RepeatedGreedy at $x_0$. The optimal guard for $x_0$ is found by finding the furthest point along the outer boundary which $x_0$ can see. This point is denoted $g_1$ which can see until $x_1$ and so on.} \label{fig:greedy_sequence_on_octagon}
\end{figure}

As long as the found endpoint/guard is closer to the next vertex than the previous vertex (along the outer boundary of $P$), it is the same vertex of the inner polygon which will block the view to the next guard/endpoint (the one marked in Figure~\ref{fig:octagon_with_coordinates}). 

We consider how far a point on the lower edge of $P$ can see on the left vertical edge of $P$. To do this, we embed $P$ into a coordinate system with the lower edge having endpoints in $(0,0)$ and $(1,0)$. The relevant coordinates are drawn on Figure~\ref{fig:octagon_with_coordinates}.

\begin{figure}[ht]
    \centering
    \begin{tikzpicture}[scale = 2.5]
        \coordinate (A0) at (-0.7071, 1.7071);
        \coordinate (A1) at (-0.7071, 0.7071);
        \coordinate (A2) at (0,0);
        \coordinate (A3) at (1,0);
        \coordinate (Am1) at ($(A0) + (0.2, 0.2)$);
        \coordinate (A4) at ($(A3) + (0.2, 0.2)$);

        \coordinate (B0) at (-0.5 , 0.5 + 0.7071);
        \coordinate (B1) at (0.5 - 0.7071, 0.5);
        \coordinate (B2) at (0.5, 0.7071 - 0.5);

        \coordinate (a) at (0.3, 0);
        \coordinate (b) at (-0.7071, 1);

        \draw (Am1) -- (A0) -- (A1) -- (A2) -- (A3) -- (A4);
        \draw (B0) -- (B1) -- (B2);

        \draw[dashed, gray] (a) -- (b);
        
        \node[draw, fill = black, circle, inner sep = 0, minimum size = 4pt, label = {[centered, label distance = 0.2cm]south:{$(\delta, 0)$}}] at (a) {};
        \node[label = {[centered, label distance = 0.2cm]south:{$(0, 0)$}}] at (A2) {};
        \node[label = {[centered, label distance = 0.2cm]south:{$(1, 0)$}}] at (A3) {};
        \node[label = {west:{$\left(-\sqrt{1/2}, \sqrt{1/2}\right)$}}] at (A1) {};
        \node[label = {west:{$\left(-\sqrt{1/2}, 1 + \sqrt{1/2}\right)$}}] at (A0) {};

        \node[draw, fill = black, circle, inner sep = 0, minimum size = 4pt, label = {north east:
        {$\left(\frac{-\sqrt{2} + 1}{2} - \varepsilon, \frac{1}{2} - \varepsilon \right)$}
        }] at (B1) {};
        
        \node[draw, fill = black, circle, inner sep = 0, minimum size = 4pt, label = {west:
        {$\left(\sqrt{1/2}, 1 +  \sqrt{1/2} - \gamma \right)$}
        }] at (b) {};

        \draw[black, decorate, decoration={brace, amplitude=1ex, raise=0.5ex}] 
        ($(A0)$) -- node[right, xshift = 0.15cm] {$\gamma$} ($(b)$);
        
        \draw[black, decorate, decoration={brace, amplitude=1ex, raise=0.5ex}] 
        ($(A2)$) -- node[above, yshift = 0.15cm] {$\delta$} ($(a)$);

    \end{tikzpicture}
    \caption{$P$ embedded into a coordinate system with coordinates marked. A guard placed in $(\delta, 0)$ can see no longer than $(-\sqrt{1/2}, 1 + \sqrt{1/2} - \gamma)$. Here $\gamma$ is the distance between the new point and the next vertex of $P$.}\label{fig:octagon_with_coordinates}
\end{figure}
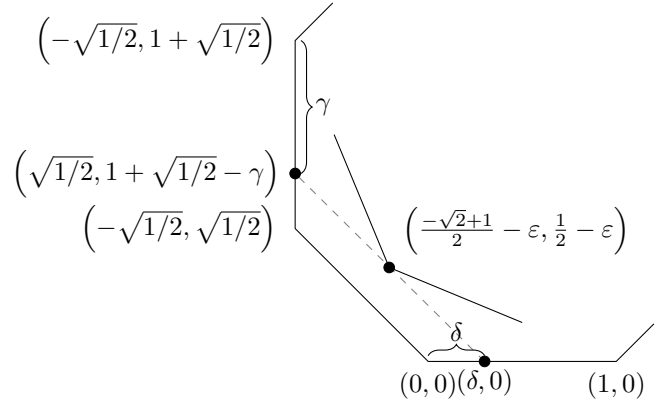

Calculating, we get that the equation of the dashed line is:

\begin{align*}
 y = \left(\frac{\sqrt{2}+2\delta}{\sqrt{2}-1 +2\varepsilon + 2\delta}-1\right)(\delta - x) 
\end{align*}

Inserting $x = -\sqrt{1/2}$ in this equation will yield the y-coordinate, $y'$, of the intersection with the dashed line and the left vertical edge:

\begin{align*}
    y' &= \left(\frac{\sqrt{2}+2\delta}{\sqrt{2}-1 +2\varepsilon + 2\delta}-1\right)\left(\delta + \sqrt{\frac12}\right)\\
    &= \frac{1 + 2\sqrt{2}\delta + 2\delta^2}{\sqrt{2}-1 +2\varepsilon + 2\delta} - \sqrt{\frac{1}{2}} - \delta
\end{align*}

And now $\gamma$ is calculated as $\gamma = 1 + \sqrt{\frac{1}{2}} - y'$:
\begin{align*}
    \gamma &= \delta + \frac{2\delta - 2\delta^2 + (2 + 2\sqrt{2})\varepsilon}{\sqrt{2}-1 +2\varepsilon + 2\delta}\\
    &\in (\delta, 7\delta + 12\varepsilon)
\end{align*}

We now chooce $\varepsilon = \frac1{2 \cdot 12 \cdot 8 ^ {2N}}$ where $N$ is some large number. Consider the sequence $(\gamma_i)_{i=0}^{2N}$ of distances from the furthest visible point and the next vertex when taking such steps. Note that in Figure~\ref{fig:octagon_with_coordinates} we have, say, $\delta=\gamma_i$ and $\gamma = \gamma_{i+1}$, i.e. the $\delta$ is the member of the of the $\gamma_i$ sequence that precedes $\gamma$. The above bound then becomes $\gamma_{i+1} \in (\gamma_{i},7\gamma_{i}+12\varepsilon)$ and especially $\gamma_i < \gamma_{i+1}$ hence the associated local greedy sequences will have negative fingerprints and no repetitions.

Furthermore we show by induction that $\gamma_i < \frac1{2 \cdot 8^{2N - i}}$ as $\gamma_0 = 0$ and if $\gamma_i < \frac1{2 \cdot 8^{2N - i}}$, we get:
\begin{align*}
    \gamma_{i+1} &< 7 \gamma_i + 12 \varepsilon\\
    &= \frac7{2\cdot 8^{2N-i}} + \frac1{2\cdot 8^{2N}}\\
    &\leq \frac1{2\cdot 8^{2N-(i+1)}}
\end{align*}
And we thus have $\gamma_i \in [0, 1)$ for all $\gamma = 0,1,\dots, 2N$, thus for all the guards will be on the same two edges, i.e. we get no edge jumps. Thus in the first $N$ greedy steps we do not reach a geometric progress condition. 

Since $N$ is independent of $n$ the runtime of \RepeatedGreedy is unbounded in the real RAM model.
    
\end{example}

\end{document}